\documentclass[aps,prx,reprint,10pt,twocolumn,superscriptaddress,floatfix,nofootinbib,showpacs,longbibliography]{revtex4-2}
\usepackage{makecell}
\usepackage{array}
\newcolumntype{C}[1]{>{\centering\arraybackslash}m{#1}}
\usepackage{booktabs,adjustbox}
\usepackage{dsfont}
\usepackage{amsmath, enumerate}
\usepackage[normalem]{ulem}
\usepackage[utf8]{inputenc}  
\usepackage[T1]{fontenc}     
\usepackage[table]{xcolor}
\usepackage{color,soul}
\usepackage[british]{babel}  
\usepackage{libertine}
\usepackage[libertine]{newtxmath}
\usepackage[scaled=0.86]{berasans}  
\usepackage[colorlinks=true, citecolor=blue, urlcolor=blue]{hyperref}  
\usepackage{graphicx} 
\usepackage[babel]{microtype}  

\usepackage{amsmath,amssymb,amsthm,amsfonts,mathrsfs,bbm} 

\usepackage{xspace}  
\usepackage{pgf,tikz}
\usepackage{xcolor}
\usepackage{multirow}
\usepackage{array}
\usepackage{bigstrut}
\usepackage{braket}
\usepackage{color}
\usepackage{natbib}
\usepackage{multirow}
\usepackage{enumerate}
\usepackage{mathtools}
\usepackage{blkarray}
\usepackage[caption = false]{subfig}
\usepackage{xcolor,colortbl}
\usepackage{color}
\usepackage{rotating}
\usepackage{tikz}
\usepackage{tabularx}
\usepackage{comment}
\usepackage[shortlabels]{enumitem}
\usepackage{cancel}
\usepackage{mathrsfs}   
\usepackage{accents}

\newcommand{\Tr}{\operatorname{Tr}}

\newcommand{\be}{\begin{equation}}
\newcommand{\ee}{\end{equation}}
\newcommand{\ba}{\begin{eqnarray}}
\newcommand{\ea}{\end{eqnarray}}
\newcommand{\ketbra}[2]{|#1\rangle \langle #2|}
\newcommand{\tr}{\operatorname{Tr}}

\newcommand{\Lcal}{\mathcal{L}}
\newcommand{\Pcal}{\mathcal{P}}
\newcommand{\Hcal}{\mathcal{H}}

\definecolor{axiomatic}{HTML}{FCF3E5}
\definecolor{microscopic}{HTML}{e9eef1}
\definecolor{operational}{HTML}{fceeee}
\definecolor{theorem}{HTML}{E1F1EC}

\usepackage{tcolorbox}
\tcbuselibrary{breakable}
\newtcolorbox[]{defbox}[1][]{breakable,colback=RoyalBlue!10!white,colframe=
white,title=#1}

\newtcolorbox[]{standardbox}[1][]{breakable,colback=RoyalBlue!5!white,colframe=white,#1}
\newtcolorbox[]{axiomaticbox}[1][]{breakable,colback=axiomatic,colframe=white, #1}
\newtcolorbox[]{microscopicbox}[1][]{colback=microscopic,colframe=white,#1}
\newtcolorbox[]{operationalbox}[1][]{breakable,colback=operational,colframe=white,#1}
\newtcolorbox[]{theorembox}[1][]{breakable,colback=theorem,colframe=white,#1}

\newtheorem{theorem}{Theorem}
\newtheorem{prop}{Proposition}
\newtheorem{corollary}{Corollary}
\newtheorem{definition}{Definition}

\newtheorem{example}{Example}
\newtheorem{remark}{Remark}
\newtheorem{lemma}{Lemma}

\newtheorem{principle}{Principle}

\newcommand{\ident}{\mathbbm{1}}

\def\>{\rangle}
\def\<{\langle}

\def\bbra#1{\mathinner{\langle \! \langle{#1}|}}

\def\kket#1{\mathinner{|{#1}\rangle \! \rangle}}

\def\kketbra#1#2{\kket{#1\vphantom{#2}}\!\bbra{#2\vphantom{#1}}}

\newcommand{\Ical}{\mathcal{I}}

\usepackage{centernot}
\usepackage{subfig}

\begin{document}

\title{Higher-order quantum thermodynamics: equilibrium and causal structure}
\author{Simon Milz}
\affiliation{Institute of Photonics and Quantum Sciences, School of Engineering and Physical Sciences, Heriot-Watt University, Edinburgh EH14 4AS, United Kingdom}
\affiliation{School of Physics, Trinity College Dublin, Dublin 2, Ireland}

\author{Kyrylo Simonov}
\affiliation{Fakult\"{a}t f\"{u}r Mathematik, Universit\"{a}t Wien, Oskar-Morgenstern-Platz 1, 1090 Vienna, Austria}

\author{Zolt\'an Zimbor\'as}
\affiliation{Department of Physics, University of Helsinki, Yliopistonkatu 4 00100 Helsinki, Finland}
\affiliation{HUN-REN Wigner Research Centre for Physics, Budapest, Hungary}
\affiliation{Algorithmiq Ltd, Kanavakatu 3C 00160 Helsinki, Finland}

\author{Tamal Guha}
\affiliation{Mathematical Institute, Slovak Academy of Sciences, \v{S}tef\'anikova 49, 814 73 Bratislava, Slovakia}

\author{Saptarshi Roy}
\affiliation{Center for Quantum Engineering, Research, and Education, TCG CREST, Bidhan Nagar, Kolkata - 700091, India}

\author{Giulio Chiribella}
\affiliation{QICI Quantum Information and Computation Initiative, School of Computing and Data Science, The University of Hong Kong, Pokfulam Road, Hong Kong}
\affiliation{Quantum Group, Department of Computer Science, University of Oxford, Wolfson Building, Parks Road, Oxford, OX1 3QD, United Kingdom}
\affiliation{Perimeter Institute for Theoretical Physics, 31 Caroline Street North, Waterloo, Ontario, Canada}

\date{\today}
\begin{abstract}
Quantum thermodynamics is traditionally formulated as a theory of equilibrium states and state transformations. Recent advances in higher-order quantum transformations, which describe physical scenarios beyond states and channels and provide a systematic framework for causal order, raise the question of how equilibrium should be defined and preserved in this more general setting. Starting from the Gibbs state as the unique equilibrium state, we identify \textit{equilibrium preservation} as its natural higher-order extension. We show that, while this principle can in general give rise to distinct classes of transformations, all such distinctions disappear when equilibrium preservation is required \textit{completely}, namely under arbitrary ancillary extensions. Remarkably, every transformation satisfying this condition is causally ordered, making causal order an emergent consequence of thermodynamic equilibrium. We establish this result for arbitrary higher-order maps and use the resulting framework to introduce free-energy-like quantities for quantum channels. Our findings reveal a fundamental connection between thermodynamic equilibrium and causal structure, providing a foundation for a fully fledged theory of higher-order quantum thermodynamics.
\end{abstract}

\maketitle

\section{Introduction}
Thermodynamics is among the most successful physical theories, applying across scales from microscopic systems to macroscopic phenomena~\cite{haddad_thermodynamics_2017, binder_thermodynamics_2018, beretta_universal_2026}. Originally formulated as a phenomenological theory of heat, work, and temperature, it has since been extended to the quantum regime, revealing deep connections between information and thermodynamics~\cite{szilard_uber_1929, landauer_1961, bennett_thermodynamics_1982, kim_quantum_2011, parrondo_thermodynamics_2015, vinjanampathy_quantum_2016, goold_role_2016}, establishing fundamental constraints on quantum thermodynamic processes~\cite{horodecki_fundamental_2013, aberg_truly_2013, reeb_improved_2014, brandao_second_2015, lostaglio_description_2015, bera_generalized_2017, lostaglio_introductory_2019, luo_trinity_2026}, and providing microscopic explanations for the thermodynamic arrow of time~\cite{jarzynski_nonequilibrium_1997, crooks_entropy_1999, PhysRevE.75.050102, aberg_fully_2018, seif_machine_2021}. 

A cornerstone of thermodynamics is the notion of thermodynamical equilibrium. For quantum \textit{states}, thermodynamical equilibrium is uniquely characterized by the Gibbs state, which emerges from diverse physical and information-theoretic principles, including weak system-bath coupling, constrained entropy maximization, and complete passivity~\cite{BreuerOpenQS, landau_statistical_2011, jaynes_information_1957, pusz_passive_1978, lenard_thermodynamical_1978}. Deviations from equilibrium constitute a thermodynamic resource~\cite{brandao_resource_2013, brandao_second_2015, aberg_truly_2013, lostaglio_description_2015, lostaglio_quantum_2015, gour_role_2022}, whose availability affects the accuracy of information processing at the fundamental level \cite{chiribella2022nonequilibrium}. For quantum \textit{processes}, the situation is more complex. Several inequivalent classes of thermodynamically relevant processes have been proposed, including thermal operations~\cite{Janzing2000, horodecki_fundamental_2013, brandao_resource_2013, brandao_second_2015}, constrained thermal operations~\cite{lostaglio_elementary_2018, czartowski_thermal_2023}, Gibbs-preserving maps~\cite{faist_gibbs-preserving_2015, e19060241, muller_correlating_2018, Shiraishi2021, tajima_gibbs-preserving_2025}, and dynamical thermalization models~\cite{davies_markovian_1974, lostaglio_continuous_2022}. Despite their differences, all these classes of processes share one basic feature: they preserve the Gibbs state. This observation motivates Gibbs-preserving transformations as the broadest notion of equilibrium processes, requiring only that equilibrium states remain invariant.  Similar to the case of quantum states, \textit{any} transformation that does not preserve equilibrium must naturally constitute a thermodynamic resource in its own right. Consequently, in recent years, the resource content of quantum processes has seen increased attention~\cite{wang_resource_2019, li_quantifying_2020, liu_operational_2020, takagi_application_2020, gour_dynamical_2020, gour_dynamical_res_2020, gour_entanglement_2021} leading to extensions of fundamental thermodynamic concepts like entropy and free energy from quantum states to transformations~\cite{gour_entropy_2021, gour_inevitable_2025, badhani_thermodynamic_2025, badhani_thermodynamics_2025a}. 

The question of thermodynamic equilibrium becomes particularly nontrivial beyond the level of states and state transformations. Once quantum processes themselves are regarded as physical resources, it is natural to consider transformations acting on processes, and more generally the hierarchy of \textit{higher-order} quantum transformations~\cite{Bisio2011, taranto_higher-order_2025}. These objects arise naturally in quantum circuit architectures~\cite{chiribella_transforming_2008, chiribella_quantum_2008, Chiribella2009}, quantum communication~\cite{chiribella_quantum_2019, Taddei2019, Kristjansson2020, Milz2022}, and non-Markovian quantum dynamics~\cite{kretschmann_quantum_2005, pollock_operational_2018, pollock_non-markovian_2018, berk_resource_2021}, and provide the natural framework for describing complex quantum processes involving preparations, transformations, and measurements at multiple times. As a theory that applies universally, thermodynamics should therefore extend to such general scenarios, raising the basic question of which higher-order transformations ought to be regarded as \textit{in equilibrium}, and what fundamental constraints equilibrium imposes on them.

\begin{figure}
    \centering
    \includegraphics[width = 0.9\linewidth]{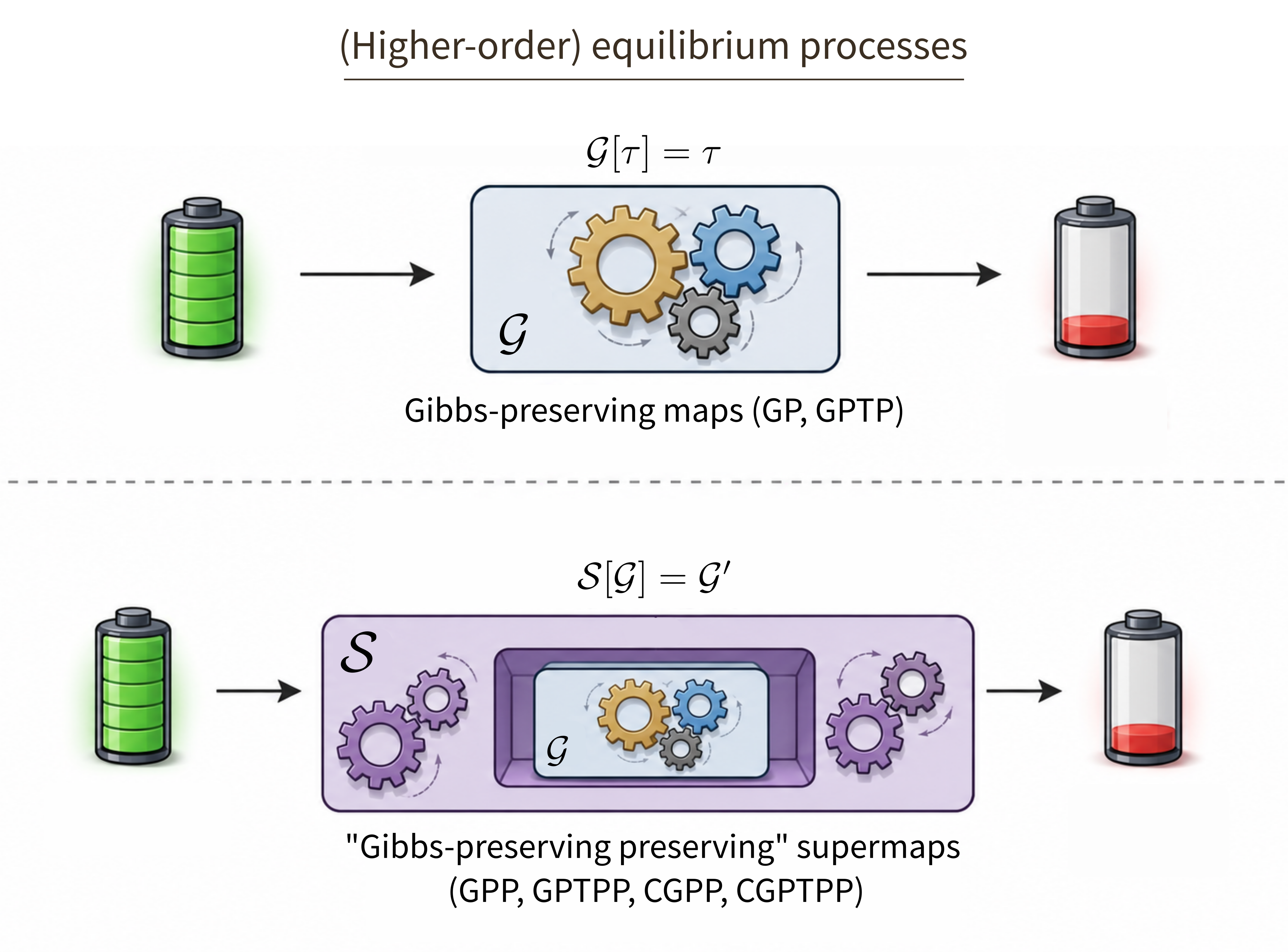}
    \caption{\textbf{Equilibrium processes.} (top) To transform quantum \textit{states}, an experimenter receives a thermodynamic system in a state $\rho$ (depicted by the charged battery) and manipulates it via a transformation $\mathcal{G}$, depleting free energy (depicted by the discharged battery). We consider such a process "in equilibrium", if it preserves the Gibbs state $\tau$, leading to the set of GP and GPTP transformations (see Sec.~\ref{sec::Gibbs_eq_trans}). (bottom) Similarly, an equilibrium \textit{supermap} $\mathcal{S}$ transforms state transformations, leaving the set of Gibbs-preserving transformations invariant. This requirement leads to several possible inequivalent sets of equilibrium supermaps ($\mathbf{GPP}$, $\mathbf{GPTPP}$, $\mathbf{CGPP}$, $\mathbf{CGPTPP}$, see Sec.~\ref{subsec::eq_supmap}). Remarkably, imposing equilibrium preservation in a "complete" sense singles out a \textit{unique} class of physically viable equilibrium transformations (see Sec.~\ref{subsec:completeSupermaps}).}
    \label{fig::supmap}
\end{figure}

The higher-order setting introduces a second question that has no analogue at the level of states or channels: does thermodynamic equilibrium constrain causal structure? Unlike ordinary quantum channels, higher-order transformations can describe genuine multi-time processes and hence encode nontrivial causal relations. In particular, depending on the physical situation they describe, they may or may not admit an underlying causally ordered realization. While some higher-order transformations, known as quantum combs~\cite{chiribella_transforming_2008, chiribella_quantum_2008, Chiribella2009},  are causally ordered by construction, more general higher-order transformations can exhibit indefinite causal structure~\cite{chiribella2009beyond,oreshkov_quantum_2012,chiribella_quantum_2013,Bisio2019, simmons_higher-order_2022, hoffreumon_projective_2026, milz_characterising_2024, Apadula2024, jencova_structure_2026, Apadula2026}, and have become a central framework for studying causality in quantum theory~\cite{oreshkov_quantum_2012, brukner_bounding_2015, Araujo2015, branciard_simplest_2015, castro-ruiz_dynamics_2018, wechs_definition_2019, costa_indefinite_2026}. 

Prominent examples include the quantum SWITCH~\cite{chiribella2009beyond,chiribella_quantum_2013}, which coherently superposes alternative orders of multiple operations, and the quantum time flip~\cite{Liu2022_Flip}, which superposes alternative input-output directions of a single operation.  Over the years, these causally indefinite transformations have been found to offer a variety of information-theoretic advantages ~\cite{chiribella_quantum_2013, chiribella_perfect_2012, araujo_computational_2014, ebler_enhanced_2018, zhao_quantum_2020, bavaresco_strict_2021, Liu2022_Flip, spencer2025indefinite, simonov2026comp} and to give rise to intriguing thermodynamical features ~\cite{felce_quantum_2020, Guha2020, liu_thermodynamics_2022, simonov2022work, dieguez_thermal_2023, Zhu2023, luo_thermodynamic_2025, simonov2025activation, lisboa2026, xue2026anomalous}. While these  features suggest that the quantum SWITCH represents a valuable thermodynamical resource, the link between thermodynamics and indefinite causal order has  not been substantiated in a systematic way thus far.  

A first natural question is whether thermodynamic equilibrium places restrictions on the causal structure of higher-order transformations. Put differently, is there a fundamental relationship between causal order and the preservation of thermodynamical equilibrium?   In this work, we answer this question affirmatively, showing that  all causally indefinite processes necessarily generate non-equilibrium.  We start by putting forward a  thermodynamic principle for admissible higher-order transformations and show that all transformations satisfying it are necessarily causally ordered. As a consequence, causal order emerges directly from thermodynamic assumptions. In doing so, our work contributes to the broader programme of deriving constraints on higher-order quantum transformations from physical principles, such as dynamical consistency requirements~\cite{castro-ruiz_dynamics_2018} or extendibility to unitary processes~\cite{araujo_purification_2017, barrett_cyclic_2021}. Quite remarkably,  we find out  that \textit{thermodynamic} constraints alone are sufficient to single out causally ordered processes.

Throughout this work, we put forward a general   \textit{principle of equilibrium preservation}: thermodynamic processes should not be able to generate non-equilibrium resources from systems at thermodynamical equilibrium  (see Fig.~\ref{fig::supmap}). Concretely, we require higher-order transformations to preserve all  Gibbs-preserving processes, including both deterministic processes (quantum channels) and probabilistic processes (quantum operations). Building on  this principle, we put forward several candidate classes of higher-order transformations that preserve equilibrium, showing that such classes are generally distinct, and establishing  inclusion relations among them. Remarkably, we find out that  the distinctions among these classes disappear once equilibrium preservation is required to hold \textit{completely,  i.e.,}  even when the higher-order transformations act locally on parts of  multipartite quantum processes.   In particular, complete equilibrium preservation singles out a unique set of admissible processes: causally ordered thermal transformations, namely those that admit a causally ordered realization and do not allow for the extraction of athermality.

This collapse into a single, physically meaningful set under \textit{complete} equilibrium preservation mirrors a recurring theme throughout quantum theory; physical channels must be \textit{completely} positive, higher-order transformations  must preserve complete positivity under arbitrary ancillary extensions~\cite{chiribella_quantum_2008, araujo_purification_2017}, and many constructions of higher-order transformations rely crucially on similar completeness requirements~\cite{burniston_necessary_2020,simmons_higher-order_2022, Milz2022, milz_characterising_2024, hoffreumon_projective_2026}. Similarly, complete equilibrium preservation emerges as the natural requirement for a higher-order theory of quantum thermodynamics. We show that the resulting characterization of causally ordered thermal transformations is not confined to a particular level of the hierarchy of possible processes, but extends to \textit{arbitrary} higher-order transformations.  Overall, this result provides a foundation for a fully fledged theory of higher-order quantum thermodynamics.  As an example of further research along this path,  we provide several  extensions of the notion of free energy from quantum states to quantum channels,  generalizing  recent proposals~\cite{badhani_thermodynamic_2025, badhani_thermodynamics_2025a}. 

In summary, our findings  indicate a systematic path towards  higher-order quantum thermodynamics, and   establish a fundamental relation between causality and the preservation of thermodynamical equilibrium. In this respect, they complement the large body of literature connecting the arrow of time to thermodynamic considerations~\cite{jarzynski_nonequilibrium_1997, crooks_entropy_1999, PhysRevE.75.050102, parrondo_entropy_2009, korzekwa_structure_2017, aberg_fully_2018, seif_machine_2021}, by relating causal order -- rather than the directionality of processes -- with axiomatic thermodynamic principles.

\section{Equilibrium transformations}
\label{sec:preliminaries}

Here we develop the general framework of higher-order transformations preserving thermodynamical equilibrium, providing the main definitions and notation that will be used throughout the paper.

\subsection{State transformations: quantum channels and operations}

A central topic in quantum thermodynamics is the characterization of the state transformations $\rho \mapsto \rho'$ achievable under thermodynamic constraints. Such  transformations are implemented by quantum processes, described by linear maps $\mathcal{T} \in \mathcal{L}(\mathcal{L}(\mathcal H_I), \mathcal{L}(\mathcal{H}_O))$ that map input states $\rho_I \in \mathcal{L}(\mathcal{H}_I)$ onto output states $\rho_O'=\mathcal T[\rho_I] \in \mathcal{L}(\mathcal{H}_O)$, where we used the notation $\mathcal{L}  (X,Y)$ for the set of linear transformation from a vector space $X$ to another vector space $Y$.    More specifically,  a linear map $\mathcal{T}$ representing a valid quantum process  must map valid  quantum states (positive semidefinite trace-1 operators) into valid quantum states, even when acting locally  on a part of a larger composite system. This leads to the notions of positivity, complete positivity, and trace preservation.

\begin{figure}[t!]
    \centering
    \includegraphics[width = 0.9\linewidth]{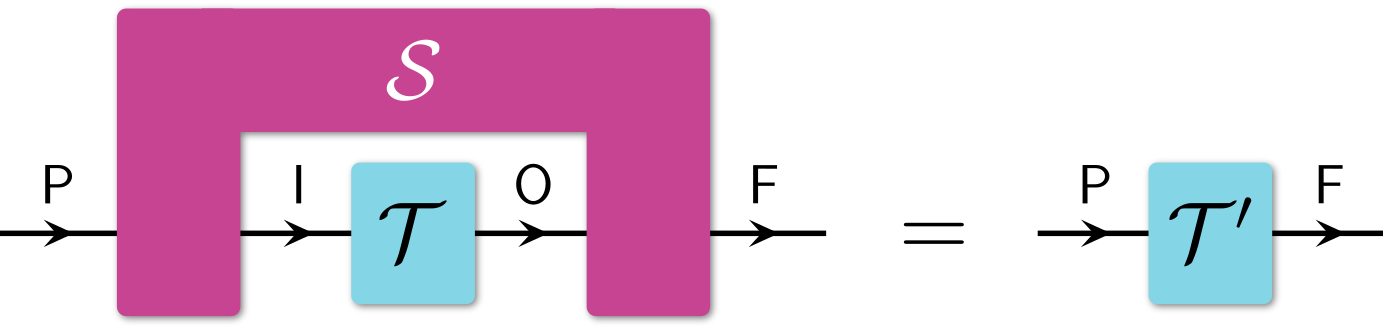}
    \caption{\textbf{The basic type of quantum supermap.} The basic type of quantum  supermap $\mathcal{S}$ (depicted in magenta)  transforms a quantum channel/operation $\mathcal{T}$ with input $I$ and output $O$  (light blue on the left side of the equality sign), into  a new quantum channel/operation $\mathcal{T}'$ with input $P$ and output $F$ (light blue on the right side of the equality sign).}
    \label{fig::one_slot_supmap}
\end{figure}

\begin{definition}[Positive, completely positive, and  trace-preserving maps]
    A linear map $\mathcal{T} \in \mathcal{L}(\mathcal{L}(\mathcal{H}_I), \mathcal{L}(\mathcal{H}_O))$ is
    \begin{enumerate}[1)]
    \item {\em positive} if it maps every positive semidefinite operator $\rho  \in  \mathcal{L}(\mathcal{H}_I)$ into a positive semidefinite operator $\mathcal{T}[\rho]  \in  \mathcal{L}(\mathcal{H}_O)$,   
    \item {\em completely positive (CP)} if the map $  \mathcal{T} \otimes \mathcal{I}_c$  is positive for every finite-dimensional Hilbert space $\mathcal{H}_c$ (here, $\mathcal{I}_c$ denotes the identity map on $\mathcal{L}  (\mathcal{H}_c)$),  
    \item  {\em trace-preserving} if $\tr [\mathcal{T}[\rho]]  =  \tr[\rho]$ for every operator $\rho  \in  \mathcal{L}  (\mathcal{H}_I)$.  
    \end{enumerate}
\end{definition}
In the following, the set of CP maps with input system $I$ and output system $O$ will be denoted  by $\mathsf{CP}(I,O)$.  In general statements, we will sometimes omit  the specification of the input and output systems $I$ and $O$.   

It is easy to prove that a linear map $\mathcal{T}$ transforms normalized quantum states into normalized quantum states (even when acting on part of a composite system) if and only if it is   completely positive and trace-preserving (CPTP).   In the following, the set of completely positive trace-preserving maps with input $I$ and output $O$ will be denoted by $\mathsf{CPTP}(I,O)$.

CPTP maps are also known as {\it quantum channels} \cite{holevo2019quantum, heinosaari2011mathematical} and represent quantum processes that happen with unit probability.
 More generally, quantum processes can take place with non-unit probability, as it happens when quantum measurements are performed. Such processes are associated to measurement outcomes, and are mathematically described by  completely positive {\em trace non-increasing} maps, {\em i.e.} CP maps $\mathcal T$ satisfying the property $\tr[\mathcal{T} [\rho]] \le \tr[\rho]$ for every operator $\rho  \in \mathcal{L}(\mathcal{H}_I)$. Completely positive trace non-increasing maps are known as {\em quantum operations}. 

Before concluding this section, we stress out the crucial role of complete positivity: this property provides the weakest  consistency requirement for a map to represent a transformation of quantum states, applicable locally on a part of a composite system.   In the next section, we will see a similar consistency requirement for  higher-order transformations.

\subsection{Higher-order transformations: quantum supermaps}

\begin{figure}[t!]
    \centering
    \includegraphics[width = 0.9\linewidth]{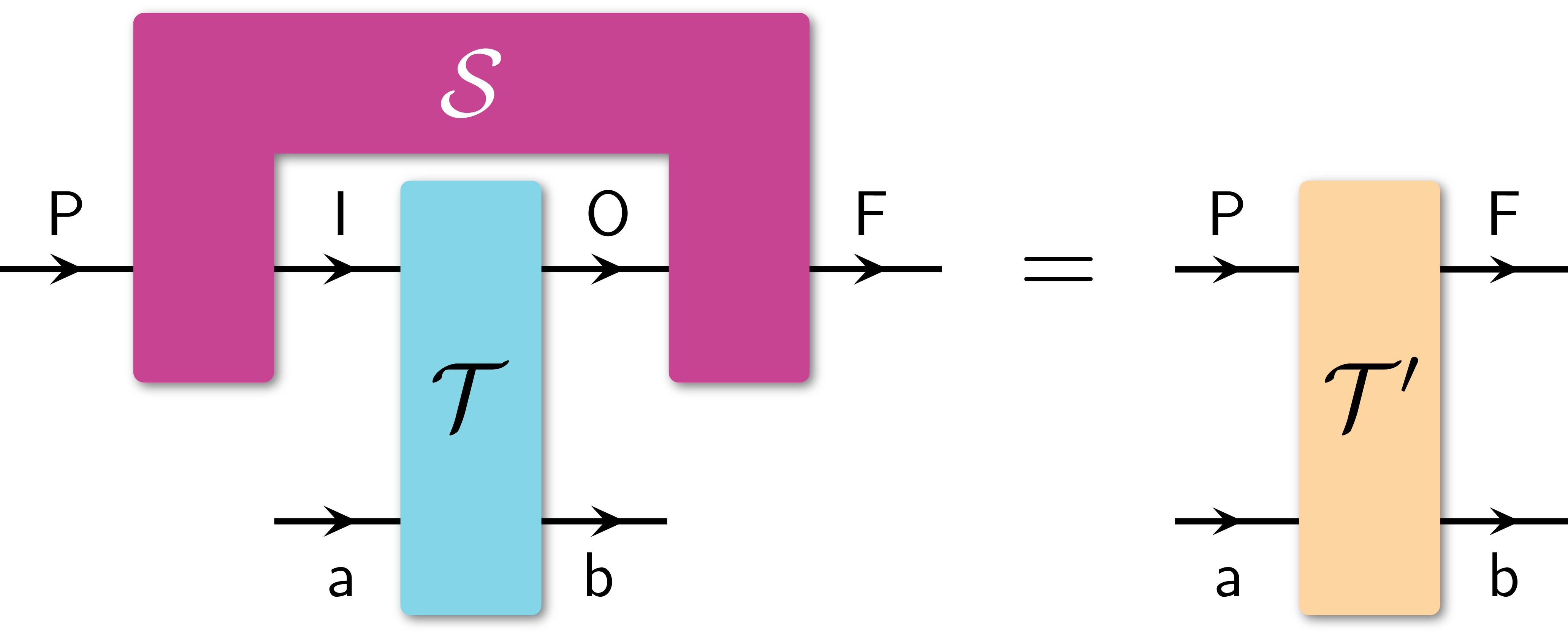}
    \caption{\textbf{CP supermaps.} A  supermap $\mathcal{S}$ is CP if the map $\mathcal{S}\otimes \mathcal I_{a\to b}$ transforms  CP maps $\mathcal{T} \in \textup{CP}(I\otimes a, O \otimes b)$ to CP maps $\mathcal{T}' \in \textup{CP}(P\otimes a, F \otimes b)$ for every possible pair of finite-dimensional systems $a$ and $b$. }
    \label{fig::CP_supmap}
\end{figure}

The defining feature of \textit{higher-order} thermodynamics is that quantum processes  themselves are treated as  resources, which can be subject to higher-order transformations, also known as \emph{quantum supermaps}~\cite{chiribella_transforming_2008,Chiribella2009, chiribella_quantum_2013}.  The basic type of supermap is a linear map   $\mathcal S$  that maps   quantum channels/operations into quantum channels/operations, as illustrated in  Fig. \ref{fig::one_slot_supmap}.   Mathematically, this  linear map  $\mathcal S$ is of the type
\begin{gather}\label{basicsupermap}
    \mathcal{S}: \Lcal(\Lcal(\Hcal_I), \Lcal(\Hcal_O)) \longrightarrow     \Lcal(\Lcal(\Hcal_P), \Lcal(\Hcal_F)) \, ,
\end{gather}
where $\mathcal{H}_P$ is the Hilbert space of the new input system, sometimes referred to as the  (system in the) \emph{global past}, and  $\mathcal{H}_F$ is the Hilbert space of the new output system, sometimes referred to as the (system in the) \emph{global future}, for reasons that will become clear later in this section.

\begin{figure}[t!]
 \centering
 \includegraphics[width=0.98\linewidth]{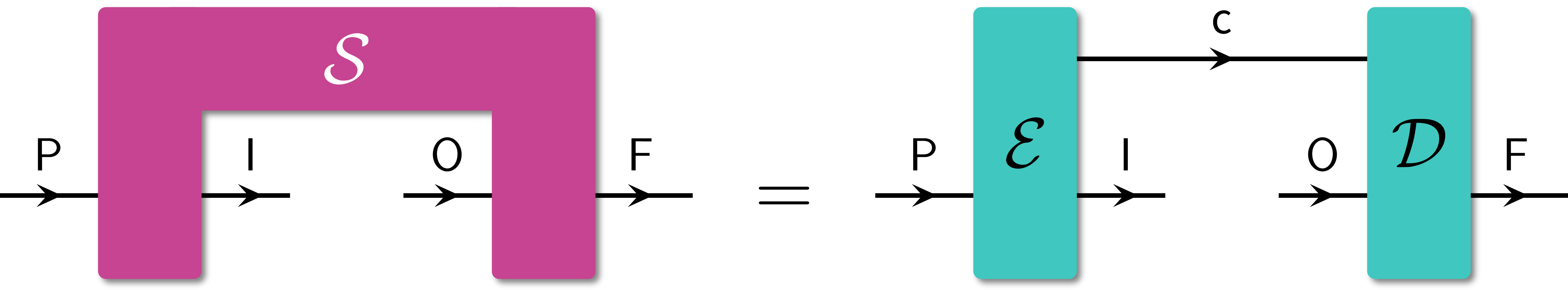}
 \caption{\textbf{Structure of superchannels.} Any proper supermap $\mathcal{S}: \mathsf{CPTP}(I,O) \rightarrow \mathsf{CPTP}(P,F)$ can be obtained via the concatenation of two CPTP maps $\mathcal{E}: \mathcal{L}(\mathcal{H}_P) \rightarrow \mathcal{L}(\mathcal{H}_I \otimes \mathcal{H}_c)$  and $\mathcal{D}: \mathcal{L}(\mathcal{H}_O \otimes \mathcal{H}_c) \rightarrow \mathcal{L}(\mathcal{H}_F)$, where $\Hcal_c$ is an auxiliary space representing memory. Whenever a supermap can be realized via an underlying circuit, we will denote the temporal order of occurrence of spaces by the $\prec$ symbol, i.e., for this superchannel we have $P \prec I \prec O \prec F$. }
 \label{fig::enc_dec}
\end{figure}

More specifically, a quantum supermap must map valid quantum channels/operations into valid quantum channels/operations even when acting locally on part of a larger quantum process~\cite{chiribella_transforming_2008}. To make requirement clear, suppose that a supermap $\mathcal S$ acts locally on a quantum channel/operation $\mathcal T$ with  additional  input and output systems $a$ and $b$, respectively, as illustrated in Fig.~\ref{fig::CP_supmap}.    
  In this setting, the basic consistency requirement is that $\cal S$ must  map every valid quantum channel/operation $ \mathcal T\in\mathsf{CP}(I\otimes a,O\otimes b)$ into a valid quantum/channel  operation  $ \mathcal T'\in\mathsf{CP}(P\otimes a,F\otimes b)$, for every possible choice of additional systems $a$ and $b$.  This requirement leads to the following definition:

\begin{definition}[Positive, completely positive, and normalization-preserving supermaps]
A linear supermap $\mathcal{S}: \mathsf{CP}(I,O) \rightarrow \mathsf{CP}(P,F)$ is 
\begin{enumerate}[1)]
\item {\em positive}, if it transforms completely positive maps into completely positive maps, namely $\mathcal S[\mathcal T]  \in  \textup{CP}  (P,F)$ for every $\mathcal{T} \in  \textup{CP} (I,O)$ 
\item {\em completely positive (CP)} if the supermap   $\mathcal{S} \otimes \mathcal I_{a\to b}$ is positive for every pair of finite-dimensional quantum systems $a$ and $b$ (here, $I_{a \to b}$ denotes the identity supermap on $\mathcal{L}(\mathcal{L}( \mathcal{H}_a ), \mathcal{L}( \mathcal{H}_b))$, the space of linear maps from $\mathcal{L}( \mathcal{H}_a$ to $\mathcal{L}( \mathcal{H}_b)$). 
\item {\em normalization-preserving}  if $\mathcal{S}[\mathcal{T}]$ is  trace-preserving whenever $\mathcal{T}$ is trace-preserving. 
\end{enumerate}
\end{definition}

A supermap that is both completely positive and normalization-preserving is called a {\em deterministic supermap} or, in the later literature, a {\em superchannel}.    Here in this paper, we will sometime use the wording {\em ``proper supermap"} as a synonim of deterministic supermap.  

For quantum supermaps at this level of the hierarchy, it turns out that complete positivity and normalization preservation are sufficient to guarantee the existence of an underlying, causally ordered quantum circuit, as illustrated in Fig.~\ref{fig::enc_dec}).  Explicitly, this fact is proven by the following: 

\begin{prop}[Superchannels and quantum circuits~\cite{chiribella_transforming_2008, chiribella_quantum_2008}]
    The action of a superchannel $\mathcal{S}: \textup{CP}(I,O) \rightarrow  \textup{CP}(P,F)$ can always be written in terms of a quantum circuit 
\begin{equation}\label{eq:propMapDecomp}
    \mathcal{S}[\mathcal{T}] = \mathcal{D} \circ (\mathcal{T} \otimes \mathcal{I}_c) \circ \mathcal{E},
\end{equation}
where $c$ is a finite dimensional quantum system, and  $\mathcal{E} \in \mathsf{CPTP}(P, I\otimes c)$ and $\mathcal{D} \in \mathsf{CPTP}(O\otimes c, F)$ are two quantum channels. 
\end{prop}

This proposition shows that all the quantum superchannels at the basic level are compatible with a well-defined underlying causal order.   More generally, a class of higher-order transformations known as quantum combs \cite{chiribella_quantum_2008, Chiribella2009, Bisio2011} satisfies this property. 
Quantum combs arise, for example, in  scenarios of distributed quantum computation, such as the one shown in Fig.~\ref{fig::multislot}.  In this example, four users (Alice, Bob, Charlie, and Dave) can freely manipulate different qubits, but they have no control over the quantum circuit in between their respective operations.
\begin{figure}
\centering
    \includegraphics[width = 0.95\columnwidth]{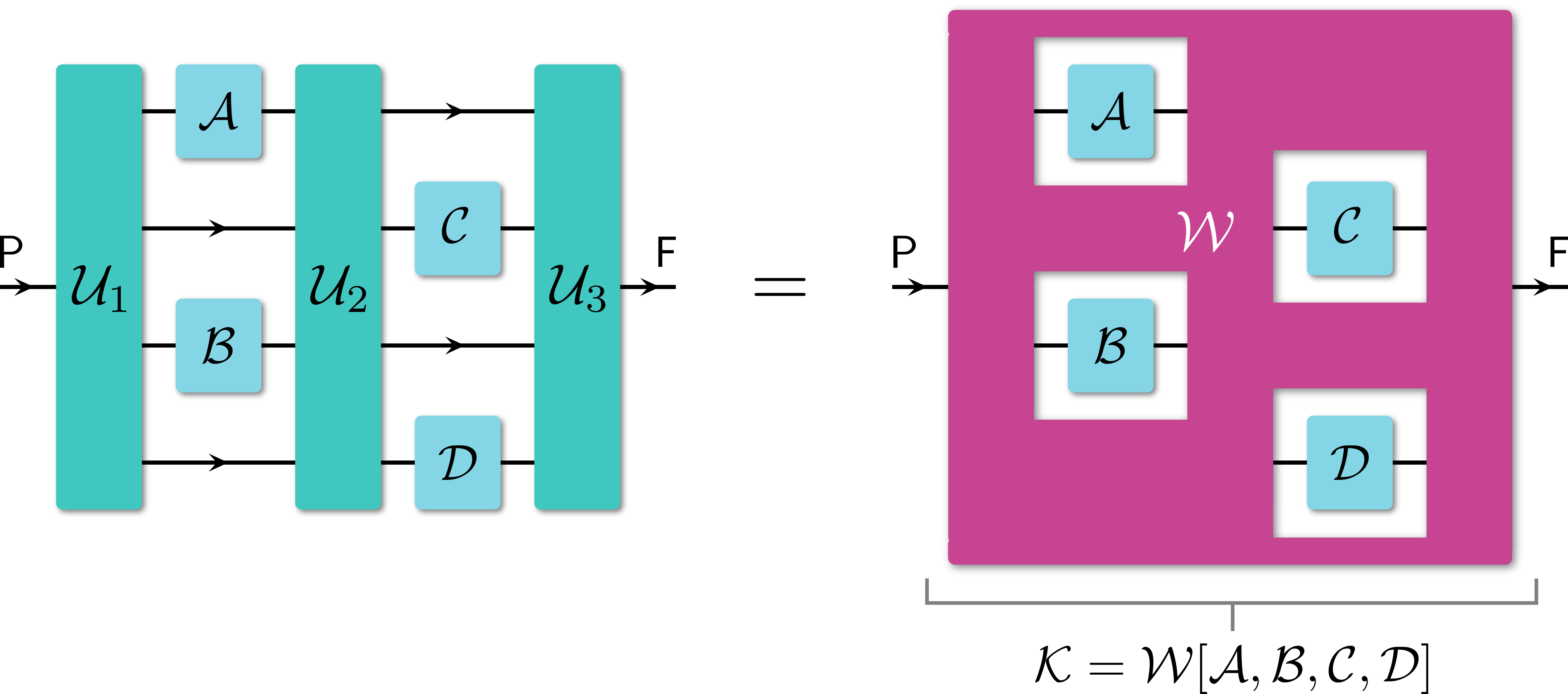}
    \caption{\textbf{Higher-order transformation from a quantum circuit.} Different users (here, Alice, Bob, Charlie and Dave) in a distributed quantum computation -- schematically given by a circuit with unitaries $\mathcal{U}_1, \mathcal{U}_2,$ and $\mathcal{U}_3$ -- can perform manipulations $\mathcal{X}\in\mathsf{CP}(I_X, O_X)$ on different qubits at different points in time. The overall action of the circuit onto the corresponding operations is given by a higher-order transformation $\mathcal{W}$, such that $\mathcal{K} = \mathcal{W}[\mathcal{A}, \mathcal{B}, \mathcal{C}, \mathcal{D}] \in \mathsf{CP}(P,F)$. Here, $\mathcal{W}$ is causally ordered -- since it stems from a quantum circuit. More generally, higher-order transformations can also account for causally indefinite correlations between different operations.}
    \label{fig::multislot}
\end{figure}
In this situation, the users can perform  arbitrary quantum operations $\mathcal{A}, \mathcal{B}, \mathcal{C},$ and $\mathcal{D}$, which inserted  in the  circuit yield an overall quantum operation   $\mathcal{K} \in \Lcal(\Lcal(\Hcal_P),\Lcal(\Hcal_F))$. This operation depends linearly on each of  the operations $\mathcal{A}, \mathcal{B}, \mathcal{C},$ and $\mathcal{D}$, and the dependence can be represented by a quantum supermap $\mathcal{W}$, as follows
\begin{gather}
    \mathcal{K} = \mathcal{W}[\mathcal{A}, \mathcal{B}, \mathcal{C}, \mathcal{D}] \, . 
\end{gather}
The map  $\mathcal{W}$ here is  an example of a quantum comb, {\em i.e.} of a quantum supermap compatible with an underlying causally ordered circuit (when the operations of two parties are performed in parallel, one can nominally just assume that one of the two operations happens before the other, without affecting the overall supermap).  
While many important transformations are described by quantum combs, there also exist higher-order transformations that cannot be embedded in a definite causal order~\cite{chiribella2009beyond, oreshkov_quantum_2012, chiribella_quantum_2013, branciard_simplest_2015, Araujo2015,  castro-ruiz_dynamics_2018, wechs_definition_2019, milz_characterising_2024,  costa_indefinite_2026}.   In the following, our aim will be  to identify which  higher-order transformations can  be regarded as thermodynamically admissible. Surprisingly, we will find that while quantum mechanics in principle allows for  causally indefinite transformation, the preservation of thermodynamical equilibrium necessarily implies a definite causal order.

\subsection{Preservation of thermal equilibrium at the state level}
\label{sec::Gibbs_eq_trans}

The core idea of this paper is to determine the largest class of  higher-order transformations that preserve the condition of thermal equilibrium.   Our guiding principle is simple: 
\begin{principle}[Equilibrium preservation]\label{basicprinciple} In a thermodynamical context, the allowed transformations   should not generate deviations from thermal  equilibrium.
\end{principle}

Before considering the case of general higher-order transformations, it is useful to summarize the key notion at the basic level of quantum states. At this level, the condition of thermal equilibrium singles out one and only one state, namely the \textit{Gibbs state}. For a quantum system with Hamiltonian $\hat{H}$, coupled to a thermal bath at inverse temperature $\beta$, the Gibbs state is 
\begin{gather}
    \tau = e^{-\beta \hat{H}}/Z, \qquad \text{with} \quad Z = \tr[e^{-\beta \hat{H}}].  
\end{gather}
This is the unique equilibrium state, maximizing the entropy of the system for a fixed average energy $E = \langle \hat{H}\rangle_\rho$ and preventing the extraction of any amount of work from the system.

For transformations of quantum states, the requirement  that the allowed transformations   should not generate deviations from equilibrium  (Principle \ref{basicprinciple}) translates into the requirement that the allowed quantum channels/operations must preserve the Gibbs state (see Fig. \ref{fig::GP_GPTP}). 

\begin{figure}
    \centering
    \includegraphics[width = 0.75\linewidth]{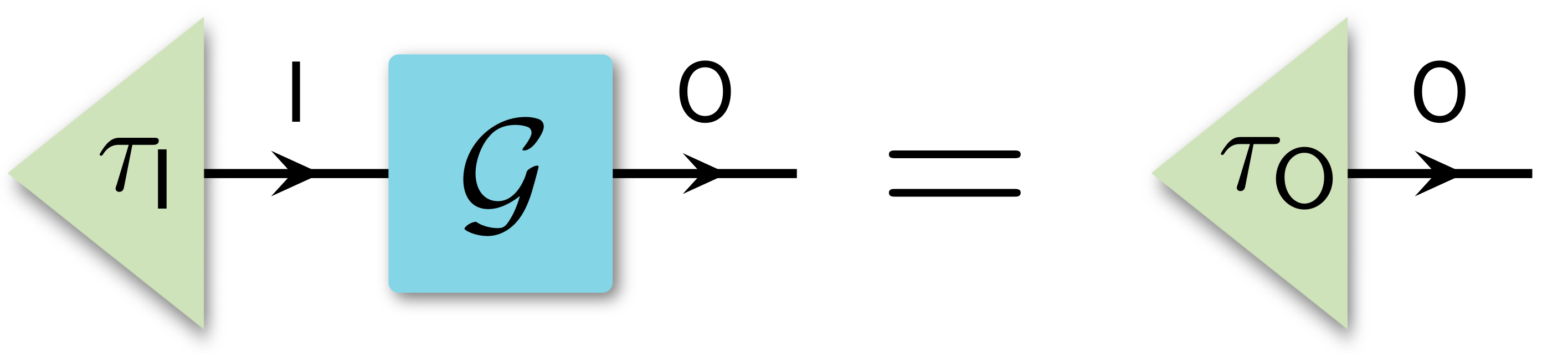}
    \caption{\textbf{Equilibrium-preserving transformations of quantum states.} A map $\mathcal{G} \in \Lcal(\Lcal(\Hcal_I), \Lcal(\Hcal_O))$ is Gibbs-preserving (GP) if it is CP and maps Gibbs states to Gibbs states. If, in addition, $\mathcal{G}$ is also TP, then it is Gibbs-preserving and trace-preserving (GPTP).}
    \label{fig::GP_GPTP}
\end{figure}

\begin{definition}[Gibbs-preserving transformations]
    A linear map $\mathcal{G} \in \mathsf{CP}(I, O)$ is called Gibbs-preserving (GP) if
    \begin{equation}
        \mathcal{G}[\tau_{I}] = \tau_{O},
    \end{equation}
    where $\tau_{I}$ and $\tau_{O}$ are the Gibbs states on the spaces $\mathcal{H}_I$ and $\mathcal{H}_O$, respectively. The corresponding class of GP maps is denoted $\mathsf{GP}(I,O)$.
    
    If, in addition, $\mathcal{G} \in \mathsf{CPTP}(I, O)$, then it is called a Gibbs-preserving and trace-preserving (GPTP) map. The corresponding class of GPTP maps is denoted $\mathsf{GPTP}(I,O)$. \label{def:mapsGP}
\end{definition}
Throughout this work each Hilbert space $\mathcal H_X$ will be be associated with its own Hamiltonian $\hat H_X$ and corresponding Gibbs state $\tau_X$. In this way, we will allow  transformations between quantum systems with different Hamiltonians and, more generally, between systems coupled to baths of different temperatures.

In the literature, other possible classes of thermodynamical transformations -- most prominently among them thermal operations~\cite{Janzing2000, horodecki_fundamental_2013, brandao_second_2015, brandao_resource_2013, faist_gibbs-preserving_2015} -- have been proposed based  on restrictions on their underlying implementation. In this work we focus on GP and GPTP maps since they a) form the largest possible set of equilibrium transformations, requiring the least number of assumptions and admitting a simple algebraic characterization (see App.~\ref{app:GP_GPTP}); and b) they are solely based on thermodynamic input-output relations and do not require an explicit circuit model for their conceptualization. This allows the systematic extension of equilibrium to higher-order quantum transformations, enabling a detailed study of its relation to causal structure. 

As for the case of positivity, one may wonder if requiring Gibbs preservation in a \textit{complete} sense adds further restrictions. Here, complete Gibbs-preservation would mean that a map $\mathcal{G}$ not only preserves the Gibbs state, but the set of \textit{locally thermal} states $\eta \in \Lcal(\Hcal_c \otimes \Hcal_I)$, with $\tr_c[\eta] = \tau_I$. That is, a map $\mathcal{G}$ is completely Gibbs preserving (CGP) if  
\begin{gather}
    \tr_c[(\Ical \otimes \mathcal{G})[\eta]] = \tau_O \quad \forall \ \text{locally thermal} \ \eta. 
\end{gather}
Since we have 
\begin{gather}
    \tr_c[(\Ical \otimes \mathcal{G})[\eta]]  = \mathcal{G}[\tr_c[\eta]] = \mathcal{G}[\tau_I] = \tau_O,
\end{gather}
it is easy to see that all Gibbs-preserving maps already satisfy this property, such that completeness does not add any additional restrictions to $\mathsf{GP}$ and $\mathsf{GPTP}$.   Later in the paper (Sec.~\ref{subsec:completeSupermaps}), we will see that the situation will change radically for higher-order transformations: in that case,  the requirement of \textit{complete} equilibrium preservation adds non-trivial restrictions.

\subsection{Higher-order equilibrium transformations}
\label{subsec:higherorder}

We now extend the notion of equilibrium preservation from quantum states and processes  to more advanced scenarios described by  higher-order transformations.  For quantum supermaps of the basic type (see Def.~\ref{basicsupermap}),  the principle of equilibrium preservation  (Principle \ref{basicprinciple}) amounts to the fact that the allowed supermaps should  transform Gibbs-preserving quantum channels/operations into  Gibbs-preserving quantum channels/operations.  More generally,  Principle \ref{basicprinciple}  generates an infinite hierarchy of equilibrium transformations:
\begin{definition}[Higher-order equilibrium-preserving transformations]
\label{def::equ_pres}
    A higher-order transformation is  equilibrium-preserving if it maps equilibrium-preserving inputs to equilibrium-preserving outputs. At the state level, the Gibbs states  are regarded as the basic equilibrium-preserving objects.
\end{definition}
In the following, equilibrium-preserving transformations will be just called {\it equilibrium transformations}, for short.   The goal of the next sections will be to characterize the set of higher-order equilibrium transformations, starting from the basic example of supermaps transforming quantum channels/operations into quantum channels/operations.

\section{The simplest higher order case: one-slot supermaps}
\label{eqn::one_slot_case}
We initially restrict our attention to  the first nontrivial level of the hierarchy, namely, supermaps of the type (\ref{basicsupermap}),  transforming  quantum channels/operations into quantum channels/operations. 

\subsection{Equilibrium supermaps}
\label{subsec::eq_supmap}
While the equilibrium state is \textit{unique}, there are already two natural notions of equilibrium state transformations: probabilistic GP maps and deterministic GPTP maps. At the level of supermaps then, two independent questions arise. First, should equilibrium be defined relative to GP maps or GPTP maps? Second, should equilibrium supermaps themselves be required to be superchannels? The latter question is motivated by causality: neither the preservation of GP maps nor the preservation of GPTP maps is, by itself, sufficient to guarantee causal order, rendering it an independent requirement. These two choices lead to four natural candidates for equilibrium supermaps $\mathcal{S} : \mathcal{L}(\mathcal{L}(\mathcal{H}_I), \mathcal{L}(\mathcal{H}_O)) \rightarrow \mathcal{L}(\mathcal{L}(\mathcal{H}_P), \mathcal{L}(\mathcal{H}_F))$:
\begin{enumerate}
    \item[(R$1$)] \textbf{GP-preserving supermaps.} 
    $\mathcal{S}$ preserves the set of Gibbs-preserving transformations, i.e.,
    \[ \mathcal S[\mathcal G]\in\mathsf{GP}(P,F) \quad \forall \mathcal G\in\mathsf{GP}(I,O). \]
    We denote the set of such \emph{GP-preserving} supermaps by $\mathbf{GPP}$.

    \item[(R$2$)] \textbf{GPTP-preserving supermaps.} 
    $\mathcal{S}$ preserves the set of deterministic Gibbs-preserving transformations, i.e., 
    \[ \mathcal S[\mathcal G]\in\mathsf{GPTP}(P,F) \qquad \forall \mathcal G\in\mathsf{GPTP}(I,O). \]
    We denote the set of such \emph{GPTP-preserving} supermaps by $\mathbf{GPTPP}$.
    
    \item[(R$1'$)] \textbf{Deterministic GP-preserving supermaps.} 
    $\mathcal{S}$ is GP-preserving and it is a deterministic supermap, {\em a.k.a.} a  superchannel. We denote the set of  \emph{deterministic GP-preserving} supermaps by $\mathbf{pGPP}$.

    \item[(R$2'$)] \textbf{Deterministic GPTP-preserving supermaps.} 
    $\mathcal{S}$ is both GPTP-preserving and a superchannel. We denote the set of such \emph{proper GPTP-preserving} supermaps $\mathcal{S}$ by $\mathbf{pGPTPP}$.
\end{enumerate}
All four requirements above follow from the principle of equilibrium preservation. R$1$ and R$2$ impose that probabilistic and deterministic equilibrium state transformations must be preserved, respectively. R$1'$ and R$2'$ strengthen these requirements by additionally demanding that the equilibrium supermaps themselves are superchannels. Clearly, 
\[ \mathbf{pGPP}\subseteq \mathbf{GPP}, \qquad \mathbf{pGPTPP}\subseteq \mathbf{GPTPP}. \] 
In addition, all four sets of possible equilibrium transformations are genuinely distinct. We analyze the properties and mutual relationship below and in more detail in App.~\ref{app:proofHierarchy}. Their inclusion structure is depicted in Fig.~\ref{fig::classes} and summarized by the following Theorem: 

\begin{theorem}[Relation between classes of free transformations]\label{thm:hierarchy}
    The classes of free transformations satisfy
    \begin{align}\label{eq:hierarchy}
        \mathbf{pGPP} \subsetneq \mathbf{pGPTPP} \subsetneq \mathbf{GPTPP},
    \end{align}
    and
    \begin{align}
        \mathbf{GPP} \cap \mathbf{pGPTPP} = \mathbf{pGPP}.
    \end{align}
    In addition, $\mathbf{GPP}$ and $\mathbf{GPTPP}$ do not coincide, i.e.,
    \begin{gather}
        \mathbf{GPP} \setminus \mathbf{GPTPP} \neq \emptyset \quad \text{and} \quad  \mathbf{GPTPP} \setminus \mathbf{GPP} \neq \emptyset.
    \end{gather}
\end{theorem}
At first sight, causal order may appear to be a natural property for equilibrium supermaps. However, at this point, it is neither implied by equilibrium preservation, nor it is a requirement to yield mathematically well-defined and physically motivated transformations. For example, the quantum time flip~\cite{Liu2022_Flip, stromberg_experimental_2024} lies in $\mathbf{GPTPP}$ but not in $\mathbf{pGPTPP}$ (see Sec.~\ref{subsec::GPTPP_pGPTPP} below).
\begin{figure}
\centering
    \includegraphics[width = \columnwidth]{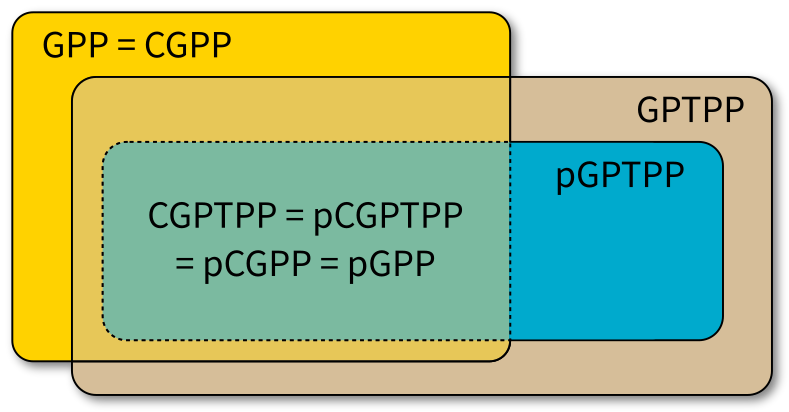}
    \caption{\textbf{Relation between the classes of free transformations.} The requirements (R$1'$), (R$2'$), and (R$2$) define a hierarchy of classes: $\mathbf{pGPP} \subsetneq \mathbf{pGPTPP} \subsetneq \mathbf{GPTPP}$. By contrast, the class $\mathbf{GPP}$ is not contained in this chain but rather overlaps with both $\mathbf{GPTPP}$ and $\mathbf{pGPTPP}$. In particular, its intersection with the latter yields the class $\mathbf{pGPP}$.}
    \label{fig::classes}
\end{figure}

\subsection{Classes of equilibrium transformations}
\label{subsec:free_transform}

\subsubsection{GPP and pGPP transformations}
\label{subsubsec::GPP_pGPP}
We begin with supermaps $\mathcal{S} \in \mathbf{GPP}$ that satisfy R$1$. That is, they preserve probabilistic Gibbs-preserving transformations. Any such supermap satisfies $\mathcal S[\mathcal G]\in\mathsf{GP}(P,F)$ for all $\mathcal G\in\mathsf{GP}(I,O)$ and thus 
\begin{equation}
    (\mathcal{S}[\mathcal{G}])[\tau_P] = \tau_F \quad \forall \mathcal G \in \mathsf{GP}(I,O).
\end{equation}
In App.~\ref{app:GPP}, we provide a full mathematical characterization of $\mathbf{GPP}$. In particular, we show that this set admits a particularly simple physical interpretation: while $\mathcal{S}$ is defined through its action on GP maps $\mathcal{G}$, we can also consider what happens when it is evaluated on a Gibbs state $\tau_P\in\mathcal{L}(\mathcal{H}_P)$. The elements of $\mathbf{GPP}$ are then exactly those supermaps that do not allow for the extraction of non-equilibrium resources. That is, evaluated on $\tau_P$, the supermap $\mathcal{S}$ induces a Gibbs state $\tau_I \in \mathcal{L}(\mathcal{H}_I)$ and a Gibbs-preserving transformation $\mathcal G_{\mathcal S} \in \mathsf{GP}(O,F)$ (see Fig.~\ref{fig::charGPP}). More explicitly, we have the following Lemma: 
\begin{lemma}[Thermality of (p)GPP supermaps] 
$\mathcal{S}$ is a GPP supermap iff it satisfies 
\begin{gather}
    \label{eq:PropGPP}
    (\mathcal{S}[\mathcal{T}])[\tau_{P}] = \mathcal{G}_{\mathcal{S}}\!\left[\mathcal{T}[\tau_{I}]\right],
\end{gather}
for arbitrary $\mathcal{T}\in \mathsf{CP}(I,O)$, where $\mathcal{G}_{\mathcal{S}} \in \mathsf{GP}(O,F)$. If, in addition, $\mathcal{S}$ is a superchannel, then it is a pGPP map.
\end{lemma}
Consequently, GPP and pGPP supermaps do not allow for the extraction of non-equilibrium resources, both when acting on an equilibrium transformation $\mathcal G\in\mathsf{GP}(I,O)$ and when probed with an equilibrium input state $\tau_P$, further underlining their interpretation as \textit{equilibrium} transformations. Throughout, we call supermaps with this property [i.e., they satisfy Eq.~\eqref{eq:PropGPP}] \textit{thermal}.

We emphasize that causal ordering does \textit{not} follow from equilibrium principles, but must be imposed manually (we provide an example of a supermap $\mathcal{S} \in \mathbf{GPP} \setminus \mathbf{pGPP}$ in the next section). This is not surprising though: GPP supermaps are only required to preserve probabilistic GP maps. This is too weak a requirement to enforce the corresponding supermaps to be proper, and thus deterministic. At first sight, then, the class $\mathbf{pGPP}$ appears to provide the most satisfactory notion of equilibrium supermaps, combining thermality and causal order. However, since the latter property enters as an additional assumption rather than as a consequence of equilibrium preservation itself, this naturally raises the question whether a stronger formulation of equilibrium can simultaneously enforce thermality \textit{and} causality from a \textit{single} thermodynamic principle.

Remarkably, thermality, the shared property of GPP and pGPP supermaps, is not shared by all candidate notions of equilibrium. As we see next, GPTPP and pGPTPP supermaps may preserve deterministic equilibrium transformations while nevertheless permitting the extraction of non-equilibrium resources. 

\begin{figure}
\centering
    \includegraphics[width = \columnwidth]{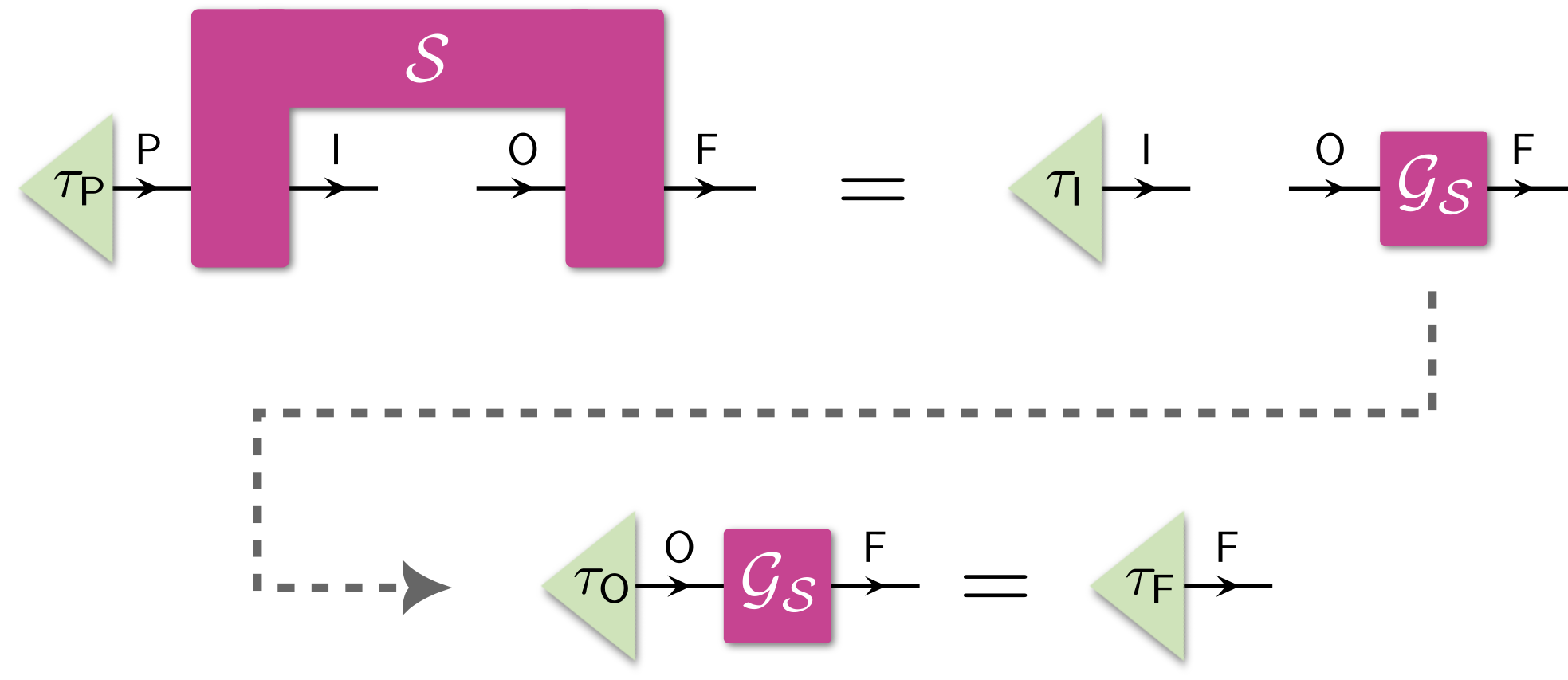}
    \caption{\textbf{Characterization of $\mathbf{GPP}$ supermaps.} If the Gibbs state is prepared in the global past $P$, any supermap $\mathcal{S} \in \mathbf{GPP}$ sends it to the input space $\mathcal{H}_I$ of the input operation. The induced transformation then acts as a Gibbs-preserving map $\mathcal{S}_{\mathrm{op}}: \mathcal{L}(\mathcal{H}_O) \rightarrow \mathcal{L}(\mathcal{H}_F)$ on the output of the operation. Throughout, we call supermaps with this property \textit{thermal}. If, in addition, $\mathcal{S}$ is a superchannel, then $\mathcal{S} \in \mathbf{pGPP}$.}
    \label{fig::charGPP}
\end{figure}

\subsubsection{GPTPP and pGPTPP transformations}
\label{subsec::GPTPP_pGPTPP}
Requirements R$2$ and R$2'$ yield the classes $\mathbf{GPTPP}$ and $\mathbf{pGPTPP}$ of equilibrium supermaps and superchannels that preserve GPTP transformations. Remarkably, preservation of GPTP transformations alone does not imply that $\mathcal{S}\in \mathbf{GPTPP}$ is causally ordered; the latter implies that $\mathcal{S}$ maps \textit{any} CPTP map onto a CPTP map, while $\mathcal{S} \in \mathbf{pGPTPP}$ only forces preservation of the subset of CPTP maps that are Gibbs preserving. We thus have the inclusion $\mathbf{pGPTPP} \subsetneq \mathbf{GPTPP}$. In App.~\ref{app:proofHierarchy}, we provide a supermap $\mathcal{S} \in \mathbf{GPTPP} \setminus \mathbf{pGPTPP}$, proving that this inclusion is strict. For the infinite temperature case, where $\tau=\ident/d$, a well-known example of such a supermap is given by the quantum time flip $\mathcal{S}_\mathsf{QTF}$~\cite{Liu2022_Flip, stromberg_experimental_2024} which superposes the temporal directions of its input transformation. While it preserves bistochastic maps (the GPTP maps for the infinite temperature case), it cannot be implemented through a quantum circuit, implying $\mathcal{S}_\mathsf{QTF} \in  \mathbf{GPTPP} \setminus \mathbf{pGPTPP}$.

To analyze these two candidates of equilibrium transformations, we first consider $\mathbf{pGPTPP}$. It is easy to see that this class has an overlap with $\mathbf{GPP}$, i.e., $\mathbf{pGPTPP}\cap \mathbf{GPP} \neq \emptyset$ (see below). However, they do not coincide. The mathematical structure of $\mathbf{pGPTPP}$ is more intricate than that of $\mathbf{GPP}$ and does not admit a similarly intuitive characterization (for the full mathematical characterization, see App.~\ref{app:GPTPP}). In particular, unlike $\mathbf{GPP}$, the class $\mathbf{pGPTPP}$ contains supermaps from which non-equilibrium resources can be extracted:

\begin{example}[Athermality extraction from pGPTPP supermaps]\label{example:TR}
    Consider the supermap $\mathcal{S}_\mathsf{TR}$ defined for any $\mathcal{T} \in \mathsf{CP}(I,O)$ as
    \begin{equation}
        \mathcal{S}_\mathsf{TR}[\mathcal{T}] = \operatorname{Tr}\big[\mathcal{T}[\eta]\big] \, \overline{\mathcal{G}},
    \end{equation}
    where $\overline{\mathcal{G}}[\rho] = \operatorname{Tr}[\rho] \, \tau_{F}$ for all $\rho \in \mathcal{L}(\mathcal{H}_P)$ is the completely thermalizing channel. Operationally, $\mathcal S_{\mathsf{TR}}$ inserts the state $\eta$ into $\mathcal T$, records only the success probability $\operatorname{Tr}[\mathcal T[\eta]]$, and outputs the completely thermalizing channel weighted by this factor. Evidently, $\mathcal{S}_\mathsf{TR}$ preserves the sets $\mathsf{GPTP}$ as well as $\mathsf{CPTP}$ and therefore lies in $\mathbf{pGPTPP}$ (as well as $\mathbf{GPTPP}$). On the other hand, if $\eta \neq \tau_I$ we have $\mathcal{S}_\mathsf{TR}\notin \mathbf{GPP}$. We can see this directly by observing that $\mathcal{S}_\mathsf{TR}$ outputs $\eta \neq \tau_I$ -- and thus allows for the extraction of athermality -- when evaluated on the Gibbs state $\tau_P$. Similarly, direct computation shows that $\mathcal{S}_\mathsf{TR}[\mathcal{G}]$ is not necessarily a GP map for all $\mathcal{G} \in \mathsf{GP}(I,O)$:
\begin{gather}
    \mathcal{S}_\mathsf{TR}[\mathcal{G}][\tau_P] = \operatorname{Tr}\big[\mathcal{G}[\eta]\big] \, \operatorname{Tr}[\tau_P] \tau_F = \operatorname{Tr}\big[\mathcal{G}[\eta]\big] \, \tau_F.
\end{gather}
If $\mathcal{G}$ is not TP, then there exists a state $\eta \in \mathcal{L}(\mathcal{H}_I)$ such that $\operatorname{Tr}\big[\mathcal{G}[\eta]\big] < 1$. For such a choice,  we see that $ \mathcal{S}_\mathsf{TR}[\mathcal{G}]$ maps $\tau_P$ to $\alpha \tau_F$, with $\alpha < 1$, which is not a Gibbs state. Hence \( \mathcal S_{\mathsf{TR}}[\mathcal G]\notin\mathsf{GP}(P,F). \)
\end{example}
The trace-and-replace construction above immediately shows that 
\begin{gather}
    \mathbf{pGPTPP}\setminus\mathbf{GPP}\neq\emptyset \quad \text{and} \quad  \mathbf{GPTPP}\setminus\mathbf{GPP}\neq\emptyset.
\end{gather} It follows that GPTPP preservation alone is insufficient to guarantee thermality, even if one requires the corresponding supermaps to be causal.
The converse inclusion also fails and simple modifications of the trace-and-replace construction above show that \[ \mathbf{pGPTPP}\cap\mathbf{GPP}\neq\emptyset, \qquad \mathbf{GPP}\setminus\mathbf{GPTPP}\neq\emptyset. \] Hence neither class is contained in the other. In particular, choosing $\eta = \tau_I$ in the definition of $\mathcal{S}_\mathsf{TR}$ yields a supermap that lies both in $\mathbf{GPP}$ \textit{and} $\mathbf{pGPTPP}$. Additionally replacing $\overline{\mathcal{G}}$ by a transformation $\mathcal{G} \in \mathsf{GP}(P,F) \setminus \mathsf{GPTP}(P,F)$ yields a supermap that lies in $\mathbf{GPP}\setminus \mathbf{GPTPP}$ as well as in $\mathbf{GPP}\setminus \mathbf{pGPP}$, proving the corresponding claim in the section above. We thus arrive at the set inclusions depicted in Fig.~\ref{fig::classes}. 

Crucially, GPTPP supermaps are neither causally ordered, nor thermal. The latter pathology is also not cured by manually imposing causal order: even \textit{proper} GPTPP supermaps can contain extractable non-equilibrium resources (see Ex.~\ref{example:TR}). On the other hand, as we show in App.~\ref{app:proofHierarchy}, we have
\begin{gather}
    \mathbf{GPP} \cap \mathbf{pGPTPP} = \mathbf{pGPP}.
\end{gather}
That is, the set of superchannels that are thermal and preserve GPTP transformations corresponds exactly with the set $\mathbf{pGPP}$ of thermal superchannels discussed in the previous section. This observation motivates the search for a stronger equilibrium principle -- provided by \textit{complete} equilibrium preservation -- that enforces both thermality and causal order simultaneously, thereby singling out $\mathbf{pGPP}$ as the unique physically meaningful class of equilibrium supermaps.

\subsection{Complete equilibrium preservation}
\label{subsec:completeSupermaps}
The four notions of equilibrium supermaps introduced above are defined in terms of their action on individual transformations. Experience from quantum theory suggests that this may not be sufficient: positive maps are not physically admissible unless they remain positive when acting on arbitrary subsystems, leading to the concept of \textit{complete} positivity. Motivated by this observation, we require equilibrium preservation to hold in the same operationally robust sense. That is, a supermap should preserve equilibrium not only for isolated transformations, but also when acting on part of a larger equilibrium transformation. This guarantees that it does not generate non-equilibrium outputs from equilibrium inputs solely by exploiting correlations with an ancillary system (see Fig.~\ref{fig::Completeness}). We thus obtain the following notion of \textit{complete} equilibrium preservation. 
\begin{definition}[Complete equilibrium preservation]
    Let $\mathcal{S}: \mathsf{CP}(I,O) \rightarrow \mathsf{CP}(P,F)$ be supermap. We call $\mathcal{S}$ completely equilibrium-preserving if, for every pair of auxiliary systems $\mathcal{H}_{a}$ and $\mathcal{H}_{b}$, the extended supermap
    \begin{align}
        \nonumber \mathcal{S} \otimes \mathfrak{id}\!:\; &\mathcal{L}(\mathcal{L}(\mathcal{H}_I\otimes\mathcal{H}_{a}),\mathcal{L}(\mathcal{H}_O\otimes\mathcal{H}_{b})) \\
        \qquad &\rightarrow \mathcal{L}(\mathcal{L}(\mathcal{H}_P\otimes\mathcal{H}_{a}),\mathcal{L}(\mathcal{H}_F\otimes\mathcal{H}_{b})),
    \end{align}
    is an equilibrium transformation, where $\mathfrak{id}$ is an identity supermap.
\end{definition}
\begin{figure}
    \centering
    \includegraphics[width = 0.9\linewidth]{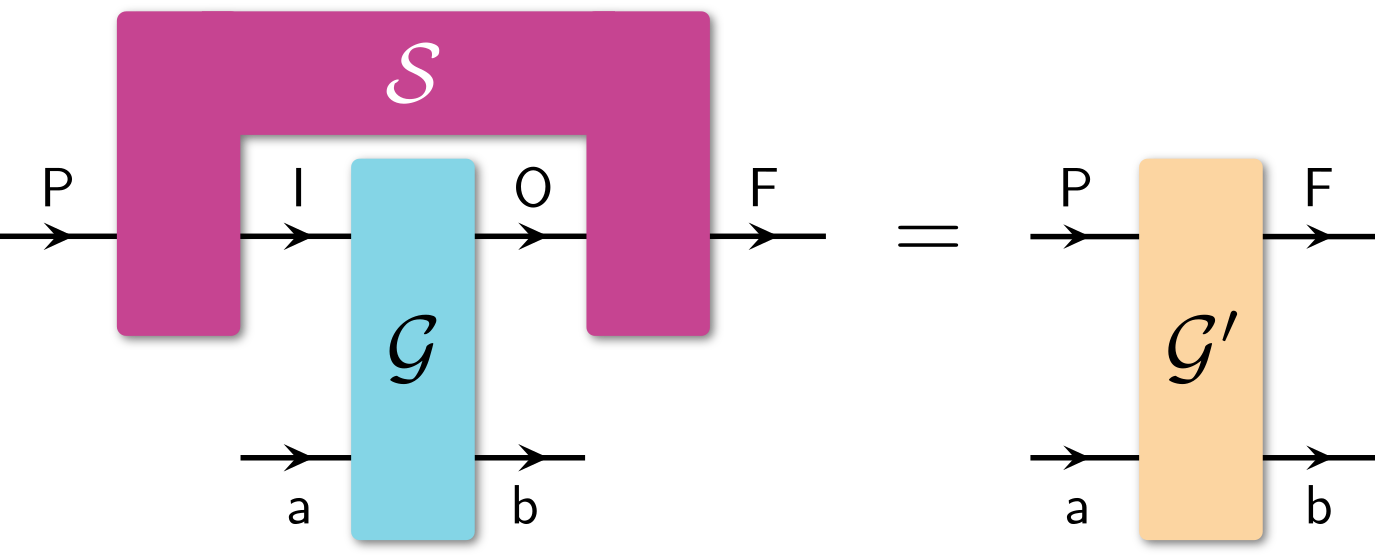}
    \caption{\textbf{Completeness of properties.} A supermap $\mathcal{S}$ preserves a property \textit{completely}, if it preserves it when acting on an arbitrary part of a larger transformation. For example, $\mathcal{S} \in \mathbf{CGPP}$ if it maps any Gibbs-preserving map $\mathcal{G}$ to a Gibbs-preserving map $\mathcal{G}'$ even when only acting on a part of it.}
    \label{fig::Completeness}
\end{figure}
Throughout, each system is assumed to be equipped with its own, well-defined Hamiltonian and corresponding Gibbs state.  Equilibrium preservation of $\mathcal{G}$ therefore requires the joint thermal state of the system and ancilla to be mapped to the corresponding joint thermal state at the output.
That is, a Gibbs-preserving map $\mathcal{G}: \mathcal{L}(\mathcal{H}_I \otimes \mathcal{H}_a) \rightarrow \mathcal{L}(\mathcal{H}_O \otimes \mathcal{H}_b)$ satisfies 
\begin{gather}
    \mathcal{G}[\tau_I \otimes \tau_a] =  \tau_O \otimes \tau_b. 
\end{gather}
The corresponding sets of GP and GPTP transformations are denoted by 
\begin{gather}
    \mathsf{GP}(I\otimes a, O \otimes b), \quad \text{and} \quad \mathsf{GPTP}(I\otimes a, O \otimes b),
\end{gather}
and similarly for input and output spaces $P$ and $F$, respectively. With this, the requirements for equilibrium transformations provided in Sec.~\ref{subsec::eq_supmap} can be strengthened to incorporate \textit{complete} equilibrium preservation as follows:
\begin{enumerate}
    \item[(CR1)] \textbf{Completely GP-preserving supermaps.} 
    Any trivial extension $\mathcal{S} \otimes \mathfrak{id}$ of $\mathcal{S}$ preserves the set of Gibbs-preserving transformations, i.e., 
    \begin{equation}
        \bigl(\mathcal{S} \otimes \mathfrak{id})[\mathcal{G}] \in \mathsf{GP}(P\otimes a, F \otimes b) 
    \end{equation}
    for all $\mathcal{G} \in \mathsf{GP}(I \otimes a, O \otimes b)$. We denote this class of \emph{completely GP-preserving} supermaps $\mathcal{S}$ by $\mathbf{CGPP}$.

    \item[(CR2)] \textbf{Completely GPTP-preserving supermaps.} 
    Any trivial extension $\mathcal{S} \otimes \mathfrak{id}$ of $\mathcal{S}$ preserves the set of Gibbs-and-trace-preserving transformations, i.e., 
    \begin{equation}
        \bigl(\mathcal{S} \otimes \mathfrak{id})[\mathcal{G}] \in \mathsf{GPTP}(P\otimes a, F \otimes b)
        \end{equation}
        for all $\mathcal{G} \in \mathsf{GPTP}(I \otimes a, O \otimes b)$. We denote this class of \emph{completely GPTP-preserving} supermaps $\mathcal{S}$ by $\mathbf{CGPTPP}$.
    
    \item[(CR$1'$)] \textbf{Proper completely GP-preserving supermaps.} 
     $\mathcal{S}$ is completely GP-preserving and a superchannel. We denote this class of \emph{proper completely GP-preserving} supermaps by $\mathbf{pCGPP}$.

    \item[(CR$2'$)] \textbf{Proper complete GPTP-preserving supermaps.} 
     $\mathcal{S}$ is completely GPTP-preserving and a superchannel. We denote this class of \emph{proper completely GPTP-preserving} supermaps by $\mathbf{pCGPTPP}$.
\end{enumerate}

Characterizing the structure of the supermaps that are completely equilibrium preserving, we establish two important results. First, completeness does not impose any new constraints on GP-preserving supermaps.

\begin{theorem}[$\mathbf{GPP}$ and $\mathbf{pGPP}$ supermaps are complete]\label{theo:ComplGPP}
    Completeness does not impose additional constraints on $\mathbf{GPP}$ and $\mathbf{pGPP}$:
    \begin{align}
        \mathbf{CGPP} &= \mathbf{GPP}, \\
        \mathbf{pCGPP} &= \mathbf{pGPP}.
    \end{align}
\end{theorem}
The proof can be found in App.~\ref{app:completenessGPP}. This result shows that thermality of $\mathcal{S}$ is already sufficient for complete GP-preservation and the classes $\mathbf{GPP}$ and $\mathbf{pGPP}$ are not reduced beyond ordinary equilibrium preservation by this requirement. We emphasize that, again, causal ordering must be imposed manually and does not follow from completeness. This is not surprising though; GP-preservation only means that $\mathcal{S}$ preserves the probabilistic structure of GP maps, which is too weak to automatically render $\mathcal{S}$ a superchannel. 

In contrast, imposing completeness on GPTP-preserving supermaps has a significantly stronger consequence. While ordinary GPTP-preservation admits higher-order transformations that are neither thermal nor causally ordered, requiring complete preservation eliminates these pathologies entirely. In particular, complete GPTP-preservation forces supermaps to be both thermal and causally ordered:

\begin{theorem}[Completeness enforces thermality and causal order]\label{theo:ClassesCollapse}
    The set of completely GPTP-preserving supermaps coincides with the set of thermal and causally ordered supermaps (see Fig.~\ref{fig::classes}). That is, we have
    \begin{align}
        \mathbf{CGPTPP} = \mathbf{pCGPTPP} = \mathbf{pCGPP} = \mathbf{pGPP}.
    \end{align}
\end{theorem}

\begin{figure*}
\centering
    \includegraphics[width = 0.7\linewidth]{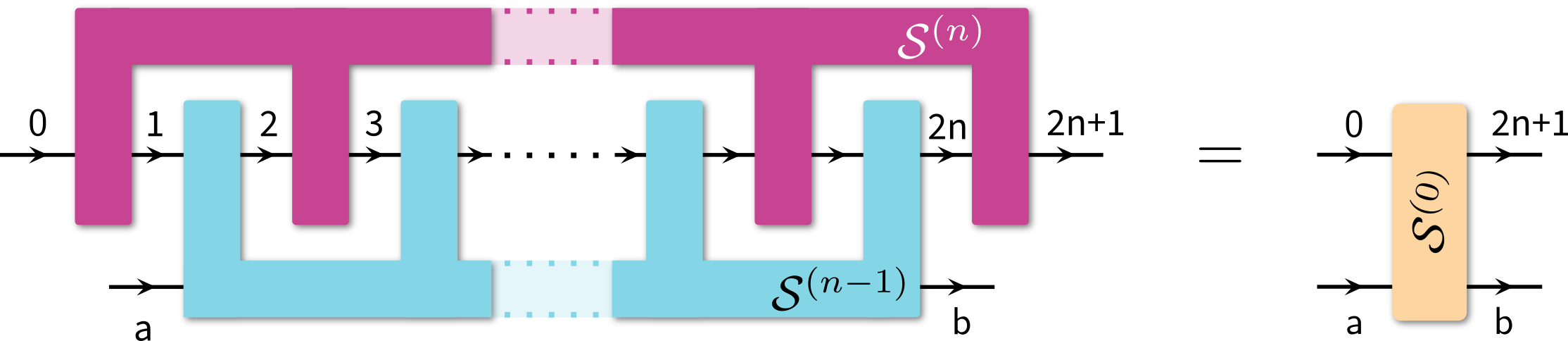}
    \caption{A $\mathbf{CGPTPP}^{(n)}$ supermap $\mathcal{S}^{(n)}$ maps any $\mathbf{CGPTPP}^{(n-1)}$ supermap $\mathcal{S}^{(n-1)}$ into a GPTP map $\mathcal{S}^{(0)}$ even when only acting on a part of it. We emphasize that the ordering of the respective supermaps, suggested by the arrangement of the respective spaces, is not a priori assumed, but follows from complete equilibrium preservation.}
    \label{fig::multi-slot_CGPP}
\end{figure*}
We provide the proof of this Theorem in App.~\ref{app:completenessCausal}. As a result of it, we see that the completeness requirement collapses the classes $\mathbf{GPTPP}$, $\mathbf{pGPTPP}$ into the single class $\mathbf{pGPP}$ of causally ordered equilibrium superchannels from which no non-equilibrium resources can be extracted. Crucially, both thermality and causal order now \textit{follow} from the principle of complete equilibrium preservation rather than being imposed independently. In turn, this makes causal ordering a property that is implied by thermodynamic principles, and renders causal indefiniteness incompatible with this notion of equilibrium. Thm.~\ref{theo:ClassesCollapse} identifies $\mathbf{pGPP}$ -- or, equivalently, $\mathbf{CGPTPP}$ and $\mathbf{pCGPTPP}$ -- as the main candidate for equilibrium supermaps, satisfying, both axiomatically as well as in terms of its causal and thermodynamic properties. From now on, we thus predominantly focus on this class of superchannels as well as its extension to the multi-time case.  

\subsection{Representation of completely GPTP-preserving superchannels}
Above, we have singled out $\mathbf{pGPP}$ as the unique class of deterministic equilibrium supermaps. It is therefore natural to ask how such transformations can be \textit{realized} physically and connect their thermodynamic properties with a concrete physical implementation. Since every element of $\mathbf{pGPP}$ is a superchannel, it admits a circuit implementation in terms of an encoder $\mathcal E$ and decoder $\mathcal D$ (Fig.~\ref{fig::enc_dec}). The question then becomes: what thermodynamic constraints must $\mathcal{E}$ and $\mathcal{D}$ satisfy, and in particular, must they be Gibbs preserving themselves? 

To answer this question, we first note that the quantum channels $\mathcal{E}$ and $\mathcal{D}$ are GPTP if 
\begin{gather}
    \mathcal{E}[\tau_P] = \tau_I \otimes \tau_c \quad \text{and} \quad \mathcal{D}[\tau_O \otimes \tau_c] = \tau_F 
\end{gather}
holds, where $\tau_c$ is the Gibbs state on the auxiliary space $\mathcal{H}_c$. Now, we first show that a superchannel $\mathcal{S}$ made up from $GPTP$ channels $\mathcal{E}$ and $\mathcal{D}$ automatically lies in $\mathbf{pGPP}$ (see App.~\ref{app:encdecpGPPproof} for a proof):

\begin{theorem}[Superchannels from GPTP maps]
\label{theo:encdecpGPP}
    Let $\mathcal{S}$ be a superchannel with $\mathcal{S}[\mathcal{T}] = \mathcal{D} \circ \mathcal{T} \circ \mathcal{E}$. Then $\mathcal{S} \in \mathbf{pGPP}$ if channels $\mathcal{E} \in \mathsf{GPTP}(P, I\otimes c)$ and $\mathcal{D} \in \mathsf{GPTP}(O\otimes c, F)$.
\end{theorem}
Conversely, given a pGPP superchannel $\mathcal{S}$, we can also characterize the structure of the channels $\mathcal{E}$ and $\mathcal{D}$ appearing in its decomposition (see App.~\ref{app:decomppGPPproof} for a proof).
\begin{theorem}[Decomposition of a $\mathbf{pGPP}$ supermap]
\label{theo:decomppGPP}
    Let $\mathcal{S} \in \mathbf{pGPP}$, with $\mathcal{S}[\mathcal{T}] = \mathcal{D} \circ \mathcal{T} \circ \mathcal{E}$. Then $\mathcal{D} \in \mathsf{GPTP}(O\otimes c, F)$ and the channel $\mathcal{E}: \mathcal{L}(\mathcal{H}_P) \rightarrow \mathcal{L}(\mathcal{H}_I \otimes \mathcal{H}_c)$ is \emph{locally} Gibbs-preserving, i.e., $\mathcal{E}[\tau_{P}] = \xi_{Ic}$, with
    \begin{gather}
        \operatorname{Tr}_I[\xi_{Ic}] = \tau_{c} \quad \text{and} \quad \operatorname{Tr}_c[\xi_{Ic}] = \tau_{I}.
    \end{gather}
\end{theorem}

Taken together, these results provide sufficient and necessary conditions for a superchannel to belong to $\mathbf{pGPP}$, connecting it with the Gibbs-preservation properties of the channels that make it up. We conjecture that Thm.~\ref{theo:decomppGPP} can be further strengthened, namely that a superchannel belongs to $\mathbf{pGPP}$ if and only if it admits a decomposition consisting entirely of Gibbs-preserving channels. We leave this as an open problem.

\section{Multi-slot equilibrium transformations}
\label{sec::multi-Slot}
The case of superchannels considered above shows that complete equilibrium preservation singles out thermal and causally ordered superchannels. As discussed in Sec.~\ref{subsec:higherorder} (see Fig.~\ref{fig::multislot}), such \textit{one-slot} supermaps only form a first level in the hierarchy of higher-order transformations. This raises the question whether the observed relationship between equilibrium preservation and causal ordering persists beyond the simplest nontrivial levels of the hierarchy. We therefore now turn to general higher-order transformations acting on multiple slots and derive the structural properties imposed by complete equilibrium preservation.

\begin{figure*}
 \centering
 \includegraphics[width=0.98\linewidth]{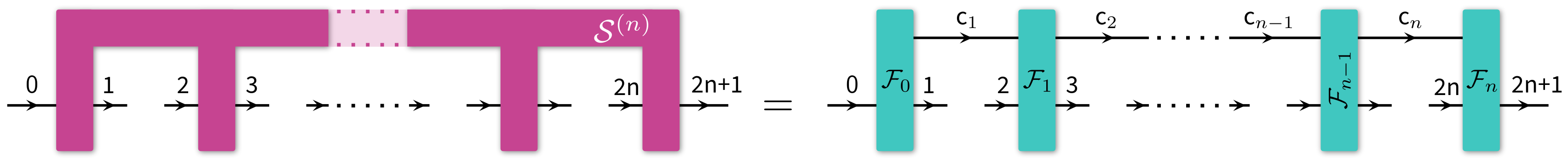}
 \caption{\textbf{Causally ordered $n$-slot supermaps.} An $n$-slot supermap $\mathcal{S}^{(n)}$ is causally ordered with causal order $0\prec 1 \prec \dots \prec 2n \prec 2n+1$ if its action can be decomposes as a sequence of $n+1$ CPTP maps $\mathcal{F}_i: \mathcal{L}(\mathcal{H}_{2i} \otimes \mathcal{H}_{c_i}) \rightarrow \mathcal{L}(\mathcal{H}_{2i+1} \otimes \mathcal{H}_{c_{i+1}})$$ \rightarrow$ for $i = 0, 1, \dots, n$, with $\mathcal{H}_{c_0} \cong \mathcal{H}_{c_{n+1}} \cong \mathbbm{C}$.}
 \label{fig::caus_ord}
\end{figure*}

\subsection{Complete equilibrium transformation for \texorpdfstring{$n$}{n}-slot supermaps}
Operationally, higher-order transformations act as devices with slots that can be filled by lower-order transformations. Filling the slots of a higher-order transformation with suitable lower-order objects produces a new object, whose type depends on the particular higher-order structure under consideration. The exact set of multi-slot transformations one obtains then depends on the respective physical constraints one imposes on them. Consequently, there is no \textit{unique} notion of a higher-order transformation. Rather, different operational closure requirements give rise to different classes of multi-slot objects. For example, requiring of a two-slot transformation to map all superchannels onto CPTP maps yields the set of two-slot combs~\cite{chiribella_quantum_2008}, which display a fixed causal order. By contrast, only demanding that two-slot supermaps map pairs of CPTP maps onto CPTP maps yields the larger class of process matrices~\cite{oreshkov_quantum_2012, Araujo2015}, which also contains causally indefinite processes. 

Here, to obtain a characterization of \textit{equilibrium} $n$-slot supermaps, we first drop the distinction between the different notions of equilibrium-preserving supermaps discussed above. Indeed, Sec.~\ref{subsec:completeSupermaps} showed that complete GPTP preservation singles out the class $\mathbf{pGPP}$ of thermal and causally ordered supermaps, while alternative notions collapse to the same set under the requirement of completeness. We therefore invoke the principle of complete equilibrium preservation as the unique defining property of $n$-slot equilibrium transformations and build the higher-order hierarchy starting from the set of Gibbs- and trace-preserving channels $\mathsf{GPTP} =: \mathbf{CGPTPP}^{(0)}$. This provides us with a recursive characterization of higher-order equilibrium transformations. Starting from Gibbs-preserving channels (i.e., zero-slot equilibrium transformations) at the bottom of the hierarchy, we can successively identify multi-slot equilibrium transformations. The first step of this process has already been carried out in Sec.~\ref{subsec:completeSupermaps}, which derived $\mathbf{pGPP}$ as the unique class of one-slot equilibrium transformations that preserve Gibbs-preserving channels in a complete sense. The next step of the hierarchy then consists of two-slot transformations $\mathcal{S}^{(2)}$ that map one-slot equilibrium transformations onto GPTP maps. Requiring the same closure property at all higher orders then naturally leads to a recursive definition of equilibrium $n$-slot supermaps $\mathcal S^{(n)}$ as those supermaps that map \textit{all} $n-1$ equilibrium supermaps onto equilibrium channels (see Fig.~\ref{fig::multi-slot_CGPP}). For notational convenience, in what follows we label the systems appearing in $n$-slot transformations by arabic numerals $\{0,1,\dots, 2n, 2n+1\}$ instead of $\{I, O, F, P\}$ used in the previous section; while the ordering of these labels suggests a causal ordering of the corresponding supermaps, no such structure is assumed a priori. With this, we obtain the following definition of the set $\mathbf{CGPTPP}^{(n)}$ of $n$-slot equilibrium supermaps (see Fig.~\ref{fig::multi-slot_CGPP}):

\begin{definition}[$n$-slot equilibrium supermaps]
\label{def::n_slot_eq}
    Let $n\in\mathbb{N}$, with $\mathbf{CGPTPP}^{(0)} = \mathsf{GPTP}$. An $n$-slot supermap $\mathcal{S}^{(n)}$ belongs to $\mathbf{CGPTPP}^{(n)}(0,\ldots,2n+1)$ if, for every $(n-1)$-slot supermap ${\mathcal{S}}^{(n-1)} \in \mathbf{CGPTPP}^{(n-1)}(1\otimes a,2,\ldots,2n-1,2n\otimes b)$, it holds that
    \begin{equation}
        (\mathcal{S}^{(n)} \otimes \mathfrak{id}_{a,b})[\mathcal{S}^{(n-1)}] \in \mathsf{GPTP}(0\otimes a,(2n+1)\otimes b),
    \end{equation}
    where $\mathfrak{id}_{a,b}$ denotes the identity supermap with respect to the auxiliary spaces $\mathcal{H}_a$ and $\mathcal{H}_b$. 
\end{definition}

\begin{remark}[Minimal equilibrium preservation requirement]
The above Definition imposes the weakest recursive notion of equilibrium preservation: an $n$-slot equilibrium transformation must map every $(n-1)$-slot equilibrium transformation to an equilibrium channel. One could instead require preservation of equilibrium objects at all lower levels of the hierarchy. As we show below, the weaker requirement already enforces both thermality and causal order, so stronger closure requirements do not lead to a smaller class of equilibrium transformations.
\end{remark}
We reiterate that the labels $0,1, \dots, 2n, 2n+1$ merely distinguish the different spaces $\mathcal{S}^{(n)}$ acts on and should not be interpreted as an assumed causal ordering. The definition only singles out $0$ and $2n+1$ as the global input and output of the resulting GPTP map, and thus as the global past and future of $\mathcal{S}^{(n)}$. We show below that complete equilibrium preservation itself enforces both thermality of $\mathcal{S}^{(n)}$ as well as the complete overall ordering $0\prec 1 \prec \cdots \prec 2n \prec 2n+1$. 

In Def.~\ref{def::n_slot_eq}, we deliberately only attach auxiliary systems to the boundary spaces of the inserted $(n-1)$-slot supermap, i.e., $\mathcal S^{(n-1)} \in \mathbf{CGPTPP}^{(n-1)} (1\otimes a,\ldots,2n\otimes b)$. More general choices with independent auxiliary systems attached to each slot may equally be considered. Under the \textit{consistency assumption} that auxiliary input (output) systems remain input (output) systems of the resulting GPTP map, these can be grouped into the effective auxiliary systems $a$ and $b$ appearing in the Definition (see Fig.~\ref{fig::aux_ident} for a graphical representation).

\begin{remark}[Relation to quantum combs] For trivial auxiliary spaces, replacing the target set $\mathsf{GPTP}$ in Def.~\ref{def::n_slot_eq} by the set of all CPTP maps recovers the standard definition of an $n$-slot quantum comb, i.e., causally ordered supermaps. Since here, we start from the strictly smaller base set $\mathbf{CGPTPP}^{(0)} = \mathsf{GPTP}$ one might (correctly) expect  resulting class of higher-order transformations to be strictly larger than the set of quantum combs. However, we require equilibrium preservation in a \textit{complete} sense, which is restrictive enough to enforce both causal ordering as well as thermality of the resulting $n$-slot supermaps.
\label{rem::rel_qc}
\end{remark}

\subsection{Causally ordered and thermal \texorpdfstring{$n$}{n}-slot supermaps}
With the notion of complete equilibrium preservation in hand, we can now show that this principle suffices to enforce causal ordering and thermality on $\mathcal{S}^{(n)}$. As mentioned in Rem.~\ref{rem::rel_qc}, the subset of $n$-slot supermaps with a fixed causal order (i.e., quantum combs) can be obtained recursively by taking CPTP maps as zero-slot combs and requiring that every $n$-slot comb maps all $(n-1)$-slot combs to a CPTP map. Equivalently, quantum combs are exactly those supermaps that can be obtained from an underlying quantum circuit~\cite{Chiribella2009, Bisio2011, Bisio2019}; we have already seen this for the case of superchannels, which can always be decomposed in terms of an encoder $\mathcal{E}$ and a decoder $\mathcal{D}$ (see Fig.~\ref{fig::enc_dec}). This notion of causal ordering of a supermap in terms of its underlying circuit representation is the one that we employ throughout:

\begin{definition}[Quantum combs/causally ordered $n$-slot supermaps]
    A supermap $\mathcal{S}^{(n)}$ is causally ordered $0\prec 1 \prec \cdots \prec 2n \prec 2n+1$ if it admits a decomposition in terms of $n+1$ CPTP maps $\{\mathcal{F}_i\}_{i=0}^{i=n}$ as in Fig.~\ref{fig::caus_ord}.
\end{definition}
We emphasize that such an order is not assumed in the definition of equilibrium supermaps, but follows as a consequence of complete equilibrium preservation (see Thm.~\ref{prop::multi-slot} below). The notion of thermality, already discussed for one-slot supermaps in Sec.~\ref{subsubsec::GPP_pGPP} (see also Fig.~\ref{fig::charGPP}) generalizes similarly: thermal $n$-slot supermaps are exactly those that do not allow for the extraction of athermality; that is, inserting Gibbs states into $\mathcal{S}^{(n)}$ only yields Gibbs states and equilibrium supermaps, but no athermal resources. Similar to the one-slot case, an $n$-slot supermap contains multiple interfaces at which Gibbs states may be inserted. Consequently, thermality is naturally defined relative to a specified causal ordering (in the one-slot case, this ordering was $P\prec I \prec O \prec F$, see Fig.~\ref{fig::charGPP}):

\begin{definition}[Thermal $n$-slot supermaps]
\label{def::thermal_n_slot_main}
We call an $n$-slot supermap thermal (with respect to the ordering $0\prec1\prec\cdots\prec2n+1$) if successively inserting the Gibbs states $\tau_0,\tau_2,\dots,\tau_{2n}$ along that order recursively produces Gibbs states and equilibrium supermaps of decreasing slot number (see Fig.~\ref{fig::thermRecursion}).
\end{definition}

\begin{figure}[t!]
\centering
    \includegraphics[width = 0.99\linewidth]{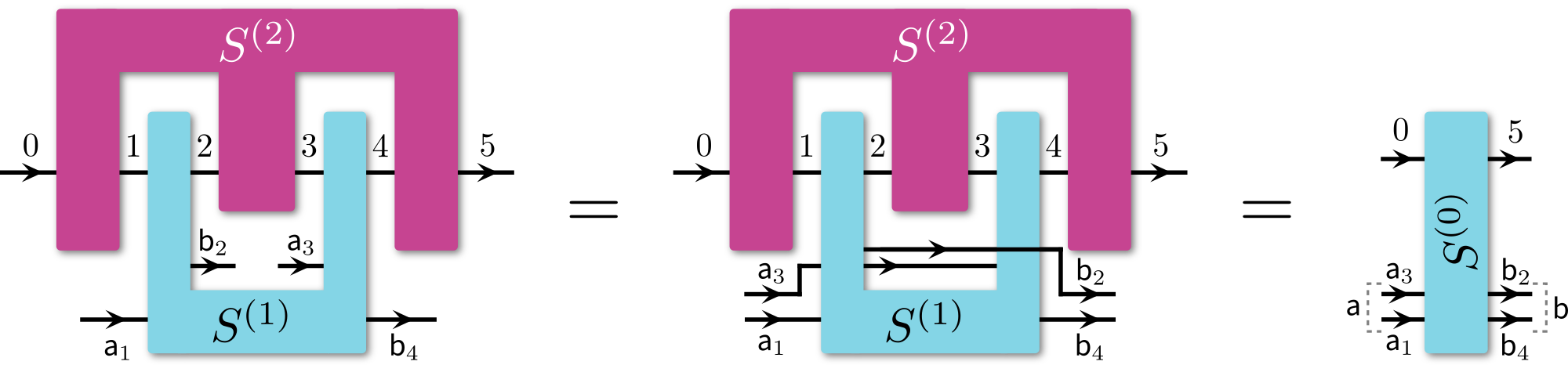}
    \caption{\textbf{Complete equilibrium preservation for $n$-slot supermaps.} For the definition of complete equilibrium preservation of $n$-slot supermaps, we assume throughout that all auxiliary input spaces $a_1, a_2, \dots$ and auxiliary output spaces $b_1, b_2, \dots$ become the input and output spaces of the resulting GPTP map. This allows us to group all auxiliary input (output) spaces into a single label $a$ ($b$).}
    \label{fig::aux_ident}
\end{figure}

With this, we can now state the central result of this section. Remarkably, the recursive notion of equilibrium supermaps derived solely from complete equilibrium preservation coincides exactly with the class of thermal and causally ordered $n$-slot supermaps:
\begin{theorem}
\label{prop::multi-slot}
    A supermap $\mathcal{S}^{(n)}$ lies in $\mathbf{CGPTPP}^{(n)}(0,\dots, 2n+1)$ iff it is causally ordered and thermal with respect to the ordering $0\prec 1\prec \cdots \prec 2n \prec 2n+1$. 
\end{theorem}
\edef\mainpropnum{\theprop}
The proof of this Theorem can be found in App.~\ref{app:multislot}. For $n=1$, it recovers Thm.~\ref{theo:ClassesCollapse} for the superchannel case. More generally, complete equilibrium preservation singles out thermal and causally ordered $n$-slot combs as the unique equilibrium transformations for all $n \geq 1$. We emphasize that, while equilibrium transformations and thermal, causally ordered combs both admit recursive definitions, their equivalence is not a priori evident. As already seen in the one-slot setting, this coincidence is a genuine consequence of \textit{complete} equilibrium preservation; without the completeness requirement, the corresponding notions of equilibrium-preserving higher-order transformations remain distinct. 

Thm.~\ref{prop::multi-slot} identifies equilibrium transformations at arbitrary order as precisely the thermal and causally ordered quantum combs. In particular, complete equilibrium preservation singles out a hierarchy of higher-order processes that neither generate athermality from equilibrium inputs nor allow it to be extracted from equilibrium systems. Having characterized the structure of equilibrium transformations, a natural next question is how to quantify athermality at the level of quantum processes themselves. In ordinary thermodynamics this role is played by the nonequilibrium free energy. Focussing on quantum channels as the next `elements up the hierarchy', we now show that the operational picture developed above naturally gives rise to corresponding free-energy-like quantities for quantum processes.

\begin{figure}[t!]
\centering
    \includegraphics[width = \linewidth]{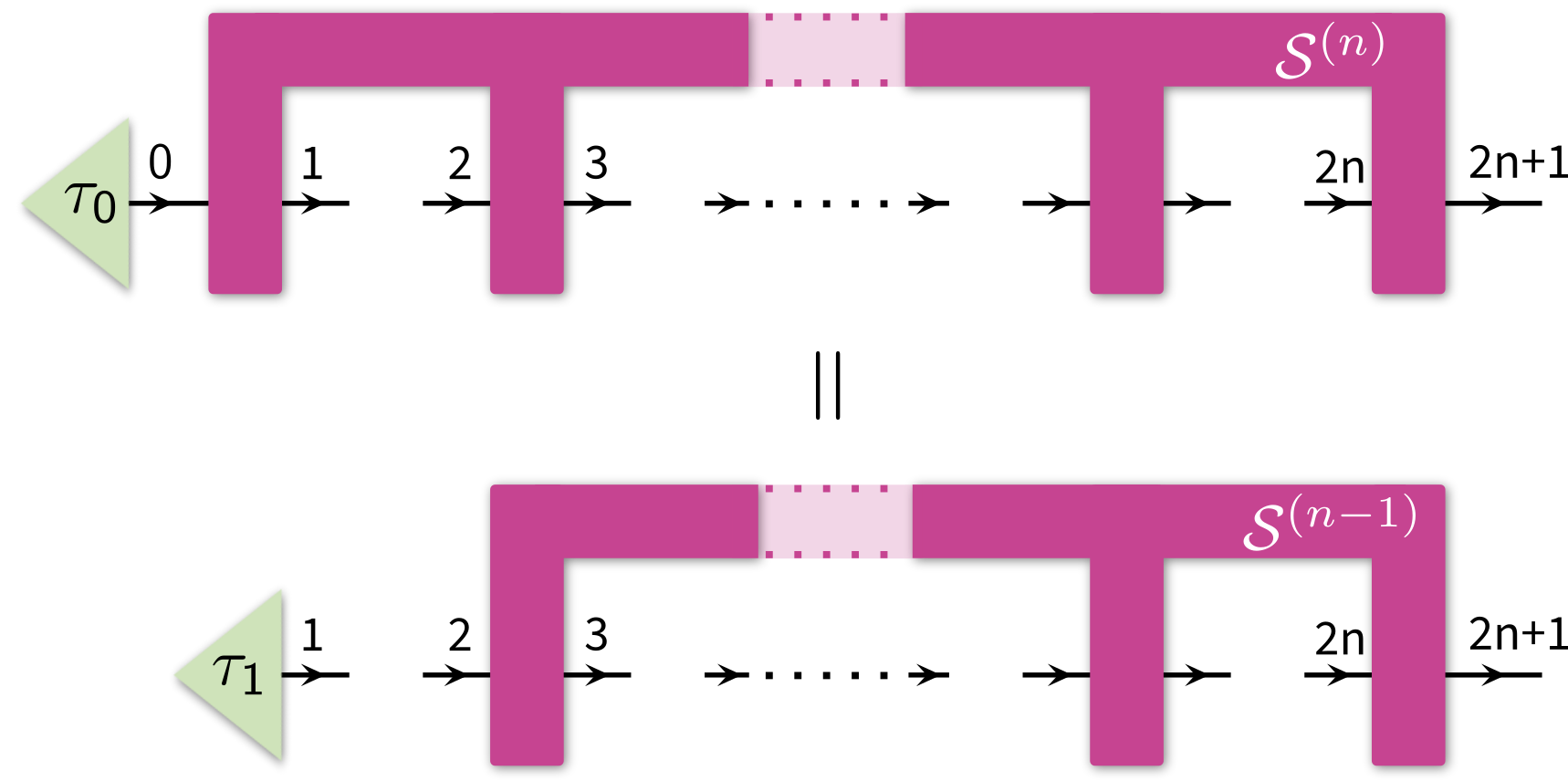}
    \caption{\textbf{Thermality of $n$-slot supermaps.} If $\mathcal{S}^{(n)}$ is thermal with respect to the ordering $0\prec 1\prec \cdots \prec 2n \prec 2n+1$, then inserting the Gibbs state $\tau_0$ into $\mathcal{S}^{(n)}$ yields the Gibbs state $\tau_1$ together with $(n-1)$-slot equilibrium supermap $\mathcal{S}^{(n-1)}$. Repeating this procedure, insertion of the Gibbs state $\tau_2$ into $\mathcal{S}^{(n-1)}$ would then yield the Gibbs state $\tau_3$ and an $(n-2)$-slot equilibrium supermap $\mathcal{S}^{(n-2)}$, and so on. Thus, equilibrium insertions recursively generate only Gibbs states and equilibrium supermaps of lower slot number.}
    \label{fig::thermRecursion}
\end{figure}

\section{Free energy of quantum channels}
\label{sec:monotones}
In standard quantum thermodynamics, the athermality of a quantum state $\rho$ is commonly quantified by the nonequilibrium free energy, which can be expressed in terms of a (pseudo-) distance from the equilibrium state $\tau$ as
\begin{equation}
    \mathcal F[\rho]=\beta^{-1} D(\rho\|\tau),
\end{equation}
where $D$ is the quantum relative entropy. Importantly, this quantity is monotone under Gibbs-preserving channels $\mathcal G$:
\begin{equation}
\label{eqn::monotonicity_free_en}
    \mathcal F[\mathcal G[\rho]]\leq \mathcal F[\rho].
\end{equation}
In thermodynamic state conversion, the free energy plays a distinguished operational role: If $\mathcal F[\rho] \geq \mathcal F[\rho']$ then $\rho$ can be transformed to $\rho'$ with arbitrary precision via an equilibrium transformation if one allows for a catalyst and correlations~\cite{muller_correlating_2018, rethinasamy_relative_2020,  Shiraishi2021, shiraishi_recovery_2025, shiraishi_quantum_2025}. More broadly, exact or more constrained catalytic transformations are characterized by \textit{families} of generalized free energies \cite{Ng2018}. Operationally, the free energy determines the average amount of work extractable from a quantum system in contact with a thermal bath at inverse temperature $\beta$ \cite{skrzypczyk2014work}. Recent results further connect free-energy- and work-extraction-based thermodynamic protocols to the certification of quantum correlations, including steering and entanglement \cite{biswas2025quantum, sarkar2025lost}.

\subsection{Free energy and equilibrium transformations}
Going beyond quantum states, it is natural to regard quantum processes themselves as thermodynamic resources. Having identified equilibrium transformations with thermal and causally ordered combs in the previous section, we now focus on quantum channels and ask how to quantify their athermality. In other words, we seek a free-energy-like quantity that extends the notion of nonequilibrium free energy from states to higher-order processes. A basic requirement for any such quantity $\mathtt{F}$ is that it be non-increasing under equilibrium transformations $\mathcal{S}$. Following the defining property of the state free energy [see~Eq.~\eqref{eqn::monotonicity_free_en}], we therefore require \begin{gather} 
\label{eqn::min_req_free}
\mathtt{F}[\mathcal{S}[\mathcal{T}]] \leq \mathtt{F}[\mathcal{T}], \end{gather} 
for every channel $\mathcal{T}\in\mathsf{CPTP}(I,O)$ and every equilibrium transformation $\mathcal{S}$. The monotonicity requirement in Eq.~\eqref{eqn::min_req_free} is closely related to conditions that have recently been used in Ref.~\cite{badhani_thermodynamics_2025a} to define free energies of quantum channels. The key difference is that, in the present work, monotonicity is imposed with respect to the class of equilibrium-preserving higher-order transformations identified in the previous sections, whereas Ref.~\cite{badhani_thermodynamics_2025a} considers a different underlying set of thermodynamically resourceless transformations.

The characterization obtained in Thm.~\ref{theo:ClassesCollapse} provides a natural operational interpretation of Eq.~\eqref{eqn::min_req_free}: a valid free energy of channels should not increase under arbitrary thermal and causally ordered higher-order processing. Since such transformations cannot generate athermality from equilibrium resources, this condition guarantees that $\mathtt F$ measures a genuine thermodynamic resource. Throughout this Section, unless stated otherwise, we only consider the set of equilibrium transformations $\mathfrak{F} = \mathbf{pGPP}$. This is the set of thermal and causally ordered superchannels identified in Sec.~\ref{subsec:completeSupermaps}. In addition, it contains the identity supermap and is closed under composition [see App.~\ref{app:compositionId}], which provides the natural closure properties expected of a set of thermodynamic transformations. Most of the results presented below remain valid if $\mathfrak F=\mathbf{pGPTPP}$ is chosen instead. For completeness, explicit proofs for both choices are provided in the Appendix, and we will indicate whenever a statement is specific to one of them.

\subsection{Gibbs-response free energy}
\label{sec::Gibbs_res}

A natural way to quantify the athermality of a quantum \textit{channel} is to extract a quantum \textit{state} from it via equilibrium higher-order processing and evaluate its free energy. The larger the free energy that can be obtained in this way, the greater the athermality of the underlying channel.

To formalize this intuition, we consider the subclass $\mathfrak{F}_{\rm TR}\subset \mathfrak{F}$ of equilibrium supermaps whose output is a trace-and-replace channel. Specifically, for every $\mathcal{S}\in \mathfrak{F}_{\rm TR}$ and every channel $\mathcal{T}\in \mathsf{CPTP}(I,O)$, the action of the channel $\mathcal{S}[\mathcal{T}]\in \mathsf{CPTP}(P,F)$ is given by 
\begin{equation}
    (\mathcal{S}[\mathcal{T}])[\rho_P] = \operatorname{Tr}[\rho_P] \eta_F^{\mathcal T}.
\end{equation}
That is, $\mathcal{S}[\mathcal{T}]$ replaces every input state $\rho_P$ by $\eta_F^{\mathcal T}\in \mathcal{L}(\mathcal{H}_F)$. This leads to the following quantity, which measures the maximal free energy that can be extracted from $\mathcal T$ in the form of a quantum state by equilibrium higher-order processing:
\begin{align}\label{eqn::free_energy}
    \nonumber \mathtt{F}_{\rm st}^{\mathfrak{F}}[\mathcal{T}] &:= \sup_{\mathcal{S} \in \mathfrak{F}_{\rm TR}}\mathcal{F}\Bigl[(\mathcal{S}[\mathcal{T}])[\rho_P]\Bigr] \\
    &= \sup_{\mathcal{S} \in \mathfrak{F}_{\rm TR}} \mathcal{F}[\eta^{\mathcal{T}}_F],
\end{align}
where $\rho_P$ is arbitrary. The following result reveals a more direct thermodynamic interpretation of $\mathtt{F}_{\rm st}^{\mathfrak{F}}[\mathcal{T}]$ in terms of the response of equilibrium-processed versions of $\mathcal T$ to Gibbs-state inputs: 
\begin{lemma}[$\mathtt{F}_{\rm st}^{\mathfrak{F}}$ as Gibbs-response free energy]
\label{lem:entirepGPP}
    Let $\mathcal{T} \in \mathsf{CPTP}(I,O)$. Then
    \begin{equation}\label{eq:opt_pGPP}
        \mathtt{F}_{\rm st}^{\mathfrak{F}}[\mathcal{T}] = \sup_{\mathcal{S} \in \mathfrak{F}} \mathcal{F}\Bigl[(\mathcal{S}[\mathcal{T}])[\tau_{P}]\Bigr].
    \end{equation}
    \begin{proof}
        See App.~\ref{app:sub:lemmaResp1}.
    \end{proof}
\end{lemma}

Lem.~\ref{lem:entirepGPP} shows that the extraction task that defines $ \mathtt{F}_{\rm st}^{\mathfrak{F}}$ can equivalently be written as an optimization over \textit{all} equilibrium transformations, provided that the processed channel is evaluated on the Gibbs state $\tau_P$. For this reason, we can interpret $\mathtt{F}_{\rm st}^{\mathfrak{F}}$ as a \emph{Gibbs-response free energy}: It measures the maximal athermality obtainable from $\mathcal{T}$ when it is subjected to equilibrium higher-order processing and probed with a Gibbs-state input.
 
Besides its operational interpretation, the quantity $\mathtt F_{\rm st}^{\mathfrak F}$ satisfies the key properties expected of a higher-order free energy: it is monotonic under equilibrium transformations and reduces to the standard state free energy for trace-and-replace channels. Moreover, for $\mathfrak F=\mathbf{pGPP}$ it admits a particularly simple closed form:

\begin{theorem}[Monotonicity and state free energy correspondence of $\mathtt{F}_{\rm st}^{\mathfrak{F}}$ and its closed form for $\mathbf{pGPP}$]\label{theo:GrespClosed}
    $\mathtt{F}_{\rm st}^{\mathfrak{F}}$ is monotonically non-increasing under transformations in $\mathfrak{F}$ and, for trace-and-replace channels
    \begin{equation}\label{eq:channelTR}
        \mathcal{T}_\eta[\rho] = \operatorname{Tr}[\rho]\eta,
    \end{equation}
    satisfies
    \begin{equation}
        \mathtt{F}_{\rm st}^{\mathfrak{F}}[\mathcal{T}_\eta] = \mathcal{F}[\eta] = D(\eta\|\tau).
    \end{equation}
    For $\mathfrak{F} = \mathbf{pGPP}$, it further admits the closed form
    \begin{equation}\label{eqn::simple_free_energy}
        \mathtt{F}_{\rm st}^{\mathbf{pGPP}}[\mathcal{T}] = \mathcal{F}[\mathcal{T}[\tau_I]].
    \end{equation}
    \begin{proof}
        See App.~\ref{app:sub:lemmaResp2}.
    \end{proof}
\end{theorem}

\subsection{Free energies from channel divergences}
\label{sec::free_en_div}
$\mathtt{F}_{\rm st}^{\mathfrak{F}}$ provides a natural extension of free energy to quantum channels. However, it probes a channel only through its action on a single state, the Gibbs state, and therefore cannot be expected to fully capture the athermality of a general quantum process. In ordinary thermodynamics, the free energy of a state can be expressed as its relative entropy from thermal equilibrium. A natural generalization is therefore to compare quantum channels directly with a thermodynamically resourceless channel. This leads to free-energy-like quantities based on channel divergences, which quantify distinguishability at the level of quantum processes themselves. In contrast to $\mathtt F_{\rm st}^{\mathfrak F}$, which captures the athermality accessible through equilibrium probing, channel divergences provide more refined measures of the distinguishability of a quantum process from thermal equilibrium.

Accordingly, we seek free-energy-like quantities obtained by replacing the state relative entropy in the standard free energy with suitable notions of channel distinguishability. Among the many possible channel divergences~\cite{Gour2021}, here we focus on two particularly important cases. First, the \emph{channel relative entropy}
\begin{equation}\label{eq:chanDiv}
    D_{\rm ch}(\mathcal{T}\|\mathcal{T}') = \sup_{\rho_{RI}} D\!\left( (\mathcal{I}_R \otimes \mathcal{T})[\rho_{RI}] \,\middle\|\, (\mathcal{I}_R \otimes \mathcal{T}')[\rho_{RI}] \right),
\end{equation}
which measures the maximal distinguishability between $\mathcal T$ and $\mathcal T'$ obtainable from their action on an arbitrary input state, possibly correlated with a reference system.

Second, the \emph{amortized channel relative entropy}
\begin{align}
    \nonumber D_{\rm am}(\mathcal{T}\|\mathcal{T}') &= \sup_{\rho_{RI},\sigma_{RI}} \Big[ D\bigl( (\mathcal{I}_R \otimes \mathcal{T})[\rho_{RI}] \,\big\|\, (\mathcal{I}_R \otimes \mathcal{T}')[\sigma_{RI}] \bigr) \\
    &\qquad - D(\rho_{RI}\|\sigma_{RI}) \Big], \label{eq:chanAmorDiv}
\end{align}
which measures the net distinguishability generated by the channels themselves, after removing any distinguishability already present in the inputs. Intuitively, the channel relative entropy captures one-shot channel distinguishability, whereas the amortized channel relative entropy incorporates arbitrary distinguishability already present in the inputs and therefore provides a more refined measure.

Following the intuition from the quantum \textit{state} case, we thus introduce free energy quantities for quantum channels via their relative entropy with respect to the completely thermalizing channel $\overline{\mathcal{G}}$ that replaces every input $\rho_I \in \mathcal{L}(\mathcal{H}_I)$ with the Gibbs state $\tau_O$, i.e., 
\begin{equation}\label{eq:channelTherm}
    \overline{\mathcal{G}}[\rho_I] = \operatorname{Tr}[\rho_I] \tau_{O}.
\end{equation}
With this, we obtain:
\begin{definition}[Free energies from channel divergences]\label{def:chanFreeEn}
Let $\mathcal{T} \in \mathsf{CPTP}(I,O)$. Then
\begin{align}
\label{eqn::chann_free}
    \mathtt{F}_{\rm ch}^{\mathfrak{F}}[\mathcal{T}] &:= \sup_{\mathcal{S}\in\mathfrak{F}} \beta^{-1}D_{\rm ch}\bigl(\mathcal{S}[\mathcal{T}]\|\overline{\mathcal{G}}\bigr),\\
    \text{and} \quad 
    \mathtt{F}_{\rm am}^{\mathfrak{F}}[\mathcal{T}] &:= \sup_{\mathcal{S}\in\mathfrak{F}} \beta^{-1}D_{\rm am}\bigl(\mathcal{S}[\mathcal{T}]\|\overline{\mathcal{G}}\bigr)
\end{align}
are the \textit{channel free energy} and \textit{amortized channel free energy} of $\mathcal{T}$, respectively.
\end{definition}
We note that the quantity $\beta^{-1}D_{\rm ch}(\mathcal T\|\overline{\mathcal G})$ without the optimization over equilibrium transformations was recently proposed as a notion of channel free energy in Ref.~\cite{badhani_thermodynamics_2025a}. Here, we obtain the following Theorem (see App.~\ref{app:freeEnHier} for a proof), establishing the monotonicity of $\mathtt{F}_{\rm ch}^{\mathfrak{F}}$ and $\mathtt{F}_{\rm am}^{\mathfrak{F}}$ as well as their relative ordering: 
\begin{theorem}
\label{thm::hier_monotones}
The quantities $\mathtt{F}_{\rm ch}^{\mathfrak{F}}$ and $\mathtt{F}_{\rm am}^{\mathfrak{F}}$ are monotonic under equilibrium superchannels $\mathcal{S} \in \mathfrak{F}$. In addition, they satisfy
\begin{equation}\label{eq:freeEnHier}
    \mathtt{F}_{\rm st}^{\mathfrak{F}}[\mathcal{T}] \leq \mathtt{F}_{\rm ch}^{\mathfrak{F}}[\mathcal{T}] \leq \mathtt{F}_{\rm am}^{\mathfrak{F}}[\mathcal{T}].
\end{equation}
\end{theorem}
The hierarchy~\eqref{eq:freeEnHier} reflects a progression from athermality accessible through equilibrium probing ($\mathtt F_{\rm st}^{\mathfrak F}$) to increasingly refined quantifiers based on channel-level distinguishability.

Having identified thermal and causally ordered combs as the free higher-order thermodynamic transformations, a natural next question concerns the interconversion of quantum processes under such operations. While the free-energy-like quantities introduced above provide useful monotones, they do not in general fully determine process convertibility. In App.~\ref{app::Convertibility}, we investigate this question in more detail. We show that broad classes of monotones based on Gibbs responses and fixed thermal references are insufficient to characterize higher-order thermodynamic conversions, but that complete characterizations can be obtained from suitably generalized divergence-based constructions. These results provide a first step towards a full theory of process conversion in higher-order thermodynamics and indicate that the convertibility structure is considerably richer than what can be captured by any single free-energy quantity.

\section{Conclusions and outlook}
\label{sec::conclusions}
In this work, we have extended the thermodynamic framework beyond quantum states and channels and derived the properties of thermodynamically admissible higher-order transformations. Specifically, we have established equilibrium preservation as a thermodynamic principle and showed that it leads to several inequivalent classes of equilibrium supermaps. In general, these candidate classes differ in their causal and thermodynamic properties. Requiring equilibrium preservation to hold \textit{completely} collapses these distinctions and uniquely singles out thermal and causally ordered superchannels. This is similar to the standard case of quantum channels, where it is the requirement of \textit{complete} positivity that renders them physical. Here, completeness of a thermodynamic property enforces the causal ordering of the respective supermaps, and ensures that no athermality can be extracted from them. We extended this characterization to \textit{arbitrary} higher-order transformations, showing that complete equilibrium preservation generically leads to thermal and causally ordered quantum combs. As a consequence, these results establish a direct link between causal order and thermodynamic principles, and identify causal indefiniteness as an out-of-equilibrium property. 

Building on this characterization, we introduced several free-energy-like quantifiers for quantum channels, established their monotonicity under equilibrium transformations, and derived a hierarchy between them. In addition, we investigated the extent to which such quantities constrain channel convertibility. Together, these developments provide a foundation for a systematic study of equilibrium, resourcefulness, and process interconversion at arbitrary levels of the hierarchy of quantum transformations.

The foundational and practical insights provided in this work open several directions for future research. While we focused on Gibbs-preserving processes, it would be interesting to investigate analogous constructions for more restrictive thermodynamic frameworks, such as thermal operations, and to determine whether causal order of the corresponding equilibrium transformations can likewise be derived from thermodynamic principles. 

More broadly, our results raise the question of whether thermal and causally ordered higher-order transformations can be derived from principles alternative to equilibrium preservation. Just as Gibbs states admit several independent characterizations, one may ask whether the same is true for higher-order equilibrium processes. Promising directions include extending thermal operations to higher-order transformations through suitable restrictions on their purifications, generalizing passivity and complete passivity from states to channels and higher-order maps, or introducing higher-order thermodynamic potentials whose minimizers coincide with equilibrium transformations. Should these approaches lead to the same class of equilibrium transformations, thermal and causally ordered higher-order processes would play a role analogous to that of the Gibbs state in ordinary thermodynamics. Conversely, a divergence between these constructions would indicate that higher-order thermodynamics requires additional principles beyond those needed to characterize equilibrium states.

Having established causal indefiniteness as an out-of-equilibrium property, a natural next step is to quantify the connection between causal structure and thermodynamic resourcefulness. While measures for causal indefiniteness -- like, e.g., causal robustness~\cite{Araujo2015} -- exist, their relation to the free-energy-like quantities introduced here remains unknown. Deriving it would furnish causal indefiniteness with a clear thermodynamic interpretation and reveal the thermodynamic value of causally indefinite correlations. 

Finally, the higher-order free energies introduced here raise both operational and conceptual questions. While they are naturally aligned with the equilibrium transformations identified in this work, their interpretation in terms of concrete thermodynamic tasks remains open. Recent works have established such operational meanings for divergence-based channel free energies, relating them to extractable work and resource interconversion between quantum channels~\cite{badhani_thermodynamics_2025a,badhani_thermodynamic_2025}. Despite the different equilibrium superchannels considered here, it is natural to conjecture that analogous interpretations can be obtained for the free-energy-like quantities introduced in this work.

On a more conceptual level, the role of free energy in state thermodynamics suggests a number of intriguing directions. In suitably relaxed catalytic settings, allowing for approximate transformations and correlations with catalysts, the nonequilibrium free energy becomes the unique monotone governing state convertibility~\cite{muller_correlating_2018,rethinasamy_relative_2020,Shiraishi2021,shiraishi_recovery_2025,shiraishi_quantum_2025}. It would therefore be particularly interesting to investigate whether an analogous collapse occurs in higher-order thermodynamics and whether one of the free energies introduced here acquires a distinguished role in determining the convertibility of quantum processes. More broadly, the framework developed in this work not only identifies a physically motivated notion of equilibrium for higher-order quantum transformations, but also provides a foundation for a fully fledged thermodynamics of general quantum processes beyond the traditional state-and-channel paradigm.
 
\section*{Acknowledgements}

S.M. acknowledges funding from the European Union's Horizon Europe research and innovation programme under the Marie Sk\l odowska-Curie grant agreement No. 101068332.

K.S. acknowledges: This research was funded in whole or in part by the Austrian Science Fund (FWF) 10.55776/PAT4559623. For open access purposes, the author has applied a CC BY public copyright license to any author-accepted manuscript version arising from this submission. 

T.G. is supported by the Slovak Research and Development Agency through Grant no. APVV-22-0570, the Scientific Grant Agency of the Ministry of Education, Slovak Republic through Grant no. VEGA 2/0128/24 and the \v{S}tefan Schwarz Support Fund 2025/OV1/046 by the Slovak Academy of Sciences.

G.C. acknowledges support from the Chinese Ministry of Science and Technology (MOST) through grant 2023ZD0300600, from the Hong Kong Research Grant Council (RGC) through grants 17310725 and SRFS2021-7S02, and from the State Key Laboratory of Quantum Information Technologies and Materials, Chinese University of Hong Kong. Research at the Perimeter Institute is supported by the Government of Canada through the Department of Innovation, Science and Economic Development Canada and by the Province of Ontario through the Ministry of Research, Innovation and Science.

\bibliographystyle{apsrev4-1}
\bibliography{main}

\appendix

\section{Transformations and Choi-Jamio\l kowski isomorphism}

\subsection{Choi-Jamio\l kowski isomorphism}
The characterization of higher-order transformations can be significantly simplified by exploiting the duality between linear maps and operators provided by the Choi-Jamio\l kowski isomorphism.
\begin{definition}[Choi operator]
    Let $\mathcal{T} \in \mathcal{L}(\mathcal{L}(\mathcal{H}_I), \mathcal{L}(\mathcal{H}_O))$ be a linear map. Then the Choi operator $T \in \mathcal{L}(\mathcal{H}_I \otimes \mathcal{H}_O)$ of $\mathcal{T}$ is
    \begin{equation}
        T = (\mathcal{I} \otimes \mathcal{T})[|\ident\rangle\!\rangle\langle\!\langle\ident|],
    \end{equation}
    where $\mathcal{I}$ denotes the identity map on $\mathcal{L}(\mathcal{H}_I)$ and $|\ident\rangle\!\rangle = \sum_{i=1}^{d_I} |ii\rangle$ is the (unnormalized) maximally entangled vector on $\mathcal{H}_I \otimes \mathcal{H}_I$, written in the canonical computational basis $\{|i\rangle\}_i$ of $\mathcal{H}_I$, with $d_I = \dim(\mathcal{H}_I)$.
\end{definition}
\noindent Note that, throughout, we consider the canonical basis for each space $\Hcal_X$ to be the eigenbasis of the corresponding Gibbs states $\tau_X$, respectively. Importantly, a map $\mathcal{T} \in \mathcal{L}(\mathcal{L}(\mathcal{H}_I), \mathcal{L}(\mathcal{H}_O))$ is CP iff its Choi operator $T\in \Lcal(\Hcal_I\otimes \Hcal_O)$ satisfies $T\geq 0$ and TP iff we have $\tr_O[T] = \ident_I$.\\
Due to the Choi-Jamio\l kowski isomorphism, supermaps $\mathcal{S}$ acting on quantum operations can be equivalently represented as maps acting on the corresponding Choi operators. Iterating this correspondence, an arbitrary linear supermap $\mathcal{S}: \mathcal{L}(\mathcal{L}(\mathcal{H}_I), \mathcal{L}(\mathcal{H}_O)) \rightarrow \mathcal{L}(\mathcal{L}(\mathcal{H}_P), \mathcal{L}(\mathcal{H}_F))$ can be represented by a Choi operator $S \in \mathcal{L}(\mathcal{H}_P \otimes \mathcal{H}_I \otimes\mathcal{H}_O \otimes\mathcal{H}_F)$. More generally, \textit{any} higher-order map can equivalently be represented by its Choi operator $M$, which is defined on all the input and output spaces of $M$. Throughout, we assume 
\begin{equation}\label{eq:condCP}
    M\geq 0    
\end{equation}
for all maps we consider, corresponding to the fundamental requirement of complete positivity preservation. 
\begin{remark}
For ease of notation, we will phrase most statements in this appendix in terms of Choi operators rather than the corresponding maps, and drop their explicit distinction whenever there is no risk of confusion both in the text as well as the figures. Throughout, we adopt the convention that maps are denoted by calligraphic letters and their Choi operators by the corresponding upright letter. This interchangeable use also extends to set inclusions, e.g., we will write both $\mathcal{S} \in \mathbf{GPP}$ and $S \in \mathbf{GPP}$ when we want to express that the supermap preserves Gibbs-preserving maps. 
\end{remark}

\noindent The concatenation/action of maps in terms of their Choi operators is given by the \emph{link product}:

\begin{definition}[Link product~\cite{Chiribella2009,taranto_higher-order_2025}]
    Let $A \in \mathcal{L}(\mathcal{H}_{\mathsf{I}})$ and $B \in \mathcal{L}(\mathcal{H}_{\mathsf{J}})$, where $\mathsf{I}$ and $\mathsf{J}$ are finite index sets and $\mathcal{H}_\mathsf{X} = \bigotimes_{i \in \mathsf{X}} \mathcal{H}_i$. The \emph{link product} $A \star B \in \mathcal{L}(\mathcal{H}_{\mathsf{I}\setminus\mathsf{J}} \otimes \mathcal{H}_{\mathsf{J}\setminus\mathsf{I}})$ is defined as
    \begin{equation}
        A \star B := \operatorname{Tr}_{\mathsf{I}\cap \mathsf{J}} \bigl[(\mathbbm{1}_{\mathsf{J}\setminus\mathsf{I}} \otimes B^{\mathrm{T}_{\mathsf{I}\cap\mathsf{J}}})(\mathbbm{1}_{\mathsf{I}\setminus\mathsf{J}} \otimes A)\bigr],
    \end{equation}
    where $T_{\mathsf{I}\cap\mathsf{J}}$ denotes the partial transpose over the subsystems in $\mathsf{I}\cap\mathsf{J}$.
\end{definition}

Operationally, the link product maps two Choi operators $A$ and $B$ to the Choi operator $A \star B$ corresponding to the concatenation of the associated linear maps $\mathcal{A}$ and $\mathcal{B}$. In particular, action of a map $\mathcal{T}$ on a quantum state $\rho$ is equivalently represented by
\begin{equation}\label{eq:app:mapAction}
    \mathcal{T}[\rho] = T \star \rho.
\end{equation}
Similarly, the Choi operator of the map $\mathcal{T}' \in \mathcal{L}(\mathcal{L}(\mathcal{H}_P), \mathcal{L}(\mathcal{H}_F)) $ resulting from the action of a supermap $\mathcal{S}: \mathcal{L}(\mathcal{L}(\mathcal{H}_I), \mathcal{L}(\mathcal{H}_O)) \rightarrow \mathcal{L}(\mathcal{L}(\mathcal{H}_P), \mathcal{L}(\mathcal{H}_F))$ on a channel $\mathcal{T} \in  \mathcal{L}(\mathcal{L}(\mathcal{H}_I) , \mathcal{L}(\mathcal{H}_O))$ is given by 
\begin{gather}
    T'_{PF} = (S_{PIOF} \star T_{IO}) \in \mathcal{L}(\mathcal{H}_P \otimes \mathcal{H}_F),
\end{gather}
where we have added subscripts to denote the spaces the respective operators are defined on. The link product reduces to the tensor product $A \otimes B$ when $A$ and $B$ act on disjoint Hilbert spaces ($\mathsf{I} \cap \mathsf{J} = \emptyset$), and to the trace $\operatorname{Tr}[B^{T}A]$ when they act on the same Hilbert space ($\mathsf{I} = \mathsf{J}$). In addition, the link product is associative, i.e., 
\begin{gather}
    A\star (B \star C) = (A \star B) \star C = A \star B\star C
\end{gather}
and -- for all the cases we consider -- commutative, i.e., $A \star B = B\star A$.

\subsection{Characterization of maps and supermaps}

Representing supermaps in terms of their Choi operators defined on all involved input and output Hilbert spaces $\mathcal{H}_{\rm i}$ and $\mathcal{H}_{\rm o}$ allows for a compact characterization of the relevant classes of higher-order transformations in terms of projection operators~\cite{Araujo2015}. 
\begin{definition}[Projection operator]\label{def:proj}
Let $M \in \mathcal{L}(\mathcal{H})$. Then a self-adjoint, unital projection operator ${}_{X}\bullet$ that commutes with transposition is defined by
   \begin{gather}\label{eq:projMMS}
        {}_{X}M := \tr_X[M] \otimes \ident_X/d_X,
    \end{gather}
where $d_X = \operatorname{dim}(\mathcal{H}_X)$.
\end{definition}
We will employ this projector -- or variants thereof -- throughout to characterize the various sets of maps we consider. Some of its properties will prove particularly helpful:
\begin{lemma}
    The operator $L_X[\bullet] := {}_X\bullet = \tr_X[\bullet] \otimes \ident_X/d_X$ 
    \begin{itemize}
    \item is a projector, i.e., $L_X^2 = L_X$,
        \item is \textit{unital}, i.e., $L_X[\ident] = \ident$,
    \item is \textit{self-adjoint}, i.e., $\tr[(L_X[M])^\dagger N] =  \tr[M^\dagger L_X[N]]$ for all $M,N$,
    \item is \textit{trace preserving}, i.e., $\tr[L_X[M]] = \tr[M]$ for all $M$, 
    \item \textit{commutes with transposition}, i.e., $L_X[M^\mathrm{T}] = L_X[M]^\mathrm{T}$ for all $M$,
    \item and satisfies $L_X \circ L_Y$ = $L_Y \circ L_X$ for all $X,Y$.
    \end{itemize} 
\end{lemma}
\begin{proof}
The proof follows from direct insertion.
\end{proof}
We note that \textit{any} projector $L = \sum_i c_i L_{X_i}$ with $c_i \in \mathbbm{R}$ inherits the properties of unitality, self-adjointness and commutation with the transposition, as well as trace preservation for the case $\sum_i c_i =1$. 

Using this projector, we can, for example, characterize the set of CPTP maps characterized in the following way.
\begin{prop}[Characterization of CPTP maps]
\label{prop::TP}
    An operator $T \in \mathcal{L}(\mathcal{H}_I \otimes \mathcal{H}_O)$ is the Choi operator of a CPTP map $\mathcal{T}: \mathcal{L}(\mathcal{H}_I) \rightarrow  \mathcal{L}(\mathcal{H}_O)$ iff 
    \begin{align}
        \label{eqn::TP1} T &\geq 0, \\ 
        \label{eqn::TP2} \tr[T] &= d_I, \\
        \label{eqn::TP_proj} T &= T - {}_OT + {}_{IO}T =: \mathcal{P}_{\textup{CPTP}}[T].
    \end{align}
    \begin{proof}
        ($\Rightarrow$) Condition \eqref{eqn::TP1} follows from the fact that every CPTP map is CP.
    
        A map $\mathcal{T}$ is trace-preserving if and only if its Choi operator satisfies
        \begin{equation}
            \tr_O[T] = \ident_I.
        \end{equation}
        Taking the trace of both sides gives $\tr[T]=d_I$, which yields condition \eqref{eqn::TP2}.

        Furthermore,
        \begin{align}
            \nonumber {}_O T &= \tr_O[T] \otimes \ident_O/d_O \\
            \nonumber &= \ident_I \otimes \ident_O/d_O  \\
            &= \ident_{IO}/d_O, \\
            \nonumber {}_{IO}T &= \tr_{IO}[T] \otimes \ident_{IO}/(d_I d_O) \\
            \nonumber &= d_I \,\ident_{IO}/(d_I d_O) \\
            &= \ident_{IO}/d_O.
        \end{align}
        Hence ${}_O T = {}_{IO}T$, which implies
        \begin{equation}
            T - {}_O T + {}_{IO}T = T,
        \end{equation}
        and, therefore, condition \eqref{eqn::TP_proj} holds. 
    
        ($\Leftarrow$) Conversely, if condition \eqref{eqn::TP_proj} holds, then ${}_O T = {}_{IO}T$, which implies that $\tr_O[T]$ is proportional to the identity on $\mathcal{H}_I$. Using \eqref{eqn::TP2}, we obtain
        \begin{equation}
            \tr_O[T] = \ident_I,
        \end{equation}
        and therefore $\mathcal{T}$ is trace-preserving. Together with \eqref{eqn::TP1}, this implies that $\mathcal{T}$ is CPTP.
    \end{proof}
\end{prop}
We note that the operator $\mathcal{P}_{\text{CPTP}}$ defined in Eq.~\eqref{eqn::TP_proj} is a projector, i.e.,
\begin{equation}
    \mathcal{P}_{\text{CPTP}} \circ \mathcal{P}_{\text{CPTP}} = \mathcal{P}_{\text{CPTP}},
\end{equation}
which can be verified by direct substitution. 

The characterization of maps in terms of projectors and trace constraints as in the above Lemma is a convenient tool to deduce the properties of the supermaps we introduce. To this aim, we first recall a result on the characterization of linear transformations between affine spaces in terms of projectors~\cite{Milz2024}.
\begin{lemma}[Projective characterization of transformations -- specialised version]
\label{lem::SpecProj}
    Let $\mathcal{P}^{\rm i}:\mathcal{L}(\mathcal{H}_{\rm i})\to\mathcal{L}(\mathcal{H}_{\rm i})$ and $\mathcal{P}^{\rm o}:\mathcal{L}(\mathcal{H}_{\rm o})\to\mathcal{L}(\mathcal{H}_{\rm o})$ be linear projectors, corresponding to the input and output spaces, respectively. Consider the affine subspaces
    \begin{align}
        \mathsf{A}_{\rm i} &:= \left\{ W \in \mathcal{L}(\mathcal{H}_{\rm i}) \;\middle|\; \mathcal{P}^{\rm i}[W]=W,\; \operatorname{Tr}[W]=\gamma_{\rm i} \right\}, \\ 
        \mathsf{A}_{\rm o} &:= \left\{ W' \in \mathcal{L}(\mathcal{H}_{\rm o}) \;\middle|\; \mathcal{P}^{\rm o}[W']=W',\; \operatorname{Tr}[W']=\gamma_{\rm o} \right\}.
    \end{align}
    If $\mathcal{P}^{\rm i}:\mathcal{L}(\mathcal{H}_{\rm i})\to\mathcal{L}(\mathcal{H}_{\rm i})$ and $\mathcal{P}^{\rm o}:\mathcal{L}(\mathcal{H}_{\rm o})\to\mathcal{L}(\mathcal{H}_{\rm o})$ are unital, self-adjoint, and commute with the transposition, then, for $\gamma_{\rm i}\neq 0$, a linear map $\mathcal{M}:\mathcal{L}(\mathcal{H}_{\rm i})\to\mathcal{L}(\mathcal{H}_{\rm o})$ satisfies
    \begin{equation}
        \mathcal{M}[\mathsf{A}_{\rm i}] \subseteq \mathsf{A}_{\rm o}
    \end{equation}
    if and only if its Choi operator $M \in \mathcal{L}(\mathcal{H}_{\rm i}\otimes\mathcal{H}_{\rm o})$ satisfies
    \begin{align}
    \notag
        &M = M - \mathcal{P}^{\rm i}[M] + \bigl(\mathcal{P}^{\rm i} \otimes \mathcal{P}^{\rm o}\bigr)[M] - \mathcal{P}^{\rm i}[{}_{\rm o}M ] + {}_{\rm io}M\\
        &\phantom{M}=: \mathcal{P}^{\rm io}[M], \label{eqn::special_proj1} \\
        &\tr[M] = \frac{\gamma_{\rm o}}{\gamma_{\rm i}}d_{\rm i},  \label{eqn::special_proj2}
\end{align}
where $d_{\rm i} = \text{dim}(\mathcal{H}_{\rm i})$.
    \begin{proof}
        See Theorem~2 in Ref. \cite{Milz2024}.
    \end{proof}
\end{lemma}

Following the above lemma, the recipe to characterize the (Choi states of the) supermaps we consider in this manuscript will be to first find the respective projectors onto their input and output spaces together with the corresponding trace constraints, and then use Lem.~\ref{lem::SpecProj}(or Lem.~\ref{lem::FullyGenProj}, see below) for the supermap characterization. As an example, we demonstrate this approach here to find the (well-known) characterization of proper supermaps/superchannels~\cite{Araujo2015}:

\begin{prop}[Projective characterization of superchannels]
    \label{prop::supchan}
    An operator $S \in \mathcal{L}(\mathcal{H}_P \otimes \mathcal{H}_I \otimes \mathcal{H}_O \otimes \mathcal{H}_F)$ is the Choi operator of a proper supermap (superchannel) $\mathcal{S}: \mathsf{CPTP}(I, O) \rightarrow \mathsf{CPTP}(P, F)$ iff 
    \begin{align}
        S &\geq 0, \label{eqn::UTP1} \\
        \tr[S] &= d_P d_O, \label{eqn::UTP2} \\
        S &= S - {}_F S + {}_{OF} S - {}_{IOF} S + {}_{PIOF} S =: \mathcal{P}_{\textup{SC}}[S]. \label{eqn::UTP_proj}
    \end{align}
    \begin{proof}
        The set of proper (one-slot) supermaps $\mathcal{S}$ is set of completely positive linear transformations that map every channel $\mathcal{T}\in \mathsf{CPTP}(I,O)$ to a channel $\mathcal{S}[\mathcal{T}]\in \mathsf{CPTP}(P,F)$. First, we note that complete positivity is equivalent to $S\geq 0$. By Proposition \ref{prop::TP}, the Choi operators of channels from $I$ to $O$ are characterized by
        \begin{align}
            T &\geq 0, \\
            \operatorname{Tr}[T] &= d_I, \\
            T &= T - {}_{O}T + {}_{IO}T =: \mathcal P_{\rm CPTP}^{\rm i}[T].
        \end{align}
        Similarly, the Choi operators of channels from $P$ to $F$ are characterized by
        \begin{align}
            T' &\geq 0, \\
            \operatorname{Tr}[T'] &= d_P, \\
            T' &= T' - {}_{F}T' + {}_{PF} T' =: \mathcal P_{\rm CPTP}^{\rm o}[T'].
        \end{align}
        The channel projectors are unital, self-adjoint and commute with transposition, so Lem.\ref{lem::SpecProj} yields the characterization
        \begin{align}
            \operatorname{Tr}[S] &= \frac{d_P}{d_I}\,d_I d_O = d_Pd_O, \\
            S &= S-{}_{F}S+{}_{OF}S-{}_{IOF}S+{}_{PIOF}S,
        \end{align}
        which, together with $S\geq 0$, yields the characterization provided in the Proposition. 
    \end{proof}
\end{prop}
In a similar vein, the Choi states of causally ordered $n$-slot combs $S^{(n)}$ can be characterized in terms of a projector and a trace constraint~\cite{Araujo2015}. For notational simplicity, whenever we consider the $n$-slot case, following the convention in the main text, we label the respective spaces by $0,1,\dots, 2n, 2n+1$. We then have the generalization of the one-slot case:
\begin{prop}[$n$-slot quantum combs~\cite{Araujo2015, Milz2024}]
\label{prop::n_slot_def}
A matrix $S^{(n)} \in \mathcal{L}(\Hcal_0 \otimes \Hcal_1 \otimes \cdots \otimes \Hcal_{2n} \otimes \Hcal_{2n+1})$ is the Choi operator of a causally ordered supermap $\mathcal{S}^{(n)}$ with ordering $0\prec 1\prec \cdots \prec 2n \prec 2n+1$ iff
\begin{align}
    S^{(n)} &\geq 0, \\
    \tr[S^{(n)}] &= d_0 d_2 \cdots d_{2n-2}d_{2n}\\
    \notag
    S^{(n)}  &= S^{(n)}  - {}_{(2n+1)}S^{(n)} + {}_{(2n)(2n+1)}S^{(n)} \\
    \notag
    &\phantom{=}- {}_{(2n-1)(2n)(2n+1)}S^{(n)} + \cdots -{}_{12\dots (2n)(2n+1)}S^{(n)} \\
    &\phantom{=}+  {}_{012\dots (2n)(2n+1)}S^{(n)} := \mathcal{P}_{\rm SC}^{(n)}[S^{(n)}].
\end{align}
\end{prop}
The above characterizations of causally ordered superchannels and $n$-slot combs in terms of the projectors $\mathcal{P}_{\rm SC}$ and $\mathcal{P}_{\rm SC}^{(n)}$ will be frequently employed below. In addition, we will also use the equivalent definition of $n$-slot combs originally derived in~\cite{chiribella_transforming_2008, Chiribella2009}:
\begin{prop}[$n$-slot combs, alternative characterization]
\label{prop::alt_charac_caus}
A matrix $S^{(n)} \in \mathcal{L}(\Hcal_0 \otimes \Hcal_1 \otimes \cdots \otimes \Hcal_{2n} \otimes \Hcal_{2n+1})$ is the Choi operator of a causally ordered supermap $\mathcal{S}^{(n)}$ with ordering $0\prec 1\prec \cdots \prec 2n \prec 2n+1$ iff
\begin{align}
    S^{(n)} &\geq 0, \\
    \tr_{2n+1} [S^{(n)}] &= S^{\prime (n-1)} \otimes \ident_{2n},\\
     \tr_{2n-1} [S^{\prime (n-1)}] &= S^{\prime (n-2)} \otimes \ident_{2n-2},\\
     \notag &\vdots\\ 
     \tr_{3}[S^{\prime(1)}] &= S^{\prime(0)} \otimes \ident_2,\\
     \tr_{1}[S^{\prime(0)}] &= \ident_0.
\end{align}
\end{prop}
\noindent Since not all the projectors we use are self-adjoint (see, e.g., Lem.~\ref{lem::adjoint}), we also require a version of Prop.~\ref{lem::SpecProj} that does not presuppose this property:
\begin{lemma}[Projective characterization of transformations -- general version]
\label{lem::FullyGenProj}
    Let $\mathcal{P}^{\rm i}:\mathcal{L}(\mathcal{H}_{\rm i})\to\mathcal{L}(\mathcal{H}_{\rm i})$ and $\mathcal{P}^{\rm o}:\mathcal{L}(\mathcal{H}_{\rm o})\to\mathcal{L}(\mathcal{H}_{\rm o})$ be linear projectors, corresponding to the input and output spaces, respectively. Consider the affine subspaces
    \begin{align}
        \mathsf{A}_{\rm i} &:= \left\{ W \in \mathcal{L}(\mathcal{H}_{\rm i}) \;\middle|\; \mathcal{P}^{\rm i}[W]=W,\; \operatorname{Tr}[W]=\gamma_{\rm i} \right\}, \\ 
        \mathsf{A}_{\rm o} &:= \left\{ W' \in \mathcal{L}(\mathcal{H}_{\rm o}) \;\middle|\; \mathcal{P}^{\rm o}[W']=W',\; \operatorname{Tr}[W']=\gamma_{\rm o} \right\}.
    \end{align}
    Then, for $\gamma_{\rm i}\neq 0$, a linear map $\mathcal{M}:\mathcal{L}(\mathcal{H}_{\rm i})\to\mathcal{L}(\mathcal{H}_{\rm o})$ satisfies
    \begin{equation}
        \mathcal{M}[\mathsf{A}_{\rm i}] \subseteq \mathsf{A}_{\rm o}
    \end{equation}
    if and only if its Choi operator $M \in \mathcal{L}(\mathcal{H}_{\rm i}\otimes\mathcal{H}_{\rm o})$ satisfies
    \begin{align}
        &M = M - (\mathcal{P}^{\rm i})^{\rm T}[M] + \bigl((\mathcal{P}^{\rm i})^{\rm T}\otimes \mathcal{P}^{\rm o}\bigr)[M] =: \mathcal{P}^{\rm io}[M], \label{eqn::FullyGenProj1} \\
        &(\mathcal{P}^{\rm i})^{\rm T}\!\bigl[\operatorname{Tr}_{\rm o}[M]\bigr] = \frac{\gamma_{\rm o}}{\gamma_{\rm i}} (\mathcal{P}^{\rm i})^{\rm T}[\ident_{\rm i}]. \label{eqn::FullyGenProj2}
\end{align}
Here $(\mathcal{P}^{\rm i})^{\rm T}[\bullet] = \mathcal{P}^{\rm i \dagger}[\bullet^*]^*$ denotes the transpose of the projector $\mathcal{P}^{\rm i}$, i.e., if $\mathcal{P}^{\rm i}[\bullet] = \sum_\alpha L_\alpha \bullet R_\alpha^\dagger$, then $(\mathcal{P}^{\rm i})^{\rm T}[\bullet] = \sum_\alpha L_\alpha^\mathrm{T} \bullet R_\alpha^*$.
    \begin{proof}
        See Theorem 5 in Ref. \cite{Milz2024}.
    \end{proof}
\end{lemma}

\section{Characterization of GP and GPTP maps}
\label{app:GP_GPTP}

Here, we provide a characterization of the sets $\mathsf{GP}(I,O)$ and $\mathsf{GPTP}(I,O)$ of Gibbs preserving and trace \textit{and} Gibbs preserving maps in terms of projectors $\mathcal{P}_{\rm GP}$ and $\mathcal{P}_{\rm GPTP}$. Throughout, we assume that the Gibbs states $\tau_X$ on all systems $X$ are full rank, so that their inverses $\tau_X^{-1}$ exist. 

\subsection{Gibbs projectors and Gibbs transform}
To characterize GP and GPTP maps, similar to Def.~\ref{def:proj} it proves helpful to introduce the following \textit{Gibbs projectors}
\begin{definition}[Gibbs projectors]
    Let $M \in \mathcal{L}(\mathcal{H})$, and let $\tau_{X}$ be the Gibbs state on a subsystem $\mathcal{H}_X \subseteq \mathcal{H}$. The projection operators ${}_{X_\alpha}\bullet$ and ${}_{X_\beta}\bullet$ are defined as
    \begin{align}
        {}_{X_\alpha}M &:= \tr_X[M] \otimes \tau_{X}, \label{eq:projGibbs_al1} \\
        {}_{X_\beta}M &:= \tr_X[\tau_{X} M] \otimes \ident_X. \label{eq:projGibbs_al2}
    \end{align}
\end{definition}
It is straightforward to verify that both maps are projectors, i.e.,
\begin{align}
    {}_{X_\alpha}({}_{X_\alpha}M) &= {}_{X_\alpha}M, \\
    {}_{X_\beta}({}_{X_\beta}M) &= {}_{X_\beta}M.
\end{align}
Note that, for the infinite temperature case $\tau_X = \ident_X/d_X$, we have ${}_{X_\alpha}M = {}_{X_\beta}M = {}_XM$, but in general ${}_{X_\alpha}M \neq {}_{X_\beta}M$. In addition, it is easy to see that ${}_{X_\alpha}({}_{Y_\alpha} M) = {}_{Y_\alpha}({}_{X_\alpha} M)$ and ${}_{X_\beta}({}_{Y_\beta} M) = {}_{Y_\beta}({}_{X_\beta} M)$ for all $X,Y$, as well as ${}_{X_\alpha}({}_{Y_\beta}M) = {}_{Y_\beta}({}_{X_\alpha}M)$ for $X\neq Y$, but ${}_{X_\alpha}({}_{X_\beta}M) \neq {}_{X_\beta}({}_{X_\alpha}M)$. 

\noindent Unlike ${}_X\bullet$, the operators ${}_{X_\alpha}\bullet$ and ${}_{X_\beta}\bullet$ are not self-adjoint, but the following relation between them holds:
\begin{lemma}[Adjoint relation of Gibbs projectors]
\label{lem::adjoint}
Let $M, N \in \mathcal{L}(\mathcal{H})$. Then
\begin{equation}
    \tr\!\left[M^\dagger\, {}_{X_\alpha}N\right] = \tr\!\left[({}_{X_\beta}M)^\dagger N\right],
\end{equation}
i.e., ${}_{X_\beta}\bullet$ is the adjoint of ${}_{X_\alpha}\bullet$ -- denoted by ${}_{X_\beta^\dagger}\bullet = {}_{X_\alpha}\bullet$ -- with respect to the Hilbert--Schmidt inner product. Moreover,
\begin{equation}
    {}_{X_\beta^{\mathrm T}}\bullet = {}_{X_\alpha}\bullet,
\end{equation}
where ${}_{X_\beta^{\mathrm T}}[M] := ({}_{X_\alpha^\dagger}[M^*])^*$, and complex conjugation is taken in the basis defining the Choi--Jamio\l kowski isomorphism.
\end{lemma}
\noindent We emphasize that, ${}_{X_\beta^\dagger}\bullet = {}_{X_\alpha}\bullet$ is equivalent to ${}_{X_\alpha^\dagger}\bullet = {}_{X_\beta}\bullet$ and ${}_{X_\beta^{\mathrm T}}\bullet = {}_{X_\alpha}\bullet$ is equivalent to ${}_{X_\alpha^{\mathrm T}}\bullet = {}_{X_\beta}\bullet$, since the respective operations $\mathrm{T}$ and $\dagger$ are involutions.
\begin{proof}
The first identity follows by direct computation:
\begin{align}
    \tr\!\left[M^\dagger\, {}_{X_\alpha}N\right] &= \tr\!\left[M^\dagger (\tr_X[N]\otimes \tau_X)\right] \nonumber\\
    &= \tr\!\left[(\tr_X[\tau_X M^\dagger]\otimes \ident_X)\, N\right] \nonumber\\
    &= \tr\!\left[(\tr_X[\tau_X M]\otimes \ident_X)^\dagger N\right] \nonumber\\
    &= \tr\!\left[({}_{X_\beta}M)^\dagger N\right],
\end{align}
where we used that $\tau_X$ is Hermitian. 

\noindent For the second identity, we compute
\begin{gather}
\begin{split}
    {}_{X_\alpha^{\mathrm T}}[M] &:= ({}_{X_\alpha^\dagger}[M^*])^* = ({}_{X_\beta}[M^*])^* = \left(\tr_X[\tau_X M^*]\otimes \ident_X \right)^* \\
    &= \tr_X[\tau_X M]\otimes \ident_X = {}_{X_\beta}M,
\end{split}
\end{gather}
where we used that $\tau_X^*=\tau_X$ in the basis of the Choi isomorphism.
\end{proof}
In addition to the above projectors, we will frequently make use of the \textit{Gibbs transform} of maps when characterizing the sets $\mathsf{GP}(I,O
)$ and $\mathsf{GPTP}(I,O)$: 
\begin{definition}[Gibbs transform]
\label{def:mapGibbsRotated}
    Let $T \in \mathcal{L}(\mathcal{H}_I \otimes \mathcal{H}_O)$ be the Choi operator of a map $\mathcal{T} : \mathcal{L}(\mathcal{H}_I) \rightarrow \mathcal{L}(\mathcal{H}_O)$. The \emph{Gibbs transform} $\widetilde{T}$ of $T$ is defined as
    \begin{equation}
        \widetilde{T} = (\sqrt{\tau_{I}} \otimes \ident_O) T (\sqrt{\tau_{I}} \otimes \ident_O).
    \end{equation}
    We will use a tilde throughout to denote the Gibbs transform of a matrix.  Accordingly, we denote by $\mathsf{GP}_{\tau_I}$ and $\mathsf{GPTP}_{\tau_I}$ the sets of Gibbs-transformed Choi operators corresponding to GP and GPTP maps, respectively.
\end{definition}

\subsection{Characterization of GP and GPTP maps}
We can now characterize GP maps in the Choi representation. To this end, we first recall that a completely positive map $\mathcal{G}: \Lcal(\Hcal_I) \rightarrow \Lcal(\Hcal_O)$ is GP iff $\mathcal{G}[\tau_I] = \tau_O$ and GPTP if, in addition, $\mathcal{G}$ is trace preserving. In terms of the corresponding Choi state $G \in \Lcal(\Hcal_I \otimes \Hcal_O)$, these conditions read
\begin{definition}
\label{def::GPTP}
    A map $\mathcal{G}:\Lcal(\Hcal_I) \rightarrow \Lcal(\Hcal_O)$ is GP iff its Choi operator $G \in \Lcal(\Hcal_I\otimes \Hcal_O)$ satisfies
    \begin{gather}
         G \geq 0, \quad \text{and} \quad  G \star \tau_I = \tr_I[\tau_I G] = \tau_O.
    \end{gather}
It is GPTP iff, in addition, $G$ satisfies 
\begin{gather}
    \tr_O[G] = \ident_I
\end{gather}
\end{definition}
We can now prove the following Proposition: 

\begin{prop}[Characterization of GP maps]
\label{prop::GP}
    An operator $G\in \mathcal{L}(\mathcal{H}_I \otimes \mathcal{H}_O)$ is the Choi operator of a GP map iff its Gibbs transform $\widetilde{G} \in \mathcal{L}(\mathcal{H}_I \otimes \mathcal{H}_O)$ satisfies 
    \begin{align}
        \label{eqn::GP1} \widetilde{G} &\geq 0,\\
        \label{eqn::GP2} \tr[\widetilde{G}] &= 1, \\
        \label{eqn::GP3} \widetilde{G} &= \widetilde{G} - {}_{I_\alpha}\widetilde{G} + {}_{I_\alpha O_\alpha}\widetilde{G} =: \Pcal_{\textup{GP}}[\widetilde{G}]. 
    \end{align}
\end{prop}
    \begin{proof}
        ("$\Rightarrow$") Let $G\in \Lcal(\Hcal_I \otimes \Hcal_O)$ be the Choi operator of a GP map, then (see Def.~\ref{def::GPTP})
        \begin{equation}
            G\geq 0, \qquad \tr_I[\tau_I G] = \tau_O.
        \end{equation}
        By Def. \ref{def:mapGibbsRotated}, we have 
        \begin{equation}
            \widetilde G = (\sqrt{\tau_I}\otimes \ident_O) G (\sqrt{\tau_I}\otimes \ident_O),
        \end{equation}
        so $\widetilde G \geq 0$. Moreover,
        \begin{gather}
            \nonumber \tr_I[\widetilde G] = \tr_I[(\tau_I\otimes \ident_O)G] = \tr_I[\tau_I G] = \tau_O.
        \end{gather}
        Taking the trace yields
        \begin{equation}
            \tr[\widetilde G] = 1.
        \end{equation}
        In addition, we obtain
        \begin{gather}
            {}_{I_\alpha}\widetilde G = \tr_I[\widetilde G]\otimes \tau_I = \tau_O\otimes \tau_I,
        \end{gather}
        and similarly
        \begin{gather}
            \nonumber {}_{I_\alpha O_\alpha}\widetilde G = \tr[\widetilde G]\,\tau_O\otimes \tau_I = \tau_O\otimes \tau_I.
        \end{gather}
        Together, this implies 
        \begin{equation}
            {}_{I_\alpha}\widetilde G = {}_{I_\alpha O_\alpha}\widetilde G,
        \end{equation}
        which is equivalent to the projector condition \eqref{eqn::GP3}.

        ("$\Leftarrow$") Assume the conditions \eqref{eqn::GP1}--\eqref{eqn::GP3}. Positivity of $\widetilde G$ implies $G\geq 0$.

        The projector identity \eqref{eqn::GP3} gives
        \begin{equation}
            {}_{I_\alpha}\widetilde G = {}_{I_\alpha O_\alpha}\widetilde G,
        \end{equation}
        which implies
        \begin{gather}
            \tr_I[\tau_I G] = \tr_I[\widetilde G] = \tr[\widetilde G]\,\tau_O 
            = \tau_O.
        \end{gather}
        Consequently, $G$ is the Choi operator of a GP map.
    \end{proof}
By direct insertion, it is easy to see that $\mathcal{P}_{\rm GP}$ is indeed a projector, i.e., $\mathcal{P}_{\rm GP}^2 = \mathcal{P}_{\rm GP}$. However, it is neither unital, nor self-adjoint. Note that, in the infinite-temperature limit, the Gibbs state reduces to the maximally mixed state, $\tau_X = \ident_X / d_X$, and Gibbs-preserving maps coincide with unital maps. In this case, the Gibbs transformation becomes trivial, $\widetilde G = G$, and Proposition~\ref{prop::GP} reduces to conditions \eqref{eqn::GP1}--\eqref{eqn::GP3} applied directly to the Choi operator $G$, with the normalization $\Tr[G] = d_I$.

Imposing trace preservation in addition leads to the following characterization of GPTP maps.

\begin{prop}[Characterization of GPTP maps]
\label{prop::GPTP}
An operator $G\in \mathcal{L}(\mathcal{H}_I \otimes \mathcal{H}_O)$ is the Choi operator of a GPTP map iff its Gibbs transform $\widetilde{G} \in \mathcal{L}(\mathcal{H}_I \otimes \mathcal{H}_O)$ satisfies 
\begin{align}
\label{eqn::GPTP1}
    \widetilde{G} &\geq 0,\\
    \label{eqn::GPTP2}
    \tr[\widetilde{G}] &=  1, \\
    \label{eqn::GPTP3}
    \widetilde{G} &= \widetilde{G} - {}_{I_\alpha}\widetilde{G} - {}_{O_\alpha }\widetilde{G} + 2{}_{I_\alpha O_\alpha}\widetilde{G} =: \Pcal_{\textup{GPTP}}[\widetilde{G}].
\end{align}

\begin{proof}
("$\Rightarrow$")
Let $G$ be the Choi operator of a GPTP map. Since a GPTP map is GP, we obtain from Prop.~\ref{prop::GP} that 
\begin{gather}
\widetilde G \geq 0, \quad \tr[\widetilde G] = 1, \quad \text{and} \quad {}_{I_{\alpha}}\widetilde{G} = {}_{I_\alpha O_\alpha}\widetilde G.
\end{gather}
In addition, the TP property $\tr_O[G] = \ident_I$ implies
\begin{align}
\tr_O[\widetilde G] = \sqrt{\tau_I}\,\tr_O[G]\,\sqrt{\tau_I} = \tau_I.
\end{align}
Consequently, we have
\begin{gather}
{}_{O_\alpha}\widetilde G = {}_{I_\alpha O_\alpha}\widetilde G, 
\end{gather}
which, together with ${}_{I_\alpha}\widetilde G = {}_{I_\alpha O_\alpha}\widetilde G$ yields the projector identity \eqref{eqn::GPTP3}.

\noindent ("$\Leftarrow$") First, $\widetilde {G} \geq 0$ implies $G \geq 0$. Second, applying ${}_{I_\alpha}\bullet$ to the constraint \eqref{eqn::GPTP3}, we obtain
\begin{equation}
{}_{I_\alpha}\widetilde G = {}_{I_\alpha O_\alpha}\widetilde G \quad \Rightarrow \quad \tr_I[\tau_I G] \otimes \tau_I = \tr[\widetilde G] \tau_I \otimes \tau_O. 
\end{equation}
Using $\tr[\widetilde G] = 1$, this implies $\tr_I[\tau_I G] = \tau_O$, i.e., $G$ is GP.
\begin{equation}
\tr_I[\widetilde G]=\tau_O.
\end{equation}
Similarly, applying ${}_{O_\alpha}\bullet$ to both sides of Eq.~\eqref{eqn::GPTP3} gives
\begin{equation}
{}_{O_\alpha}\widetilde G = {}_{I_\alpha O_\alpha}\widetilde G \quad \Rightarrow \quad \tr_O[\widetilde G]=\tau_I.
\end{equation}
The latter implies
\begin{equation}
\tr_O[\widetilde{G}] = \sqrt{\tau_I} \tr_O[G] \sqrt{\tau_I} = \tau_I
\end{equation}
and thus $\tr_O[G] = \ident_I$; consequently, $G$ is both GP \textit{and} trace-preserving, i.e., it is (the Choi operator of) a GPTP map.
\end{proof}
\end{prop}
Again, it is easy to see that $\Pcal_{\textup{GPTP}}$ is indeed a projector, i.e., $\Pcal_{\textup{GPTP}}^2 = \Pcal_{\textup{GPTP}}$. Similarly to the projector $\mathcal{P}_{\textup{GP}}$, the projector $ \Pcal_{\textup{GPTP}}$ is neither unital, nor self-adjoint. The above Propositions provide a characterization of the affine spaces $\mathsf{GP}_{\tau_I}(I,O)$ and $\mathsf{GPTP}_{\tau_I}(I,O)$ (as well as $\mathsf{GP}_{\tau_P}(P,F)$ and $\mathsf{GPTP}_{\tau_P}(P,F)$ via simple relabeling) in terms of projectors $\Pcal_{\textup{GP}}$ and $\Pcal_{\textup{GPTP}}$ (plus additional trace constraints). In turn, this fully pins down the properties of the respective non-Gibbs-transformed sets $\mathsf{GP}(I,O)$ and $\mathsf{GPTP}_{\tau_I}(I,O)$ (as well as $\mathsf{GP}(P,F)$ and $\mathsf{GPTP}(P,F)$). The characterization in terms of projectors is particularly useful, since it is well suited for deriving the properties of supermaps -- like GPP or GPTPP maps -- in a systematic manner, by employing Lem.~\ref{lem::FullyGenProj}. We note that Gibbs transform is crucial to be able to leverage Lem.~\ref{lem::FullyGenProj}, since, for example, not all maps in $\mathsf{GP}(I,O)$ have the same trace, while maps in $\mathsf{GPTP}(I,O)$ are difficult to characterize directly by means of a \textit{single} projector. Using the Gibbs transformed versions of both spaces alleviates these problems entirely.

\section{Characterization of GPP and GPTPP supermaps}
\label{app:freeTrans}

In this Appendix, we prove the structural characterizations of the free supermaps introduced in Sec.~\ref{subsec:free_transform}. 

\subsection{Gibbs transform of supermaps}
As shown above (see Props.~\ref{prop::GP} and~\ref{prop::GPTP}), GP and GPTP maps are most conveniently described using the Gibbs transforms of their Choi operators. Here, we first extend this notion of Gibbs transform to supermaps. In turn, this will allow us to apply Lem.~\ref{lem::FullyGenProj} directly to the affine spaces of Gibbs transforms of GP and GPTP Choi operators.

\begin{definition}[Gibbs transform of a supermap]
\label{def:rotatedSupermap}
    Let $S \in \mathcal{L}(\mathcal{H}_P \otimes \mathcal{H}_I \otimes \mathcal{H}_O \otimes \mathcal{H}_F)$ be the Choi operator of a supermap $\mathcal{S}: \mathsf{CP}(I,O) \to \mathsf{CP}(P,F)$. The \emph{Gibbs-transform} of $S$ is defined as
    \begin{equation}
        \breve{S} = (\sqrt{\tau_{P}}\otimes\sqrt{\tau_{I}}^{-1})\, S\, (\sqrt{\tau_{P}}\otimes\sqrt{\tau_{I}}^{-1}).
    \end{equation}
    We use the $\breve{\bullet}$ notation throughout to denote Gibbs transforms of (Choi operators of) supermaps.
\end{definition}

The following Lemma shows that the characterization of GPP and GPTPP supermaps can equivalently be phrased in terms of their Gibbs transform:
\begin{lemma}[Equivalence under Gibbs transform]
\label{lem::Stilde_S}
    Let $S \in \mathcal{L}(\mathcal{H}_P \otimes \mathcal{H}_I \otimes \mathcal{H}_O \otimes \mathcal{H}_F)$ be the Choi operator of a supermap $\mathcal{S}: \mathsf{CP}(I,O) \to \mathsf{CP}(P,F)$, and let $\breve{S}$ be the Gibbs transform of its Choi operator. Then the following are equivalent:
    \begin{enumerate}
        \item $\mathcal{S}(\mathsf{GP}(I,O)) \subseteq \mathsf{GP}(P,F)$,
        \item $\breve{\mathcal{S}}({\mathsf{GP}_{\tau_I}}(I,O)) \subseteq {\mathsf{GP}_{\tau_P}}(P,F)$,
    \end{enumerate}
    where $\breve{\mathcal{S}}$ denotes the supermap associated with $\breve{S}$. The same equivalence holds with $\mathsf{GP}$ replaced by $\mathsf{GPTP}$.
    \begin{proof}
        First, using the definition of the link product, one verifies that
        \begin{align}\label{eq:intertwining}
           \nonumber \sqrt{\tau_{P}}\,(S \star G)\,\sqrt{\tau_{P}} &= \sqrt{\tau_{P}}\,(S \star \sqrt{\tau_{I}}^{-1}\widetilde{G}\sqrt{\tau_{I}}^{-1} )\,\sqrt{\tau_{P}} \\
            \nonumber &=  (\sqrt{\tau_{P}} \otimes \sqrt{\tau_{I}}^{-1}) S (\sqrt{\tau_{P}} \otimes \sqrt{\tau_{I}}^{-1}) \star \widetilde{G} \\
            &= \breve{S} \star \widetilde G.
        \end{align}

        \medskip
        \noindent
        ("$1. \Rightarrow 2.$") Assume that $S$ maps $\mathsf{GP}(I,O)$ into $\mathsf{GP}(P,F)$. Let $\widetilde G \in {\mathsf{GP}}_{\tau_I}(I,O)$ and define $G := \sqrt{\tau_{I}}^{-1}\widetilde G\sqrt{\tau_{I}}^{-1}$. By definition, $G$ is GP, and hence $S \star G =: G' \in \textsf{GP}(P,F)$. Applying $\sqrt{\tau_P} \bullet \sqrt{\tau_P}$ on both sides of this equation and using Eq.~\eqref{eq:intertwining}, we obtain
        \begin{gather}
            \breve{S} \star \widetilde{G} = \sqrt{\tau_P} \,G' \sqrt{\tau_P}  \in \textsf{GP}_{\tau_P}(P,F), 
        \end{gather}
        which holds for all $\widetilde{G} \in  \textsf{GP}_{\tau_I}(I,O)$. We thus have the implication $"1. \Rightarrow 2.$". The proof of the implication $"2. \Rightarrow 1.$" follows in the same way.
    \end{proof} 
\end{lemma}
As a consequence of the above Lemma, the characterization of GPP and GPTPP supermaps $S$ is equivalent to the characterization of (the Choi operators of) supermaps $\breve{S}$ that transform $\mathsf{GP}_{\tau_I}(I,O)$ into $\mathsf{GP}_{\tau_P}(P,F)$ and  $\mathsf{GPTP}_{\tau_I}(I,O)$ into $\mathsf{GPTP}_{\tau_P}(P,F)$, respectively. Since all four of these latter sets can be characterized in terms of projectors and trace constraints (see Props.~\ref{prop::GP} and~\ref{prop::GPTP}), this allows us to use Lem.~\ref{lem::FullyGenProj} to deduce the projectors and trace constraints that define the respective maps $\breve{S}$.  

\subsection{Characterization of GPP transformations}
\label{app:GPP}
We start by providing a characterization of $\mathbf{GPP}$ supermaps:
\begin{lemma}[Projective characterization of $\mathbf{GPP}$ supermaps]\label{lem:app-GPP-projective}
    An operator $S \in \mathcal{L}(\mathcal{H}_{P} \otimes \mathcal{H}_{I} \otimes \mathcal{H}_{O} \otimes \mathcal{H}_{F})$ is the Choi operator of a $\mathbf{GPP}$ supermap iff its Gibbs transform $\breve{S} \in \mathcal{L}(\mathcal{H}_{P} \otimes \mathcal{H}_{I} \otimes \mathcal{H}_{O} \otimes \mathcal{H}_{F})$ satisfies
    \begin{align}
        \label{eqn::GPP1} \breve{S} &\geq 0,\\
        \nonumber \breve{S} &= \breve{S} -{}_{P_\alpha}\breve{S} + {}_{P_\alpha I_\beta}\breve{S} - {}_{P_\alpha I_\beta O_\beta}\breve{S} + {}_{P_\alpha F_\alpha}\breve{S} \\
        \label{eqn::GPP3}
        &\phantom{=,}- {}_{P_\alpha I_\beta F_\alpha}\breve{S} + {}_{P_\alpha I_\beta O_\beta F_\alpha}\breve{S} =: \Pcal_{\textup{GPP}}[\breve{S}].  \\
        \text{and} \quad \label{eqn::GPP2} \operatorname{Tr}_{PF}&[ \breve{S} - {}_{I_\beta}\breve{S} + {}_{I_\beta O_\beta}\breve{S} ] = \ident_{IO}.
    \end{align}
    \end{lemma}
    \begin{proof}
        By Lem.~\ref{lem::Stilde_S}, $S$ is the Choi operator of a $\mathbf{GPP}$ supermap $\mathcal{S}$, i.e.
        \begin{equation}
            \mathcal{S}\bigl(\mathsf{GP}(I,O)\bigr) \subseteq \mathsf{GP}(P,F),
        \end{equation}
        if and only if its Gibbs transform $\breve{S}$ maps ${\mathsf{GP}}_{\tau_I}(I,O)$ into ${\mathsf{GP}}_{\tau_P}(P,F)$. By Proposition \ref{prop::GP}, $\widetilde{G} \in \mathsf{GP}_{\tau_I}(I,O)$ iff 
        \begin{gather}
            \begin{split}
\widetilde{G} &\geq 0, \quad \tr[\widetilde{G}] = 1 := \gamma_{\rm{i}} \\
\text{and} \quad \widetilde{G}  &=   \widetilde{G} - {}_{I_\alpha} \widetilde{G} +{}_{I_\alpha O_\alpha} \widetilde{G} := \mathcal{P}_{\rm GP}^{\rm i}[\widetilde{G}],               
            \end{split}
        \end{gather}
        where $\rm{i} = IO$. Analogously, $\widetilde{G}' \in \mathsf{GP}_{\tau_P}(P,F)$ iff 
        \begin{gather}
            \begin{split}
                \widetilde{G}' &\geq 0, \quad \tr[\widetilde{G}'] = 1 := \gamma_{\rm{o}} \\
                \text{and} \quad \widetilde{G}' &= \widetilde{G}' - {}_{P_\alpha}\widetilde{G}' + {}_{P_\alpha F_\alpha}\widetilde{G}' :=  \mathcal P_{\rm GP}^{\rm o}[\widetilde G'],
        \end{split}
        \end{gather}
        where $\rm{o} = PF$. From Lem.~\ref{lem::FullyGenProj}, we know that $\breve S$ maps $\textsf{GP}_{\tau_I}(I,O)$ into $\textsf{GP}_{\tau_P}(P,F)$ iff
        \begin{align}
        \label{eq:app-affine-map-GPP-1}
            &\breve{S} = \breve{S} - (\mathcal{P}_{\rm GP}^{\rm i})^\mathrm{T} [\breve{S}] + ((\mathcal{P}_{\rm GP}^{\rm i})^\mathrm{T} \otimes \mathcal{P}_{\rm GP}^{\rm o})[\breve S] \\
            \label{eq:app-affine-map-GPP-2}
            \text{and} \quad &(\mathcal{P}_{\rm GP}^{\rm i})^\mathrm{T}\left[\tr_{PF}[\breve{S}]\right] = (\mathcal{P}_{\rm GP}^{\rm i})^\mathrm{T}[\ident_{\rm i}],    
        \end{align}
        where we have used $\gamma_{\rm o}/\gamma_{\rm i} = 1$. By Lem.~\ref{lem::adjoint}, the transpose of the input projector $\mathcal{P}_{\rm GP}^{\rm i}$ is obtained by replacing all instances of ${}_{X_\alpha}\bullet$ by ${}_{X_\beta}\bullet$, i.e., 
        \begin{equation}
            (\mathcal P_{\rm GP}^{\rm i})^{\rm T}[\breve S] = \breve{S} -{}_{I_\beta}\breve{S} + {}_{I_\beta O_\beta} \breve{S}.
        \end{equation}
        
        With this, expansion of Eq.~\eqref{eq:app-affine-map-GPP-1} gives Eq.~ \eqref{eqn::GPP3}. Similarly, expanding Eq.~\eqref{eq:app-affine-map-GPP-2} yields
        \begin{equation}\label{eq:proof2}
            \operatorname{Tr}_{PF} \left[\breve{S} - {}_{I_\beta}\breve{S} + {}_{I_\beta O_\beta}\breve{S} \right] = \ident_{IO} - {}_{I_\beta}\ident_{IO} + {}_{I_\beta O_\beta}\ident_{IO}.
        \end{equation}
        Since
        \begin{equation}
            {}_{I_\beta}\ident_{IO} = {}_{I_\beta O_\beta}\ident_{IO},
        \end{equation}
        the right-hand side of \eqref{eq:proof2} reduces to $\ident_{IO}$, yielding \eqref{eqn::GPP2}.

        Finally, complete positivity of the supermap is equivalent to positivity of its Choi operator $S$, which is equivalent to $\breve S \geq 0$, implying condition~\eqref{eqn::GPP1}. This proves the Lemma.
    \end{proof}
With the above Lemma, we can prove a more intuitive characterization of $\mathbf{GPP}$ supermaps as the set of \textit{thermal supermaps}, i.e., the set of supermaps that do not allow for an extraction of athermality from equilibrium. More rigorously, we have:

\begin{prop}[Thermality of $\mathbf{GPP}$ supermaps]
\label{prop:app-GPP-link}
    An operator $S \in \mathcal{L}(\mathcal{H}_{P} \otimes \mathcal{H}_{I} \otimes \mathcal{H}_{O} \otimes \mathcal{H}_{F})$ is the Choi operator of a $\mathbf{GPP}$ supermap iff $S\geq 0$ and there exists a Gibbs preserving map $ \mathcal{G} \in \mathsf{GP}(O, F)$ such that
    \begin{equation}
        \label{eq:app-GPP-action}
        (\mathcal{S}[\mathcal{T}])[\tau_{P}] = \mathcal{G}\!\left[\mathcal{T}[\tau_{I}]\right],
    \end{equation}
    for any $\mathcal{T} \in \mathsf{CP}(I,O)$; or, equivalently,
    \begin{align}\label{eq:app-GPP-link1}
        \tau_{P}\star S_{PIOF} &= \tau_{I}\otimes G_{OF}, \\
        \tau_{O}\star G_{OF} &= \tau_{F}\label{eq:app-GPP-link2},
    \end{align}
    where we have used subscripts to denote the spaces the respective operators are defined on.
\end{prop}
    \begin{proof} 
    First, we note that equivalence of the two characterizations follows from the fact that, expressed in terms of the respective Choi operators, Eq.~\eqref{eq:app-GPP-action} reads 
    \begin{gather}
        S_{PIOF} \star T_{IO} \star \tau_{P} = G_{OF} \star T_{IO} \star \tau_{I},
    \end{gather}
    where $G_{OF}$ is the Choi operator of a GP map. Since the link product is commutative and associative, this implies $(S_{PIOF} \star \tau_P) \star T_{IO} = (\tau_I \otimes G_{OF}) \star T_{IO}$. Requiring this to hold for all CP maps $T_{IO}$ then yields $S_{PIOF} \star \tau_P = \tau_I \otimes G_{OF}$, while Gibbs preservation of $G_{OF}$ implies $\tau_O \star G_{OF} = \tau_F$. Proof of the converse direction follows along the same lines. Now, to prove the main statement of the Proposition, we start with the "only if" direction:\medskip\\
\noindent         ("$\Rightarrow$") Assume $\mathcal{S}$ is a $\mathbf{GPP}$ supermap. The strategy is now to find the properties of the corresponding Gibbs transform $\operatorname{Tr}_P[\breve S]$ and show that they imply thermality of $S$. By Lem.~\ref{lem:app-GPP-projective}, the corresponding $\breve{S}$ satisfies Eqs.~\eqref{eqn::GPP2} and~\eqref{eqn::GPP3}. Applying ${}_{P_\alpha I_\beta O_\beta} \bullet$ to~\eqref{eqn::GPP3} gives
        \begin{equation}\label{eq:ABC-ABCD}
            {}_{P_\alpha I_\beta O_\beta}\breve{S} = {}_{P_\alpha I_\beta O_\beta F_\alpha}\breve{S}.
        \end{equation}
        Furthermore, Eq.~\eqref{eqn::GPP2} implies
        \begin{equation}\label{eq:proof1}
            \operatorname{Tr}_{PF}[\breve{S}] -{}_{I_\beta}\operatorname{Tr}_{PF}[\breve{S}] + {}_{I_\beta O_\beta} \operatorname{Tr}_{PF}[\breve{S}] = \ident_{IO},
        \end{equation}
        Applying ${}_{I_\beta O_\beta}\bullet$ on both sides yields
        \begin{equation}
            {}_{I_\beta O_\beta}\operatorname{Tr}_{PF}[\breve{S}] = \ident_{IO}.
        \end{equation}
        Substituting this back into Eq.~\eqref{eq:proof1}, we obtain
        \begin{equation}
            \operatorname{Tr}_{PF}[\breve{S}] ={}_{I_\beta} \operatorname{Tr}_{PF}[\breve{S}], 
        \end{equation}
        which implies 
        \begin{equation}\label{eq:AD-ABD}
            {}_{P_\alpha F_\alpha}\breve{S} = {}_{P_\alpha I_\beta F_\alpha}\breve{S} .
        \end{equation}
        Applying ${}_{P_\alpha}\bullet$ to \eqref{eqn::GPP3} yields
        \begin{gather}
            {}_{P_\alpha}\breve{S} = {}_{P_\alpha I_\beta}\breve{S} - {}_{P_\alpha I_\beta O_\beta}\breve{S} + {}_{P_\alpha F_\alpha}\breve{S} - {}_{P_\alpha I_\beta F_\alpha}\breve{S} +  {}_{P_\alpha I_\beta O_\beta F_\alpha}\breve{S} .
        \end{gather}
        Using Eqs.~\eqref{eq:ABC-ABCD} and~\eqref{eq:AD-ABD}, this reduces to
        \begin{equation}\label{eq:Palpha-PIbeta}
            {}_{P_\alpha}\breve{S} = {}_{P_\alpha I_\beta}\breve{S}.
        \end{equation}
        We can now use the definitions of the projectors ${}_{X_\alpha}\bullet$ and ${}_{X_\beta}\bullet$ as well as $\breve S$ to recover the statement of the Proposition. First, using the definition of the projectors, the above Eq.~\eqref{eq:Palpha-PIbeta} is equivalent to
        \begin{equation}\label{eq:A-factor}
            \operatorname{Tr}_P[\breve{S}] = \ident_I\otimes G_{OF},
        \end{equation}
        where
        \begin{gather}\label{eq:chanGPS}
            \nonumber G_{OF} = \operatorname{Tr}_I\Bigl[\tau_{I} \operatorname{Tr}_P[\breve{S}]\Bigr] = \tau_{I} \star \operatorname{Tr}_P[\breve{S}].
        \end{gather}
        Since $\breve{S}\geq 0$, we have $\operatorname{Tr}_P[\breve{S}] \geq 0$, and therefore $G_{OF}\geq 0$. From Def.~\ref{def:rotatedSupermap}, we see that  
        \begin{align}
            \nonumber \operatorname{Tr}_P[\breve{S}] &= \operatorname{Tr}_P\Bigl[(\sqrt{\tau_{P}}\otimes\sqrt{\tau_{I}}^{-1}) S (\sqrt{\tau_{P}}\otimes\sqrt{\tau_{I}}^{-1})\Bigr] \\
            &= \sqrt{\tau_{I}}^{-1} \bigl(\tau_{P}\star S_{PIOF}\bigr) \sqrt{\tau_{I}}^{-1},
        \end{align}
        holds, or, equivalently,
        \begin{equation}
            \tau_{P}\star S_{PIOF} = \sqrt{\tau_{I}}\operatorname{Tr}_P[\breve{S}] \sqrt{\tau_{I}}.
        \end{equation}
        Substituting this into $\operatorname{Tr}_P[\breve{S}] = \ident_I\otimes G_{OF}$ gives
        \begin{equation}
            \tau_{P}\star S_{PIOF} = \tau_{I}\otimes G_{OF}, 
        \end{equation}
i.e., Eq.~\eqref{eq:app-GPP-link1} of the Proposition. It remains to show that $G_{OF}$ is Gibbs-preserving. Tracing out subsystem $P$ in \eqref{eq:ABC-ABCD}, we obtain
        \begin{equation}
            {}_{I_\beta O_\beta} \operatorname{Tr}_P[\breve{S}] = {}_{I_\beta O_\beta F_\alpha}\operatorname{Tr}_P[\breve{S}].
        \end{equation}
        Using $\operatorname{Tr}_P[\breve{S}] = \ident_I\otimes G_{OF}$, this becomes
        \begin{equation}
            \ident_I\otimes \ident_O \otimes (\tau_{O} \star G_{OF}) = (\ident_I\otimes \ident_O \otimes \tau_{F})\,\operatorname{Tr}[\tau_{O}G_{OF}].
        \end{equation}
        Therefore,
        \begin{equation} \label{eq:G-proportional-Gibbs}
            \tau_{O} \star G_{OF} = \operatorname{Tr}[\tau_{O}G_{OF}]\,\tau_{F},
        \end{equation}
        i.e., $G_{OF}$ is proportional to a Gibbs-preserving map. In turn, \eqref{eq:proof1} implies, after applying ${}_{I_\beta O_\beta} \bullet$ to both sides, that
        \begin{equation}
            {}_{I_\beta O_\beta}\operatorname{Tr}_{PF}[\breve{S}] = \ident_{IO},
        \end{equation}
        or, equivalently,
        \begin{equation}
            \operatorname{Tr}\Bigl[(\tau_{I} \otimes \tau_{O})\operatorname{Tr}_{PF}[\breve{S}]\Bigr] = 1.
        \end{equation}
        Using \eqref{eq:A-factor}, this becomes
        \begin{equation}
            \operatorname{Tr}_{OF}[(\tau_{O}\otimes \ident_F)G_{OF}] = 1,
        \end{equation}
        or, equivalently,
        \begin{equation}
            \operatorname{Tr}[\tau_{O}G_{OF}] = 1.
        \end{equation}
        Therefore,
        \begin{equation}
            \tau_{O}\star G_{OF} = \tau_{F},
        \end{equation}
        and thus $\mathcal{G}\in\mathsf{GP}(O,F)$.\medskip\\

\noindent        ("$\Leftarrow$") Conversely, assume that Eq.~\eqref{eq:app-GPP-action} holds. Let $\mathcal{T} \in \mathsf{GP}(I,O)$ be Gibbs-preserving. Then,
        \begin{equation}
            (\mathcal{S}[\mathcal{T}])[\tau_{P}] = \mathcal{G}[\mathcal{T}[\tau_{I}]] = \mathcal{G}[\tau_{O}] = \tau_{F}.
        \end{equation}
        Hence $\mathcal{S}[\mathcal{T}] \in \mathsf{GP}(P,F)$ for every GP map $\mathcal{T}$, and $\mathcal{S}$ is a $\mathbf{GPP}$ supermap.
    \end{proof}
\subsection{Characterization of GPTPP transformations}
\label{app:GPTPP}
In a similar fashion to the previous section, here we provide a complete characterization of GPTPP supermaps $\mathcal{S}: \textsf{GPTP}(I,O) \rightarrow \textsf{GPTP}(P,F)$ in terms of the Gibbs transforms $\breve S$ of their Choi states:
\begin{lemma}[Projective characterization of $\mathbf{GPTPP}$ supermaps] \label{lem:app-GPTPP}
    An operator $S \in \mathcal{L}(\mathcal{H}_{P} \otimes \mathcal{H}_{I} \otimes \mathcal{H}_{O} \otimes \mathcal{H}_{F})$ is the Choi operator of a $\mathbf{GPTPP}$ supermap $\mathcal{S}$ iff the corresponding Gibbs transform $\breve{S} \in \mathcal{L}(\mathcal{H}_{P} \otimes \mathcal{H}_{I} \otimes \mathcal{H}_{O} \otimes \mathcal{H}_{F})$ satisfies
    \begin{align}
        \label{eqn::GPTPP1} \breve{S} &\geq 0,\\
        \label{eqn::GPTPP2} \operatorname{Tr}_{PF}&[\breve{S} - {}_{I_\beta}\breve{S} - {}_{O_\beta}\breve{S} + 2{}_{I_\beta O_\beta}\breve{S}] = \ident_{IO}, \\
        \nonumber \breve{S} &= \breve{S} - {}_{P_\alpha}\breve{S} + {}_{P_\alpha I_\beta}\breve{S} + {}_{P_\alpha O_\beta}\breve{S}  - 2 {}_{P_\alpha I_\beta O_\beta}\breve{S} - {}_{F_\alpha}\breve{S}  \\
        \nonumber &\phantom{=,}+ {}_{I_\beta F_\alpha }\breve{S} + {}_{O_\beta F_\alpha }\breve{S} - 2{}_{I_\beta O_\beta F_\alpha}\breve{S} + 2{}_{P_\alpha F_\alpha}\breve{S} - 2{}_{P_\alpha I_\beta F_\alpha}\breve{S} \\
        &\nonumber \phantom{=,}- 2{}_{P_\alpha O_\beta F_\alpha}\breve{S} + 4{}_{P_\alpha I_\beta O_\beta F_\alpha}\breve{S}\\
        &:= \Pcal_{\textup{GPTPP}}[\breve{S}]. \label{eqn::GPTPP3}
    \end{align}
\end{lemma}
    \begin{proof}
    The proof follows along the same lines as that of Lem.~\ref{lem:app-GPP-projective}
    First, we recall that, by Lem.~\ref{lem::Stilde_S}, $S$ is a $\mathbf{GPTPP}$ supermap, i.e.
        \begin{equation}
            \mathcal{S}\bigl(\mathsf{GPTP}(I,O)\bigr) \subseteq \mathsf{GPTP}(P,F),
        \end{equation}
        if and only if the Gibbs transform $\breve S$ of its Choi operator maps the affine space ${\mathsf{GPTP}}_{\tau_I}(I,O)$ into the affine space ${\mathsf{GPTP}}_{\tau_P}(P,F)$.

        \noindent By Proposition~\ref{prop::GPTP}, $\widetilde{G} \in {\mathsf{GPTP}}_{\tau_I}(I,O)$ iff
        \begin{gather}
        \begin{split}
            \widetilde{G} &\geq 0, \quad \tr[\widetilde{G}] = 1 =: \gamma_{\rm{i}}, \\
            \text{and} \quad \widetilde{G} &= \widetilde{G} - {}_{I_\alpha}\widetilde{G} - {}_{O_\alpha}\widetilde{G} + 2{}_{I_\alpha O_\alpha}\widetilde{G} =: P_{\rm GPTP}^{\rm i} [\widetilde{G}].
        \end{split}
        \end{gather}
        Similarly, $\widetilde{G}' \in {\mathsf{GPTP}}_{\tau_P}(P,F)$ iff
        \begin{gather}
        \begin{split}
            \widetilde{G}' &\geq 0, \quad \tr[\widetilde{G}'] = 1 =: \gamma_{\rm{o}}, \\
            \text{and} \quad  \widetilde{G}' &= \widetilde{G}' - {}_{P_\alpha}\widetilde{G}' - {}_{F_\alpha}\widetilde{G}' + 2{}_{P_\alpha F_\alpha}\widetilde{G}' =: P_{\rm GPTP}^{\rm o} [\widetilde{G}'].
        \end{split}
        \end{gather}
        From Lem.~\ref{lem::FullyGenProj}, we know that $\breve{S}$ maps ${\mathsf{GPTP}}_{\tau_I}(I,O)$ into ${\mathsf{GPTP}}_{\tau_P}(P,F)$ if and only if
        \begin{align}\label{eq:GPTPP-proof-proj}
            &\breve{S} = \breve{S} - (\mathcal P_{\rm GPTP}^{\rm i})^{\rm T}[\breve{S}] + \bigl((\mathcal P_{\rm GPTP}^{\rm i})^{\rm T} \otimes \mathcal P_{\rm GPTP}^{\rm o}\bigr)[\breve{S}],\\
            \label{eq:GPTPP-proof-affine}
            &\text{and} \quad  (\mathcal P_{\rm GPTP}^{\rm i})^{\rm T}\!\left[\operatorname{Tr}_{PF}\breve{S}\right] = (\mathcal P_{\rm GPTP}^{\rm i})^{\rm T}[\ident_{IO}].
        \end{align}
        By Lem.~\ref{lem::adjoint}, the transpose of the input projector is obtained by replacing $\alpha \mapsto \beta$, yielding
        \begin{equation}
            (\mathcal P_{\rm GPTP}^{\rm i})^{\rm T}[\breve S] = \breve{S} - {}_{I_\beta}\breve{S} - {}_{O_\beta}\breve{S} + 2{}_{I_\beta O_\beta}\breve{S}.
        \end{equation}

        \noindent Using this in the expansion of Eq.~\eqref{eq:GPTPP-proof-proj} yields Eq.~\eqref{eqn::GPTPP3}, while expanding Eq.~\eqref{eq:GPTPP-proof-affine} gives
        \begin{align}\label{eq:proof3}
            \nonumber \operatorname{Tr}_{PF}\!\left[ \breve{S} - {}_{I_\beta}\breve{S} - {}_{O_\beta}\widetilde{S} + 2{}_{I_\beta O_\beta}\breve{S} \right] &= \ident_{IO} - {}_{I_\beta}\ident_{IO} \\
            &\quad - {}_{O_\beta}\ident_{IO} + 2{}_{I_\beta O_\beta}\ident_{IO}.
        \end{align}
        Since
        \begin{equation}
            {}_{I_\beta}\ident_{IO} = {}_{O_\beta}\ident_{IO} = {}_{I_\beta O_\beta}\ident_{IO},
        \end{equation}
        the right-hand side of \eqref{eq:proof3} reduces to $\ident_{IO}$, yielding \eqref{eqn::GPTPP2}. Finally, complete positivity of the supermap is equivalent to $S\geq 0$, which is equivalent to $\breve S \geq 0$, giving \eqref{eqn::GPTPP1}.
    \end{proof}
Unlike $\mathbf{GPP}$, the characterization of $\mathbf{GPTPP}$ does not allow for an intuitive physical interpretation (say, in terms of thermality of the corresponding supermaps). This changes drastically when we require Gibbs preservation to hold in a \textit{complete} sense -- yielding the sets $\mathbf{CGPP}$ and $\mathbf{CGPTPP}$ -- in which case we have $\mathbf{GPP} = \mathbf{CGPP}$ and $\mathbf{pGPP} = \mathbf{pCGPP} = \mathbf{CGPTPP}$ (see App.~\ref{app:complete-free} for discussion and proof).

\subsection{Characterization of pGPP and pGPTPP transformations}
\label{app:properGPP}

We now impose, in addition to Gibbs preservation, a definite causal structure, yielding the sets of $\mathbf{pGPP}$ and $\mathbf{pGPTPP}$ supermaps, respectively. We know that the Choi operator $S\in \Lcal(\Hcal_P \otimes \Hcal_I \otimes \Hcal_O \otimes \Hcal_F)$ of a proper (i.e., causally ordered) supermap $\mathcal{S}: \mathsf{CPTP}(I,O) \rightarrow \mathsf{CPTP}(P,F)$ satisfies (see Prop.~\ref{prop::supchan}) 
\begin{gather}
\label{eqn::caus_req}
    S\geq 0, \quad \tr[S] = d_Pd_O, \quad \mathcal{P}_{\textup{SC}}[S] = S.
\end{gather}
Having characterized the sets $\mathbf{GPP}$ and $\mathbf{GPTPP}$ above, we can now simply add these constraints to obtain a characterization of the sets $\mathbf{pGPP}$ and $\mathbf{pGPTPP}$ based on Prop.~\ref{prop:app-GPP-link} and Lem.~\ref{lem:app-GPTPP}, respectively. 

\begin{prop}[Projective characterization of $\mathbf{pGPP}$ supermaps]
\label{prop:app-pGPP}
    An operator $S \in \mathcal{L}(\mathcal{H}_P \otimes \mathcal{H}_I \otimes \mathcal{H}_O \otimes \mathcal{H}_F)$ is the Choi operator of a $\mathbf{pGPP}$ supermap $\mathcal{S}$ if and only if the following conditions hold:
    \begin{align}
        S &\geq 0, \quad \tr[S] = d_P d_O, \quad  \mathcal{P}_{\mathrm{SC}}[S] = S,\\
        \tau_P \star S_{PIOF} &= \tau_I \otimes G_{OF} \quad \text{and} \quad \tau_O \star G_{OF} = \tau_F.\label{eqn::pGPP1}
    \end{align}
    Note that, in~\eqref{eqn::pGPP1} we have added subscripts on $S$ and $G$ to clarify what spaces they are defined on, respectively. 
    \begin{proof}
       The proof follows from combining the characterization of $\mathbf{GPP}$ maps given in Prop.~\ref{prop:app-GPP-link} with causality requirement of~\eqref{eqn::caus_req}.
    \end{proof}
\end{prop}
Consequently, $\mathbf{pGPP}$ consists precisely of thermal superchannels, i.e., causally ordered supermaps that cannot be used to extract athermality from equilibrium inputs. Similarly, we obtain a characterization of $\mathbf{pGPTPP}$ by combining Lem.~\ref{lem:app-GPTPP} with the causality constraint of~\eqref{eqn::caus_req}.

\begin{prop}[Projective characterization of $\mathbf{pGPTPP}$ supermaps] \label{prop:app-pGPTPP}
    An operator $S \in \mathcal{L}(\mathcal{H}_P \otimes \mathcal{H}_I \otimes \mathcal{H}_O \otimes \mathcal{H}_F)$ is the Choi operator of a $\mathbf{pGPTPP}$ supermap $\mathcal{S}$ if and only if the following conditions hold:
     \begin{align}
        S &\geq 0, \quad \tr[S] = d_P d_O, \quad  \mathcal{P}_{\mathrm{SC}}[S] = S,\\
         \breve{S} &= \mathcal{P}_{\mathrm{GPTPP}}[\breve{S}], \label{eqn::pGPTPP4}\\
        \operatorname{Tr}_{PF}&[\breve{S} - {}_{I_\beta}\breve{S} - {}_{O_\beta}\breve{S} + 2{}_{I_\beta O_\beta}\breve{S}] = \ident_{IO}, \label{eqn::pGPTPP5}
    \end{align}
    where $\breve{S}$ is the Gibbs transform of $S$.
\end{prop}
While we have $\mathbf{pGPP} \neq \mathbf{pGPTPP}$ (see App.~\ref{app:proofHierarchy}), this distinction between different sets of equilibrium superchannels vanishes when requiring Gibbs preservation to hold in a \textit{complete} sense, i.e., $\mathbf{pCGPP} = \mathbf{pCGPTPP}$ (see App.~\ref{app:complete-free} for discussion and proof).

\section{Hierachy of classes of equilibrium supermaps: Proof of Thm.~\ref{thm:hierarchy}}
\label{app:proofHierarchy}
Here, we discuss the relationship of the different possible sets of equilibrium supermaps stated in Thm.~\ref{thm:hierarchy} in more detail. We first note that the second part of the Theorem, i.e., $\mathbf{GPP}\setminus \mathbf{GPTPP} \neq \emptyset$ and $\mathbf{GPTPP}\setminus \mathbf{GPP} \neq \emptyset$ has already been shown explicitly in Sec.~\ref{subsec::GPTPP_pGPTPP}. To analyze the remaining statements, we start by providing a proof of the first part of the Theorem, i.e., the inclusions
\begin{gather} 
\mathbf{pGPP} \subset \mathbf{pGPTPP} \subset \mathbf{GPTPP}
\end{gather}
and demonstrate that they are strict. \\
\noindent \textit{Inclusions.} We first prove the inclusions. Let $\mathcal{S}\in \mathbf{pGPP}$. Then $\mathcal{S}$ is proper and maps every GP map to a GP map. Since every GPTP map is, in particular, GP, $\mathcal{S}$ maps GPTP maps to GP maps. Properness guarantees that channels are mapped to channels; hence GPTP maps are mapped to GPTP maps. Therefore
\begin{equation}    
    \mathbf{pGPP}\subseteq \mathbf{pGPTPP}.
\end{equation}
The inclusion
\begin{equation}
    \mathbf{pGPTPP}\subseteq \mathbf{GPTPP}
\end{equation}
is immediate, since $\mathbf{pGPTPP}$ is the subclass of $\mathbf{GPTPP}$ consisting of proper supermaps.

\medskip
\noindent
We now prove that both inclusions are strict, i.e., 
\begin{gather}
    \mathbf{pGPP}\subsetneq \mathbf{pGPTPP} \quad \text{and} \quad \mathbf{pGPTPP} \subsetneq \mathbf{GPTPP}.
\end{gather}
For the first part, consider the trace-and-replace supermap $\mathcal{S}_\mathsf{TR}$ of Ex.~\ref{example:TR}, with its action given by 
\begin{gather}
    (\mathcal{S}_\mathsf{\mathsf{TR}}[\mathcal{T}])[\rho] = \tr[\mathcal{T}[\ketbra{\Psi}{\Psi}]] \tr[\rho_P] \tau_F
\end{gather}
for all $\mathcal{T} \in \mathsf{CP}(I,O)$ and $\rho_P \in \Lcal(\Hcal_P)$ and some $\ketbra{\Psi}{\Psi} \in \Lcal(\Hcal_I)$. Now, let let $\mathcal{T} \in \mathsf{CPTP}(I,O)$. Then 
\begin{equation}
    \operatorname{Tr}\!\left[\mathcal{T}[|\Psi\rangle\langle\Psi|]\right] = 1,
\end{equation}
such that 
\begin{equation}
    \mathcal{S}_{\mathsf{TR}}[\mathcal{T}] = \overline{\mathcal{G}},
\end{equation}
where $\overline{\mathcal{G}}[\rho] = \operatorname{Tr}[\rho] \, \tau_{F}$ for all $\rho \in \mathcal{L}(\mathcal{H}_P)$ is the completely thermalizing channel. Thus $\mathcal{S}_{\mathsf{TR}}$ maps every channel to $\overline{\mathcal{G}}$ and is a proper supermap. In particular, it maps every GPTP map to a GPTP map, so
\begin{equation}
    \mathcal{S}_{\mathsf{TR}}\in \mathbf{pGPTPP}.
\end{equation}
However, $\mathcal{S}_{\mathsf{TR}}$ does \textit{not} preserve all GP maps. To see this, choose a positive operator $E$ on $\mathcal{H}_I$ such that
\begin{equation}
    \operatorname{Tr}[E\tau_{I}] = 1, \qquad \langle\Psi|E|\Psi\rangle\neq 1,
\end{equation}
and define the CP map $\mathcal{T}_\mathsf{E}$ such that
\begin{equation}
    \mathcal{T}_\mathsf{E}[\rho] = \operatorname{Tr}[E\rho]\,\tau_{O}.
\end{equation}
for any $\rho \in \mathcal{L}(\mathcal{H}_I)$. Then
\begin{equation}
    \mathcal{T}_\mathsf{E}[\tau_{I}] = \tau_{O},
\end{equation}
so $\mathcal{T}_\mathsf{E}$ is GP. But
\begin{equation}
    \mathcal{S}_{\mathsf{TR}}[\mathcal{T}_E] = \langle\Psi|E|\Psi\rangle\,\overline{\mathcal{G}}.
\end{equation}
Hence
\begin{equation}
    (\mathcal{S}_{\mathsf{TR}}[\mathcal{T}_\mathsf{E}])[\tau_{P}] = \langle\Psi|E|\Psi\rangle\,\tau_{F} \neq \tau_{F}.
\end{equation}
Therefore, $\mathcal{S}_{\mathsf{TR}}\notin\mathbf{pGPP}$, proving $\mathbf{pGPP}\subsetneq \mathbf{pGPTPP}$.

\medskip
\noindent
Similarly, to prove the strictness of the inclusion
\begin{gather}
    \mathbf{pGPTPP}\subsetneq\mathbf{GPTPP},
\end{gather}
choose a positive operator $A \geq 0$ on $\mathcal{H}_O$ such that
\begin{equation}
    \operatorname{Tr}[A\tau_{O}]=1, \qquad A\neq \ident_O.
\end{equation}
This is possible whenever $\dim\mathcal{H}_O > 1$. Define a supermap
\begin{equation}
\label{eq:witness-GPTPP-not-pGPTPP}
    (\mathcal{S}_\mathsf{A}[\mathcal{T}])[\rho_P] :=
    \operatorname{Tr}\!\left[A\,\mathcal{T}[\tau_{I}]\right]\,
    \operatorname{Tr}[\rho_P]\,\tau_{F}.
\end{equation}
The Choi operator of this supermap is of the form $S_\mathsf{A} = \ident_P \otimes \tau_I \otimes A^\mathrm{T} \otimes \tau_F \geq 0$ and thus $\mathcal{S}_\mathsf{A}$ is completely positive. It is easy to see that $\mathcal{S}_\mathsf{A} \in \mathbf{GPTPP}$: for any $\mathcal{T} \in \mathbf{GPTP}$, we have 
\begin{gather}
\begin{split}
    (\mathcal{S}_\mathsf{A}[\mathcal{T}])[\rho_P] &=
    \operatorname{Tr}\!\left[A\,\mathcal{T}[\tau_{I}]\right]\,
    \operatorname{Tr}[\rho_P]\,\tau_{F} \\
    &= \operatorname{Tr}\!\left[A\, \tau_{O}]\right]\, \tr[\rho_P] \tau_{F} = \tr[\rho_P] \tau_F, 
\end{split}
\end{gather}
i.e., $\mathcal{S}_\mathsf{A}[\mathcal{T}]$ is GPTP (in particular, $\mathcal{S}_\mathsf{A}$ maps every GPTP map onto the completely thermalizing channel $\overline{\mathcal{G}}$). Hence $\mathcal{S}_\mathsf{A}\in\mathbf{GPTPP}$.

\noindent However, $\mathcal{S}_\mathsf{A}$ is not a \textit{proper} supermap. Since $A\neq\ident_O$, there exists a quantum state $\sigma_O$ such that
\begin{equation}
    \operatorname{Tr}[A\sigma_O] \neq 1.
\end{equation}
Let $\mathcal{T}_\sigma$ be the trace-and-prepare channel
\begin{equation}
    \mathcal{T}_\sigma[\rho_I] := \operatorname{Tr}[\rho_I]\,\sigma_O.
\end{equation}
Then $\mathcal{T}_\sigma \in \mathsf{CPTP}(I,O)$, but
\begin{equation}
    (\mathcal{S}_\mathsf{A}[\mathcal{T}_\sigma])[\rho_P] = \operatorname{Tr}[A\sigma_O]\, \operatorname{Tr}[\rho_P]\,\tau_{F},
\end{equation}
which is not trace-preserving because $\operatorname{Tr}[A\sigma_O] \neq 1$. Hence $\mathcal{S}_\mathsf{A}$ does not map all channels to channels, so it is not proper. Therefore
\begin{equation}
    \mathcal{S}_\mathsf{A} \notin \mathbf{pGPTPP} \quad \Rightarrow \quad  \mathbf{pGPTPP} \subsetneq \mathbf{GPTPP}.
\end{equation}

\medskip
\noindent
For the second part of Thm.~\ref{thm:hierarchy}, we show that 
\begin{gather}
    \mathbf{GPP}\cap\mathbf{pGPTPP} = \mathbf{pGPP}
\end{gather}
holds. First, if $\mathcal{S}\in\mathbf{pGPP}$, then $\mathcal{S}$ is proper and $\mathbf{GPP}$. Moreover, since every GPTP map is GP and proper supermaps preserve trace preservation, $\mathcal{S}$ maps GPTP maps to GPTP maps. Hence
\begin{equation}
    \mathcal{S} \in \mathbf{GPP}\cap\mathbf{pGPTPP},
\end{equation}
and thus
\begin{equation}
    \mathbf{pGPP}\subseteq \mathbf{GPP}\cap\mathbf{pGPTPP}.
\end{equation}
Conversely, let
\begin{equation}
    \mathcal{S} \in \mathbf{GPP}\cap\mathbf{pGPTPP}.
\end{equation}
Since $\mathcal{S} \in \mathbf{GPP}$, it maps GP maps to GP maps. Since $\mathcal{S} \in \mathbf{pGPTPP}$, it is proper. Therefore $\mathcal{S}$ is both proper and $\mathbf{GPP}$, which is precisely the definition of $\mathbf{pGPP}$. Thus
\begin{equation}
    \mathbf{GPP} \cap \mathbf{pGPTPP} \subseteq \mathbf{pGPP}.
\end{equation}
Therefore
\begin{equation}
    \mathbf{GPP}\cap\mathbf{pGPTPP} = \mathbf{pGPP}.
\end{equation}

\section{Characterization of completely Gibbs preserving supermaps}
\label{app:complete-free}

In this Appendix we prove the characterization of the sets of completely Gibbs preserving supermaps, i.e., the sets $\mathbf{CGPP}, \mathbf{pCGPP}, \mathbf{CGPTPP},$ and $\mathbf{pCGPTPP}$ discussed in Sec.~\ref{subsec:completeSupermaps} in the main text. Throughout, all Gibbs states are assumed to be full rank.

\subsection{Proof of Thm.~\ref{theo:ComplGPP}: Completeness of GPP and pGPP transformations}
\label{app:completenessGPP}
First, we show that the additional requirement of \textit{complete} Gibbs preservation does not impose any additional restrictions on $\mathbf{GPP}$, the set of supermaps that preserve Gibbs preserving maps, i.e., we have 
\begin{gather}
      \mathbf{CGPP} = \mathbf{GPP} \quad \text{and} \quad \mathbf{pCGPP} = \mathbf{pGPP}.    
\end{gather}
For the proof, we first note that the former equality implies the latter, so we can restrict ourselves to demonstrating that $\mathbf{CGPP} = \mathbf{GPP}$ holds. Now, recall that, a supermap $\mathcal{S}:\mathsf{CP}(I,O) \rightarrow \mathsf{CP}(P,F)$ is CGPP iff it maps GP maps onto GP maps, even when only acting on a part of them. Expressed in terms of Choi operators $S \in \Lcal(\Hcal_P \otimes \Hcal_I \otimes \Hcal_O \otimes \Hcal_F)$ this requirement translates to  
\begin{gather}
    S_{PIOF}\star G_{IaOb} =:  G'_{PaFb} \in \mathsf{GP}(P\otimes a, F\otimes b)
\end{gather}
for all $G_{IaOb} \in \mathsf{GP}(I\otimes a, O \otimes b)$, where $a$ and $b$ denote auxiliary degrees of freedom that $S$ does not act on (see Fig.~\ref{fig::CompletenessAux} for reference). Note that $G'_{PaFb} \in \mathsf{GP}(P\otimes a, F\otimes b)$ is equivalent to 
\begin{gather}
    G'_{PaFb} \geq 0 \quad \text{and} \quad  G'_{PaFb} \star (\tau_P \otimes \tau_{a}) = \tau_F \otimes \tau_{b}, 
\end{gather}
where $\tau_a$ and $\tau_b$ are the Gibbs states on the auxiliary spaces, respectively.

\noindent Now, it is easy to see that we immediately obtain
\begin{gather}
    \mathbf{CGPP} \subseteq \mathbf{GPP}
\end{gather}
by choosing trivial auxiliary systems $a$ and $b$. 

\noindent For the converse inclusion, let $\mathcal{S}\in\mathbf{GPP}$. By Prop.~\ref{prop:app-GPP-link}, there exists a $M \in\mathsf{GP}(O,F)$ such that
\begin{gather}
    \tau_{P}\star S_{PIOF} = \tau_{I}\otimes M_{OF}, \label{eq:CGPP-proof-factor}  \quad \text{and} \quad 
    \tau_{O}\star M_{OF} = \tau_{F}.
\end{gather}
Using the factorization above, we can directly verify that $G'_{PaFb} =  S_{PIOF}\star G_{IaOb}$ is GP. In particular, we have 
\begin{align}
\notag
    &G'_{PaFb} \star(\tau_P \otimes \tau_a) \\
\notag    &\quad= S_{PIOF}\star G_{IaOb} \star (\tau_P \otimes \tau_a) = (\tau_P \star S_{PIOF}) \star G_{IaOb} \star \tau_a \\
    \notag&\quad= (\tau_I \otimes M_{OF}) \star G_{IaOb} \star \tau_a = M_{OF} \star G_{IaOb} \star (\tau_I \otimes \tau_a) \\
   &\quad= M_{OF} \star (\tau_O \otimes \tau_b) =  \tau_F \otimes \tau_b,
\end{align}
where we have used the commutativity and associativity of the link product, as well as the fact that $S\in \mathbf{GPP}$, $G \in \mathsf{GP}(I\otimes a, O\otimes b)$ and $M \in \mathsf{GP}(O,F)$. From the above, we see that $S_{PIOF}\star G_{IaOb} \in \mathsf{GP}(P\otimes a, F\otimes b)$ whenever $G_{IaOb} \in \mathsf{GP}(I\otimes a, O\otimes b)$, implying that $S \in \mathbf{CGPP}$. Consequently, we have the inclusion
\begin{gather}
    \mathbf{GPP} \subseteq \mathbf{CGPP},
\end{gather}
which, together with $\mathbf{CGPP} \subseteq \mathbf{GPP}$ proves $\mathbf{CGPP} = \mathbf{GPP}$. Imposing, in addition, that the corresponding supermaps are proper (i.e., they are superchannels), then yields the equality of sets 
\begin{gather}
    \mathbf{pGPP} =\mathbf{pCGPP}.
\end{gather}

\subsection{Proof of Thm.~\ref{theo:ClassesCollapse}: Complete GPTP-preservation enforces thermality and causal ordering}
\label{app:completenessCausal}
Here, we demonstrate that every supermap that transforms GPTP maps into GPTP maps, even when only acting on a part of them, is both thermal and causally ordered, i.e., we show the set equalities 
\begin{gather}
    \mathbf{CGPTPP} = \mathbf{pCGPTPP} = \mathbf{pCGPP} = \mathbf{pGPP}.
\end{gather}
Above, we have already shown that $ \mathbf{pCGPP} = \mathbf{pGPP}$, such that we can focus on the remaining equalities here. 

\noindent To this end, first, we note that, if a supermap $S$ is a superchannel, it maps CPTP maps to CPTP maps, even when only acting on a part of them~\cite{Chiribella2009}. Consequently, any $S\in \mathbf{pGPP}$ maps GPTP maps to GPTP maps even when only acting on a part of them. Hence
\begin{equation}
    \mathbf{pGPP}\subseteq \mathbf{pCGPTPP}\subseteq \mathbf{CGPTPP}.
\end{equation}
It remains to show that $\mathbf{CGPTPP}\subseteq \mathbf{pGPP}$, i.e., that any completely GPTP preserving supermap $S$ is a superchannel (i.e., it is causally ordered) and it is GPP, i.e., it satisfies 
\begin{gather}
    \tau_{P} \star S_{PIOF} = \tau_I \otimes G_{OF} \quad \text{and} \quad \tau_O \star G_{OF} = \tau_F.
\end{gather}

To show both of these properties, assume $\mathcal{S}\in\mathbf{CGPTPP}$. The proof now proceeds by exploiting the additional constraints imposed by \textit{complete} GPTP preservation. By attaching auxiliary input and output systems and considering the corresponding extended supermap $S^\#$, we obtain a family of independent constraints on the Gibbs transform of $S$. We will show that these constraints imply both the thermality condition characterizing $\mathbf{GPP}$ and the causal constraints characterizing superchannels. Specifically, for all auxiliary systems $a,b$, the extended supermap 
\begin{gather}
    \mathcal{S}^\# := \mathcal{S}\otimes\mathcal{I}_{a'\rightarrow a} \otimes \mathcal{I}_{b\rightarrow b'},
\end{gather} 
maps $\mathsf{GPTP}(I\otimes a,O\otimes b)$ into $\mathsf{GPTP}(P\otimes a',F\otimes b')$, where $\Hcal_{a'}\cong \Hcal_a$ and $\Hcal_{b'}\cong \Hcal_b$ are isomorphic output copies of the auxiliary systems. The Choi operator $S^{\#}$ corresponding to $\mathcal{S}^\#$ is given by 
\begin{gather}
S_{PIOFa'ab'b}^\# := S\otimes \kketbra{\ident}{\ident}_{aa'} \otimes \kketbra{\ident}{\ident}_{bb'},
\end{gather}
where $\kket{\ident}_{xx'} \in \Hcal_{x} \otimes \Hcal_{x'}$ is the unnormalised maximally entangled state. Using this factorized form of $S^{\#}$, we obtain its corresponding Gibbs transform (by making the replacements $P \rightarrow Pa'$, $I \rightarrow Ia$, $O\rightarrow Ob$ and $F \rightarrow Fb'$ in Def.~\ref{def:rotatedSupermap}) as
\begin{align}
    \nonumber \breve{S}_{PIOF a' a b b'}^\# &= (\sqrt{\tau_{P}}\otimes\sqrt{\tau_{I}}^{-1})S_{PIOF}(\sqrt{\tau_{P}}\otimes\sqrt{\tau_{I}}^{-1}) \\
    \nonumber &\qquad \otimes (\sqrt{\tau_{a'}} \otimes \sqrt{\tau_{a}}^{-1}) \kketbra{\ident}{\ident}_{a'a} (\sqrt{\tau_{a'}} \otimes \sqrt{\tau_{a}}^{-1}) \\
    &\qquad \otimes \kketbra{\ident}{\ident}_{bb'} \\
    &= \breve{S}_{PIOF} \otimes \kketbra{\ident}{\ident}_{a'a} \otimes \kketbra{\ident}{\ident}_{bb'},
\end{align}
where we used the identity $(A \otimes B) \kket{\ident} = (AB^{\mathrm{T}} \otimes \ident) \kket{\ident}$, together with the fact that the Gibbs states are diagonal in the chosen basis and identical on isomorphic auxiliary systems.
\begin{figure}[t!]
    \centering
    \includegraphics[width = 0.9\linewidth]{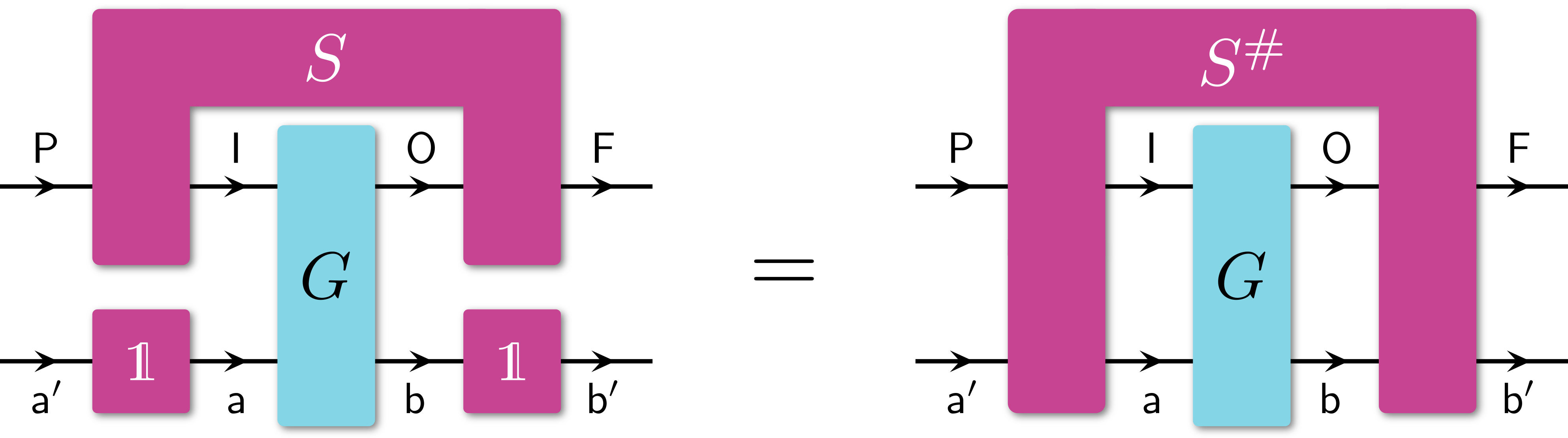}
    \caption{\textbf{Action of a supermap on a subsystem.}
    The action of a supermap $\mathcal{S}: \mathsf{CP}(I,O) \rightarrow \mathsf{CP}(P,F) $ with corresponding Choi operator $S \in \Lcal(\Hcal_P \otimes \Hcal_I \otimes \Hcal_O \otimes \Hcal_F)$ on a map $\mathcal{G} \in \mathsf{CP}(I\otimes a, O\otimes b)$ with corresponding Choi operator $G \in \Lcal(\Hcal_I \otimes \Hcal_a \otimes \Hcal_O \Hcal_b)$ can be equivalently described by the extended supermap $\mathcal{S}^{\#} = \mathcal{S} \otimes \mathcal{I}_{a'\rightarrow a} \otimes \Ical_{b\rightarrow b'}$ with corresponding Choi operator $S^\# = S\otimes \kketbra{\ident}{\ident}_{aa'} \otimes \kketbra{\ident}{\ident}_{bb'}$. The systems $a'$ and $b'$ denote isomorphic copies of $a$ and $b$, respectively.}
    \label{fig::CompletenessAux}
\end{figure}

Since $\breve S^{\#}$ maps $\mathsf{GPTP}_{\tau_I \otimes \tau_a}(I\otimes a, O\otimes b)$ into $\mathsf{GPTP}_{\tau_P \otimes \tau_{a'}}(P\otimes a', F\otimes b')$, we can apply Lem.~\ref{lem::FullyGenProj} to deduce the properties of $\breve S^{\#}$. In particular, the respective input and output projectors onto the sets $\mathsf{GPTP}_{\tau_I \otimes \tau_a}(I\otimes a, O\otimes b)$ are given by (see Prop.~\ref{prop::GPTP}
\begin{align}
    \Pcal^{\rm i}[\widetilde{G}] &= \widetilde{G} - {}_{I_\alpha a_\alpha}\widetilde{G} - {}_{O_\alpha b_\alpha }\widetilde{G} + 2{}_{I_\alpha a_\alpha O_\alpha b_\alpha}\widetilde{G} \\
    \text{and} \quad  \Pcal^{\rm o}[\widetilde{G}'] &= \widetilde{G}' - {}_{P_\alpha a'_\alpha}\widetilde{G}' - {}_{F_\alpha b'_\alpha }\widetilde{G}' + 2{}_{P_\alpha a'_\alpha F_\alpha b'_\alpha}\widetilde{G}'
\end{align}
Using Eq.~\eqref{eqn::FullyGenProj1} from  Lem.~\ref{lem::FullyGenProj} we then obtain
\begin{align}
\notag
    \breve{S} &= \breve{S} +\left(-{}_{P_\alpha}\breve{S} + {}_{P_\alpha I_\beta}\breve{S}\right) \otimes \ident_a \otimes \tau_{a'} \otimes \kketbra{\ident}{\ident}_{bb'} \\
    \notag
    &\phantom{=,}+\left(-{}_{F_\alpha}\breve{S}+ {}_{F_\alpha O_\beta}\breve{S}\right)\otimes \kketbra{\ident}{\ident}_{aa'} \otimes \ident_b \otimes \tau_{b'}\\
    \notag
    &\phantom{=,}+ \left({}_{P_\alpha O_\beta}\breve{S} - 2{}_{P_\alpha I_\beta O_\beta}\breve{S} + {}_{I_\beta F_\alpha }\breve{S} - 2{}_{I_\beta O_\beta F_\alpha}\breve{S} \right.\\
    \notag &\phantom{=ba,}\left.+ 2{}_{P_\alpha F_\alpha}\breve{S}  - 2{}_{P_\alpha I_\beta F_\alpha}\breve{S} - 2{}_{P_\alpha O_\beta F_\alpha}\breve{S} \right.\\
    &\phantom{=ba,} \left.+ 4{}_{P_\alpha I_\beta O_\beta F_\alpha}\breve{S}\right)
    \otimes  \ident_{a} \otimes \tau_{a'} \otimes \ident_b \otimes \tau_{b'},
\end{align}
where we have sorted the expression according to the respective terms on the $aa'bb'$ subspace, and we have used that ${}_{x_\beta}\kketbra{\ident}{\ident}_{xx'} =  {}_{x'_\alpha}\kketbra{\ident}{\ident}_{xx'}$. Since the respective terms on the $aa'bb'$ subspace are linearly independent, the above yields
\begin{align}
    {}_{P_\alpha}\breve S &= {}_{P_\alpha I_\beta}\breve S, \label{eq:CGPTPP-id1}\\
    {}_{F_\alpha}\breve S &= {}_{F_\alpha O_\beta}\breve S,  \label{eq:CGPTPP-id2}\\
    \notag
    \text{and} \quad  0&={}_{P_\alpha O_\beta}\breve{S} - 2{}_{P_\alpha I_\beta O_\beta}\breve{S} + {}_{I_\beta F_\alpha }\breve{S} - 2{}_{I_\beta O_\beta F_\alpha}\breve{S} + 2{}_{P_\alpha F_\alpha}\breve{S}  \\
    &\phantom{=,}- 2{}_{P_\alpha I_\beta F_\alpha}\breve{S} - 2{}_{P_\alpha O_\beta F_\alpha}\breve{S} + 4{}_{P_\alpha I_\beta O_\beta F_\alpha}\breve{S}.  \label{eq:CGPTPP-id3}
\end{align}
Inserting Eqs.~\eqref{eq:CGPTPP-id1} and~\eqref{eq:CGPTPP-id2} into Eq.~\eqref{eq:CGPTPP-id3} further simplifies the latter to 
\begin{gather}
     {}_{P_\alpha I_\beta O_\beta}\breve{S} +  {}_{I_\beta O_\beta F_\alpha}\breve{S}  = 2{}_{P_\alpha I_\beta O_\beta F_\alpha} {}\breve{S},
\end{gather}
which implies 
\begin{align}
\label{eqn::aux_eq_1}
    {}_{P_\alpha I_\beta O_\beta}\breve{S} &= {}_{P_\alpha I_\beta O_\beta F_\alpha} {}\breve{S}\\
    \text{and} \quad {}_{I_\beta O_\beta F_\alpha}\breve{S} &= {}_{P_\alpha I_\beta O_\beta F_\alpha} {}\breve{S}
    \label{eqn::aux_eq_2}
\end{align} 
With this, we can show that that $\mathcal{S}\in\mathbf{GPP}$. Eq.~\eqref{eq:CGPTPP-id1} implies, after undoing the Gibbs transform,
\begin{equation}
    \tau_P \star S_{PIOF} = \tau_I \otimes G_{OF},
\end{equation}
where $G_{OF} := \tr_{PI}[\tau_P S]$. From Eq.~\eqref{eqn::aux_eq_1}, we see that 
\begin{equation}
    \tr_{PIO}[(\tau_P \otimes \tau_O)S] =  \tr[(\tau_P \otimes \tau_O)S] \tau_F.
\end{equation}
Since $\tr_{PIO}[(\tau_P \otimes \tau_O)S] = \tau_O \star G_{OF}$, this implies 
\begin{gather}
     \tau_O \star G_{OF} = c\tau_F,
\end{gather}
where $c = \tr[(\tau_P \otimes \ident_I \otimes \tau_O \otimes \tau_F)S]$. To fix $c$, we note that $(\ident_I \otimes \tau_O)$ is the Choi operator of the completely thermalizing channel
\begin{equation}
    \overline{\mathcal{G}}[\rho] = \operatorname{Tr}[\rho]\tau_O,
\end{equation}
which is GPTP. Consequently,  
\begin{gather}
    c = \tr[(\tau_P \otimes \ident_I \otimes \tau_O \otimes \tau_F)S] = \tr[\mathcal{S}[\overline{\mathcal{G}}][\tau_P]] = 1,
\end{gather}
where we have used that $\mathcal{S}\in \mathbf{CGPTPP}$, and thus $\mathcal{S}[\overline{\mathcal{G}}]$ is TP. As a result, $c=1$, and we obtain
\begin{equation}
    \tau_O \star G_{OF} = \tau_F.
\end{equation}
$S$ thus satisfies \eqref{eq:app-GPP-link1}--\eqref{eq:app-GPP-link2}, making it the Choi operator the Choi operator of a $\mathbf{GPP}$ supermap by Prop.~\ref{prop:app-GPP-link}, i.e., $\mathbf{CGPTPP} \subseteq \mathbf{GPP}$.\\
\noindent It remains to show that $S$ is proper, i.e., causally ordered. First, undoing the Gibbs transform in Eq.~\eqref{eq:CGPTPP-id2}, we directly obtain
\begin{gather}
\label{eqn::caus_proof_1}
    \tr_F[S] = \tr_{FO}[\tau_O S] \otimes \ident_O =: S_{PI} \otimes \ident_O.
\end{gather}
In addition, from Eq.~\eqref{eqn::aux_eq_2}, we see that 
\begin{align}
\notag
    {}_{I_\beta O_\beta F_\alpha}\breve{S} &= \tr_{IOF}[\tau_O S] \otimes \ident_{IO} \otimes \tau_{F_\alpha} \\
    \notag
    &= \sqrt{\tau_P}\tr_I[S_{PI}]\sqrt{\tau_P} \otimes \ident_{IO} \otimes \tau_{F} \\
    &={}_{P_\alpha I_\beta O_\beta F}\breve{S} = \tr[(\tau_P \otimes \tau_O)S]\tau_P \otimes \ident_{IO} \otimes \tau_{F}, 
\end{align}
which implies $\tr_I[S_{PI}] = \tr[(\tau_P \otimes \tau_O)S]\ident_P =: g\ident_P$. It remains to show that $g=1$ holds. To this end, we use that, from Eq.~\eqref{eqn::FullyGenProj2}, the second part of Lem.~\ref{lem::FullyGenProj}, we obtain after some straight forward algebra
\begin{gather}
{}_{P_\alpha F_\alpha} \breve{S} - {}_{P_\alpha I_\beta F_\alpha}\breve{S} - {}_{P_\alpha O_\beta F_\alpha}\breve{S} + 2{}_{P_\alpha I_\beta O_\beta F_\alpha}\breve{S}= \tau_P \otimes \ident_{IO} \otimes \tau_F,
\end{gather}
where we have used $\gamma_{\rm i} = \gamma_{\rm{o}} = 1$. Applying ${}_{P_\alpha I_\beta O_\beta F_\alpha}\bullet$ on both sides yields 
\begin{gather}
    {}_{P_\alpha I_\beta O_\beta F_\alpha}\breve{S}= \tau_P \otimes \ident_{IO} \otimes \tau_F,
\end{gather}
and thus we have $g =  \tr[(\tau_P \otimes \tau_O)S] = 1$. Combining the above results, we have shown that 
\begin{align}
    \operatorname{Tr}_F[S_{PIOF}] &= S_{PI}\otimes\ident_O, \\
    \operatorname{Tr}_I[S_{PI}] &= \ident_P,
\end{align}
which is equivalent to the causality conditions \eqref{eqn::UTP2}--\eqref{eqn::UTP_proj} of superchannels. Hence $S$ is a superchannel. Therefore
\begin{equation}
    \mathbf{CGPTPP}\subseteq \mathbf{pGPP}.
\end{equation}
Taking into account Thm.~\ref{theo:ComplGPP} (i.e., $\mathbf{pCGPP} = \mathbf{pGPP}$), we thus conclude that $\mathbf{CGPTPP} = \mathbf{pCGPTPP} = \mathbf{pCGPP} = \mathbf{pGPP}$.

\section{Realization of pGPP transformations}

\subsection{Proof of Thm.~\ref{theo:encdecpGPP}: Superchannels from GPTP maps}
\label{app:encdecpGPPproof}
Here, we show that a superchannel $\mathcal{S}[\mathcal{T}] = \mathcal{D} \circ \mathcal{T} \circ \mathcal{E}$ lies in $\mathbf{pGPP}$ if $\mathcal{E} \in \textsf{GPTP}(P, I \otimes c)$ and $\mathcal{D} \in \textsf{GPTP}(O\otimes c,F)$. To this end, we first note that maps $\mathcal E$ and $\mathcal D$ are CPTP, therefore $\mathcal{S}$ -- with corresponding Choi operator $S_{PIOF} = E_{PIc} \star D_{cOF}$ -- is a superchannel. In turn, since $\mathcal E \ \in \mathsf{GP}(P, I \otimes c)$, we have
\begin{equation}
    \tau_P\star E_{PIc} = \tau_I\otimes\tau_c,
\end{equation}
and therefore
\begin{align}
    \nonumber \tau_P\star S_{PIOF} &= \tau_P\star(E_{PIc}\star D_{cOF}) = (\tau_P\star E_{PIc})\star D_{cOF} \\
    \notag&= (\tau_I\otimes\tau_c)\star D_{cOF} = \tau_I\otimes(\tau_c\star D_{cOF})\\
    &=: \tau_I \otimes G_{OF}.
\end{align}
Since $\mathcal{D} \in \textsf{GPTP}(O\otimes c,F)$, the matrix $G_{OF}$ defined above is a Gibbs preserving map. Consequently, we have (see Prop.~\ref{prop:app-GPP-link}) $\mathcal{S} \in \mathbf{GPP}$, which, together with the fact that $\mathcal{S}$ is proper implies $\mathcal{S} \in \mathbf{pGPP}$.

\subsection{Proof of Thm.~\ref{theo:decomppGPP}: Decomposition of a pGPP supermap}
\label{app:decomppGPPproof}
Here, we show that, if $\mathcal{S}\in\mathbf{pGPP}$, with $\mathcal{S}[\mathcal{T}] = \mathcal{D} \circ \mathcal{T} \circ \mathcal{E}$, then $\mathcal{D} \in \mathsf{GPTP}(O\otimes c,F)$ and $\tr_c[\mathcal{E}[\tau_P]] = \tau_I$ as well as $\tr_I[\mathcal{E}[\tau_P]] = \tau_c$. To show this, we first recall that, if $\mathcal{S}\in\mathbf{pGPP}$, then there exists a positive operator $G_{OF}$ such that
\begin{align}
    \tau_P\star S_{PIOF} &= \tau_I\otimes G_{OF}, \label{eq:encdec-pGPP-factor}\\
    \text{and} \quad \tau_O\star G_{OF} &= \tau_F. \label{eq:encdec-pGPP-GP}
\end{align}
Using the decomposition of a superchannel,
\begin{align}
    \nonumber \tau_P\star S_{PIOF} &= \tau_P\star(E_{PIc}\star D_{cOF}) \\&
    = (\tau_P\star E_{PIc})\star D_{cOF} =: \xi_{Ic}\star D_{cOF},
\end{align}
where $\xi_{Ic} = \mathcal{E}[\tau_P]$. Combining the above with \eqref{eq:encdec-pGPP-factor} gives
\begin{equation}
    \xi_{Ic}\star D_{cOF} = \tau_I\otimes G_{OF}.
\end{equation}
Taking the partial trace over $F$ and using that $\mathcal{D}$ is trace-preserving,
\begin{equation}
    \operatorname{Tr}_F[D_{cOF}] = \ident_{cO},
\end{equation}
we obtain
\begin{equation}
    \operatorname{Tr}_F[\xi_{Ic}\star D_{cOF}] = \operatorname{Tr}_c[\xi_{Ic}]\otimes\ident_O.
\end{equation}
On the other hand,
\begin{equation}
    \operatorname{Tr}_F[\tau_I\otimes G_{OF}] = \tau_I\otimes\ident_O,
\end{equation}
since $G_{OF}$ is trace-preserving as is stems from a superchannel. Hence
\begin{equation}
    \operatorname{Tr}_c[\xi_{Ic}] = \tau_I.
\end{equation}
Now define
\begin{equation}
\label{eqn::tau_c}
    \tau_c := \operatorname{Tr}_I[\xi_{Ic}].
\end{equation}
Applying a partial trace over $I$ on both sides of 
\begin{equation}
    \xi_{Ic}\star D_{cOF} = \tau_I\otimes G_{\mathcal{S},OF}
\end{equation}
gives
\begin{equation}
    \tau_c\star D_{cOF} = G_{OF}.
\end{equation}
Using \eqref{eq:encdec-pGPP-GP}, we conclude
\begin{equation}
    (\tau_O\otimes\tau_c)\star D_{cOF} = \tau_O\star G_{OF} = \tau_F,
\end{equation}
i.e., $\mathcal{D} \in \mathsf{GPTP}(O \otimes c,F)$. Finally, we note that, $\tau_c$ as defined in Eq.~\eqref{eqn::tau_c} can indeed be considered a Gibbs state on the system $c$. In particular, if $\tau_c$ is full rank, then it can be written as a Gibbs state
\begin{equation}
    \tau_c=\frac{e^{-\beta H_c}}{\operatorname{Tr}[e^{-\beta H_c}]}
\end{equation}
for a Hamiltonian
\begin{equation}
    H_c = -\beta^{-1} \log\tau_c + \lambda\ident_c,
\end{equation}
where $\lambda$ fixes an energy scale.

\section{Multi-slot equilibrium transformations}
\label{app:multislot}
\noindent In this Appendix we prove the characterization of equilibrium transformations, i.e., completely Gibbs-and-trace-preserving $n$-slot maps $S^{(n)}\in \mathbf{CGPTPP}^{(n)}$.

\subsection{Preliminaries}
\noindent Recall that, in the main text, such $n$-slot maps $\mathcal{S}^{(n)} \in \mathbf{CGPTPP}^{(n)}(0, 1, \dots, 2n, 2n+1)$ were defined as the maps that transform any ${\mathcal{S}}^{(n-1)} \in \mathbf{CGPTPP}^{(n-1)}(1 \otimes a, \dots,2n-1, 2n \otimes b)$ into a GPTP map $\mathcal{S}^{(0)} \in \mathsf{GPTP}(0\otimes a, 2n+1 \otimes b) \equiv \mathbf{CGPTPP}^{(0)}(1 \otimes a, 2n \otimes b)$ even when only acting on a part of them (see Fig.~\ref{fig::multi_app}). In terms of the corresponding Choi operators, we can succinctly phrase this as 
\begin{gather}
\label{eqn::defSn}
    S^{(n)} \star S^{(n-1)} =: S^{(0)} \in \mathsf{GPTP}(0\otimes a, 2n+1 \otimes b)
\end{gather}
for all $S^{(n)}\in \mathbf{CGPTPP}^{(n)}$ and $S^{(n-1)} \in \mathbf{CGPTPP}^{(n-1)}(1 \otimes a, \dots,2n-1, 2n \otimes b)$. Evidently, unlike in the previous sections, here, we do not label the respective spaces by $P,I,O,$ and $F$, but rather by arabic numerals (see Fig.~\ref{fig::multi_app}). In addition, throughout, we make the spaces that maps in the respective sets $\mathsf{GPTP},  \mathbf{CGPTPP}^{(n)}$ and $\mathbf{CGPTPP}^{(n-1)}$ are defined on explicit when necessary, and oftentimes $S^{(n-1)}$ will start on space $1$ or $1\otimes a$ and end on $2n$ or $2n \otimes b$. Whenever it is helpful to make this distinction to $S^{(n)}$, which begins on $0$ and ends on $2n+1$ manifest, we will denote the latter by $R^{(n)}$ and the resulting map $R^{(n)} \star S^{(n-1)}$ by $T^{(0)}$ (see Sec.~\ref{app::Comp_n_slot} below).
\begin{figure}[t!]
    \centering
    \includegraphics[width = \linewidth]{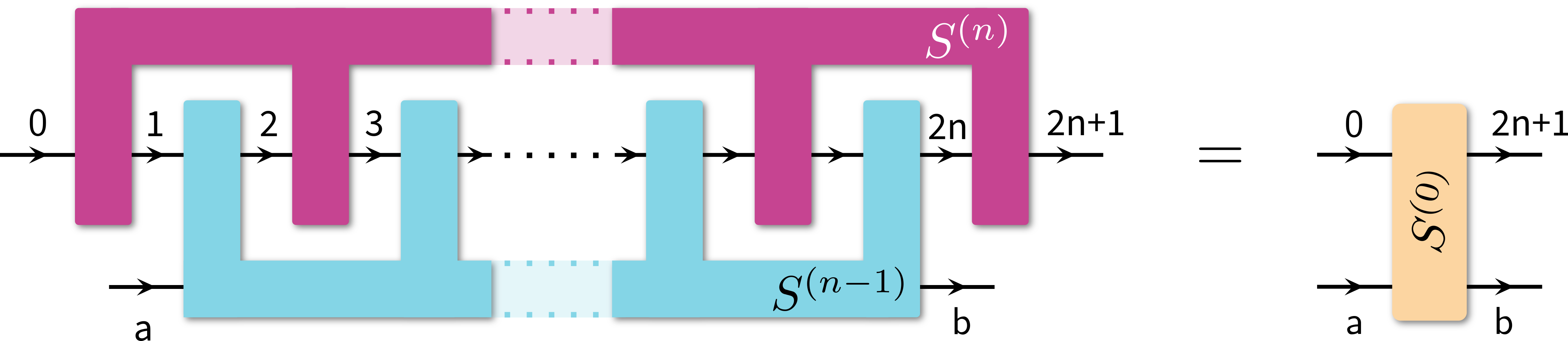}
    \caption{\textbf{CGPTPP property}. An $n$-slot supermap $S^{(n)} \in \Lcal(\Hcal_0 \otimes \Hcal_1 \otimes \cdots \Hcal_{2n+1})$ satisfies $S^{(n)} \in \mathbf{CGPTPP}^{(n)}(0,1,\dots, 2n+1)$ iff $S^{(0)} := S^{(n)} \star S^{(n-1)} \in \textsf{GPTP}(0\otimes a, 2n+1 \otimes b)$ for all $S^{(n-1)} \in  \mathbf{CGPTPP}^{(n-1)}(1\otimes a,\dots, 2n+1 \otimes b)$. Note that, in the proof of Prop.~\ref{prop::CGPTPP_n_slot}, we denote the $n$-slot supermap by $R^{(n)}$, with $n=k$.}
    \label{fig::multi_app}
\end{figure}

Above, we have already provided a characterization for the set $\mathbf{CGPTPP}^{(1)}(0,1,2,3)$, i.e., the set of superchannels that preserve GPTP maps in a complete sense. In particular, we showed in App.~\ref{app:complete-free} that $S^{(1)} \in \mathcal{L}(\mathcal{H}_{0} \otimes \mathcal{H}_{1} \otimes \mathcal{H}_2 \otimes \mathcal{H}_3)$ is a $1$-slot CGPTPP supermap iff it is causally ordered and thermal, i.e., 
\begin{align}
\label{eqn::relabelCGPP1}
    S^{(1)} &\geq 0, \; \tr[S^{(1)}] = d_0 d_2, \\
    \label{eqn::relabelCGPP2}
    S^{(1)} &=  S^{(1)} - {}_3S^{(1)} +  {}_{23}S^{(1)} \\
    &\phantom{=,}- {}_{123}S^{(1)} + {}_{0123}S^{(1)} = \Pcal^{(1)}_{\textup{SC}}[S^{(1)}],\\
    \label{eqn::thermal_arabic}
    \text{and} \quad \tau_0 \star S^{(1)} &= \tau_1 \otimes G_{23} \quad \& \quad \tau_2 \star G_{23} = \tau_3, 
\end{align}
where we have written the projector $\Pcal^{(1)}_{\textup{SC}}$ onto the (span of the) space of superchannels as well as the thermality constraints~\eqref{eqn::thermal_arabic} in agreement with the labels of spaces that we employ throughout this appendix (see Props.~\ref{prop::supchan} and~\ref{prop:app-pGPP} for the respective definitions in terms of the labels $P,I,O,F$). Recall that we can equivalently replace Eqs.~\eqref{eqn::relabelCGPP1} and~\eqref{eqn::relabelCGPP2} by 
\begin{gather}
\begin{split}
    S^{(1)} &\geq 0,\\
    \tr_3[S^{(1)}] &= S^{\prime (0)} \otimes \ident_2\quad \text{and} \quad \tr_2[S^{\prime(0)}] = \ident_0, 
\end{split}
\end{gather}
see Prop.~\ref{prop::alt_charac_caus}. Below, we will use both characterizations of $1$-slot causal superchannels.

\noindent Analogously, for the multi-slot case $S^{(n)} \in \Lcal(\Hcal_0 \otimes \Hcal_1 \otimes \cdots \otimes \Hcal_{2n} \otimes \Hcal_{2n+1})$ and it is causally ordered (with ordering $0\prec 1 \prec 2 \prec \cdots \prec 2n+1$ if (see Prop.~\ref{prop::n_slot_def})
\begin{align}
\label{eqn::causal_ord_app2}
S^{(n)}&\geq 0, \quad \tr[S^{(n)}] = d_0d_2\cdots d_{(2n)},\\
\label{eqn::caus_order_app}
   S^{(n)} &= S^{(n)} - {}_{(2n+1)}S^{(n)} +  {}_{(2n)(2n+1)}S^{(n)} \\
    &\phantom{=}- {}_{(2n-1)(2n)(2n+1)}S^{(n)} + \dots + {}_{12\dots (2n)(2n+1)}S^{(n)}\\
    &=  \Pcal^{(n)}_{\textup{SC}}[S^{(n)}].
\end{align}
Equivalently, $S^{(n)}$ is causally ordered (with ordering $0\prec 1 \prec 2 \prec \cdots \prec 2n+1$ (see Prop.~\ref{prop::alt_charac_caus}) iff
\begin{align}
    S^{(n)} &\geq 0, \\
    \tr_{2n+1} [S^{(n)}] &= S^{\prime (n-1)} \otimes \ident_{2n},\\
     \tr_{2n-1} [S^{\prime (n-1)}] &= S^{\prime (n-2)} \otimes \ident_{2n-2},\\
     \notag &\vdots\\ 
     \tr_{3}[S^{\prime(1)}] &= S^{\prime(0)}] \otimes \ident_2,\\
     \tr_{1}[S^{\prime(0)}] &= \ident_0.
\end{align}
We will both characterizations of causally ordered $n$-slot supermaps below.

\noindent Finally, based on the definition of thermality for a $1$-slot case (see Prop. \ref{prop:app-GPP-link}), we have the following generalization to the multi-slot case:
\begin{definition}[Thermal $n$-slot supermaps]
\label{def::thermal_n_slot}
We call an $n$-slot supermap $S^{(n)} \in \Lcal(\Hcal_0 \otimes \Hcal_1 \otimes \cdots \otimes \Hcal_{2n} \otimes \Hcal_{2n+1})$ thermal (with respect to the ordering $0\prec 1 \prec 2 \prec \cdots \prec 2n \prec 2n+1$) if it satisfies 
\begin{align}
    \tau_0 \star S^{(n)}_{01\cdots (2n)(2n+1)} &= \tau_1 \otimes S^{(n-1)}_{2\cdots (2n)(2n+1)}, \\
    \tau_2 \star S^{(n-1)}_{2\cdots (2n)(2n+1)} &= \tau_3 \otimes S^{(n-2)}_{4\cdots (2n)(2n+1)}, \\
    \notag &\vdots\\
     \tau_{2n} \star S^{(0)}_{(2n)(2n+1)} &= \tau_{2n+1}. 
\end{align}
\end{definition}
We will spend the rest of this appendix proving that $S^{(n)} \in \mathbf{CGPTPP}^{(n)}$ if and only if it is causally ordered and thermal with respect to the ordering $0\prec 1 \prec 2 \prec \cdots \prec 2n \prec 2n+1$. To do so, we require the following generalization of Lem.~\ref{lem::Stilde_S}, which states that the recursive composition defining $\mathbf{CGPTPP}^{(n)}$ is preserved by the corresponding Gibbs transformations:
\begin{lemma}[Equivalence under Gibbs transform -- $n$-slot case]
\label{lem::Stilde_S_gen}
    Let 
    \begin{align}
    \notag
        R^{\#(k)} &\in \mathcal{L}(\Hcal_a \otimes \Hcal_{a'} \otimes \Hcal_0 \otimes \Hcal_1 \otimes \cdots \\
        \notag &\phantom{\in.\mathcal{L}(} \cdots \otimes \Hcal_{2k} \otimes \Hcal_{2k+1} \otimes \Hcal_b \otimes \Hcal_{b'} ), \\ \notag S^{(k-1)} &\in \mathcal{L}(\Hcal_a \otimes \Hcal_1 \otimes \Hcal_2 \otimes \cdots \otimes \Hcal_{2k-1} \otimes \Hcal_{2k}\otimes \Hcal_b),\\
        \notag\text{and} \quad T^{(0)} &\in \mathcal{L}(\Hcal_{a'} \otimes \Hcal_0 \otimes \Hcal_{2k+1} \otimes \Hcal_{b'})
    \end{align}
    Then the following statements are equivalent:
    \begin{enumerate}
        \item $R^{\#(k)} \star S^{(k-1)} = T^{(0)}$,
        \item $\breve{R}^{\#(k)} \star \widetilde{S}^{(k-1)} = \widetilde{T}^{(0)}$,
    \end{enumerate}
    where, 
        \begin{align}
        \breve{R}^{\#(k)}&:= (\sqrt{\xi_{a'0a}} \otimes \sqrt{\tau_{1::2k-1}}^{-1})R^{\#(k)} (\sqrt{\xi_{a'0a}} \otimes \sqrt{\tau_{1::2k-1}}^{-1} ), \\
    \widetilde{S}^{(k-1)}&:= (\sqrt{\tau_a} \otimes \sqrt{\tau_{1::2k-1}}) S^{(k-1)} (\sqrt{\tau_a} \otimes \sqrt{\tau_{1::2k-1}}), \\
    \widetilde{T}^{(0)} &:= (\sqrt{\tau_{a'}} \otimes \sqrt{\tau_0}) T^{(0)} (\sqrt{\tau_{a'}} \otimes \sqrt{\tau_0}), 
\end{align}
with $\sqrt{\tau_{1::2k-1}} = \sqrt{\tau_1} \otimes \sqrt{\tau_3} \otimes \cdots \otimes \sqrt{\tau_{2k-1}}$ and $\sqrt{\xi_{a'0a}} = \sqrt{\tau_a'} \otimes \sqrt{\tau_0}\otimes \sqrt{\tau_a}^{-1}$. 
\end{lemma}
    \begin{proof}
        First, using the definition of the link product, one verifies that
        \begin{align}\label{eq:intertwining2}
           \nonumber &R^{\#(k)} \star S^{(k-1)} \\
           \nonumber&=  [(\sqrt{\xi_{a'0a}}^{-1} \otimes \sqrt{\tau_{1::2k-1}}) \breve{R}^{\#(k)} (\sqrt{\xi_{a'0a}}^{-1} \otimes \sqrt{\tau_{1::2k-1}}^{-1} )] \\
           \notag &\phantom{=,}\star (\sqrt{\tau_a}^{-1} \otimes \sqrt{\tau_{1::2k-1}})^{-1} \widetilde{S}^{(k-1)} (\sqrt{\tau_a}^{-1} \otimes \sqrt{\tau_{1::2k-1}}^{-1})\\
           &= (\sqrt{\tau_{a'}}^{-1} \otimes \sqrt{\tau_0}^{-1})(\breve{R}^{\#(k)} \star \widetilde{S}^{(k-1)})  (\sqrt{\tau_{a'}}^{-1} \otimes \sqrt{\tau_0}^{-1}).
        \end{align}
Consequently, if $R^{\#(k)} \star S^{(k-1)} = T^{(0)}$, then the above implies 
\begin{gather}
\begin{split}
    \breve{R}^{\#(k)} \star \widetilde{S}^{(k-1)} &=  (\sqrt{\tau_{a'}} \otimes \sqrt{\tau_0}) T^{(0)}  (\sqrt{\tau_{a'}} \otimes \sqrt{\tau_0}) \\
    &= \widetilde{T}^{(0)}. 
\end{split}
\end{gather}
Similarly, if $\breve{R}^{\#(k)} \star \widetilde{S}^{(k-1)} = \widetilde{T}^{(0)}$, then we obtain from Eq.~\eqref{eq:intertwining2}:
\begin{align}
    \notag R^{\#(k)} \star S^{(k-1)} &= (\sqrt{\tau_{a'}}^{-1} \otimes \sqrt{\tau_0}^{-1})\widetilde{T}^{(0)}  (\sqrt{\tau_{a'}}^{-1} \otimes \sqrt{\tau_0}^{-1})\\ &= T^{(0)}.
\end{align}
\end{proof}

\subsection{Causally ordered thermal \texorpdfstring{$n$}{n}-slot supermaps} 
To derive a complete characterization of $\mathbf{CGPTPP}^{(n)}$ -- and thus give a proof of Prop.~\ref{prop::multi-slot} from the main text, we now proceed as follows. First, we deduce a characterization of processes $S^{(n)}$ that are both causally ordered and thermal (according to Prop.~\ref{prop::n_slot_def} and Def.~\ref{def::thermal_n_slot} by means of a trace constraint and a projector (see Lem.~\ref{lem::supermap_k}). This Lemma then provides a helpful characterization for $\mathbf{CGPTPP}^{(0)}$ and $\mathbf{CGPTPP}^{(1)}$, which we can, in turn, employ to show that  $\mathbf{CGPTPP}^{(n)}$ coincides exactly with the set of causally ordered $n$-slot supermaps. In analogy with the one-slot case, causal order and thermality correspond to two commuting projector constraints, whose intersection completely characterizes the set of interest:
\begin{lemma}[Necessary and sufficient conditions for causally ordered thermal supermaps]
\label{lem::supermap_k}
A supermap $S^{(n)} \in \Lcal(\Hcal_0 \otimes \Hcal_1 \otimes \cdots \otimes \Hcal_{2n+1})$ is causally ordered and thermal with respect to the ordering $0\prec 1 \prec \cdots \prec 2n+1$ iff it satisfies
\begin{align}
    &\widetilde{S}^{(n)} \geq 0, \quad \tr[\widetilde{S}^{(n)}] = 1, \\
    &\widetilde{S}^{(n)} = (\Pcal^{(n)}_{\textup{CO}} \circ \Pcal^{(n)}_{\textup{TH}})[\widetilde{S}^{(n)}] = (\Pcal^{(n)}_{\textup{TH}}\circ \Pcal^{(n)}_{\textup{CO}})[\widetilde{S}^{(n)}], 
\end{align}
where 
\begin{equation}
    \widetilde{S}^{(n)} := \sqrt{\tau_{0::2n}} S^{(n)} \sqrt{\tau_{0::2n}},
\end{equation}
with $\sqrt{\tau_{0::2n}} = \sqrt{\tau_0} \otimes \sqrt{\tau_2} \otimes \cdots \otimes \sqrt{\tau_{2n}}$, and the action of the two projectors $\Pcal^{(n)}_{\textup{CO}}$ and $\Pcal^{(n)}_{\textup{TH}}$ is given by
\begin{align}
\notag
    \Pcal^{(n)}_{\textup{CO}}[\widetilde{S}^{(n)}] &= \widetilde{S}^{(n)} - {}_{(2n+1)_\alpha}\widetilde{S}^{(n)} + {}_{(2n)_\alpha(2n+1)_\alpha}\widetilde{S}^{(n)}\\
    \notag
    &\phantom{=}- {}_{(2n)_\alpha(2n-1)_\alpha(2n+1)_\alpha }\widetilde{S}^{(n)} + \cdots -{}_{1_\alpha 2_\alpha \cdots (2n+1)_\alpha}\widetilde{S}^{(n)} \\
\label{eqn::Proj_CO}
    &\phantom{=}+ {}_{0_\alpha 1_\alpha 2_\alpha \cdots (2n+1)_\alpha}\widetilde{S}^{(n)},\\
\notag
\Pcal^{(n)}_{\textup{TH}}[\widetilde{S}^{(n)}] &= \widetilde{S}^{(n)} - {}_{0_\alpha}\widetilde{S}^{(n)} + {}_{0_\alpha 1_\alpha}\widetilde{S}^{(n)} - {}_{0_\alpha 1_\alpha 2_\alpha}\widetilde{S}^{(n)} \\
\label{eqn::Proj_TH}
&\phantom{=}+ \cdots - {}_{0_\alpha 1_\alpha \cdots (2n)_\alpha}\widetilde{S}^{(n)} 
+ {}_{0_\alpha 1_\alpha 2_\alpha \cdots (2n+1)_\alpha}\widetilde{S}^{(n)}.
\end{align}
\end{lemma}
\noindent Up to the replacement ${}_{X_\alpha}\bullet \leftrightarrow {}_{X}\bullet$, the projector $\Pcal^{(n)}_{\textup{CO}}$ above would coincide with the projector onto the set of causally ordered $n$-slot supermaps, hence the subscript CO (causally ordered). Similarly, $\Pcal^{(n)}_{\textup{TH}}$ plays the role of the projector onto the set of $n$-slot thermal supermaps, but for the Gibbs-transformed version $\widetilde{S}^{(n)}$ instead of $S^{(n)}$ (hence the subscript TH for 'thermal').

\begin{proof}
We first prove that $\Pcal^{(n)}_{\textup{CO}}$ and $\Pcal^{(n)}_{\textup{TH}}$ are indeed projectors. To this end, first assume that $\Pcal^{(k)}_{\textup{CO}}$ is a projector for some $k\in \mathbbm{N}_0$. From Eq.~\eqref{eqn::Proj_CO}, it is easy to see that 
\begin{gather}
    \Pcal^{(k+1)}_{\textup{CO}}[M]= M - {}_{(2k+3)_\alpha}M + {}_{(2k+2)_\alpha(2k+3)_\alpha}[\Pcal^{(k)}_{\textup{CO}}[M]]
\end{gather} for all $M$. Using that, by assumption, $\Pcal^{(k)}_{\textup{CO}}\circ \Pcal^{(k)}_\textup{CO} = \Pcal^{(k)}_\textup{CO}$, we can then verify $\Pcal^{(k+1)}_{\textup{CO}} \circ \Pcal^{(k+1)}_{\textup{CO}} = \Pcal^{(k+1)}_{\textup{CO}}$ by direct insertion. It is easy to see that $\Pcal^{(0)}_{\textup{CO}}$ is indeed a projector, providing the start of the induction. \\
\noindent Analogously, let $\Pcal^{(k)}_{\textup{TH}}$ be a projector for some $k\in \mathbbm{N}_0$. It is easy to see that $\Pcal^{(k+1)}_{\textup{TH}}[M] = \Pcal^{(k)}_{\textup{TH}}[M] - {}_{0_\alpha 1_\alpha \cdots (2k+1)_\alpha(2k+2)_\alpha}M + {}_{0_\alpha 1_\alpha \cdots (2k+2)_\alpha(2k+3)_\alpha}M$. Using 
\begin{align}
    {}_{0_\alpha 1_\alpha \cdots (2k+1)_\alpha(2k+2)_\alpha}[\Pcal^{(k)}_{\textup{TH}}[M]] &= {}_{0_\alpha 1_\alpha \cdots (2k+1)_\alpha(2k+2)_\alpha}M,\\
    {}_{0_\alpha 1_\alpha \cdots (2k+2)_\alpha(2k+3)_\alpha}[\Pcal^{(k)}_{\textup{TH}}[M]] &= {}_{0_\alpha 1_\alpha \cdots (2k+2)_\alpha(2k+3)_\alpha}M, \\
    \text{and} \quad \Pcal^{(k)}_{\textup{TH}} \circ \Pcal^{(k)}_{\textup{TH}} &= \Pcal^{(k)}_{\textup{TH}},
\end{align}
the projector property $\Pcal^{(k+1)}_{\textup{TH}} \circ \Pcal^{(k+1)}_{\textup{TH}} = \Pcal^{(k+1)}_{\textup{TH}}$ can then be seen by direct insertion. The start of the induction ($k=0$) follows from verifying that $\Pcal^{(0)}_{\textup{TH}}$ is a projector. Since both $\Pcal^{(k)}_{\text{CO}}$ and $\Pcal^{(k)}_{\text{TH}}$ only consist of operators of the form ${}_{X_\alpha}\bullet$ and ${}_{Y_\alpha}\bullet$, and we have ${}_{X_\alpha Y_\alpha}\bullet = {}_{Y_\alpha X_\alpha}\bullet$, we also have $\Pcal^{(k)}_{\text{CO}} \circ \Pcal^{(k)}_{\text{TH}} =\Pcal^{(k)}_{\text{TH}} \circ \Pcal^{(k)}_{\text{CO}}$.

\medskip
\noindent We can now prove the main statement. For simplicity, we first restrict the proof to the one-slot case $n=1$, such that 
\begin{gather}
\begin{split}
    &S^{(1)} \in \Lcal(\Hcal_0\otimes \Hcal_1\otimes\Hcal_2\otimes\Hcal_3)\\ 
    \text{and} \quad &\widetilde{S}^{(1)} = (\sqrt{\tau_0} \otimes \sqrt{\tau_2}) S^{(1)} (\sqrt{\tau_0} \otimes \sqrt{\tau_2}).
\end{split}
\end{gather}

\medskip
\noindent \emph{Only if.} ("$\Rightarrow$") Suppose that $S^{(1)}$ is causally ordered and thermal. By definition,
\begin{gather}
    \tr[\widetilde{S}^{(1)}] = \tr[(\tau_0 \otimes \tau_2) \star S^{(1)}].
\end{gather}
Since $S^{(1)}$ is thermal, we have $(\tau_0 \otimes \tau_2) \star S^{(1)} = \tau_1 \otimes \tau_3$ and the above thus implies $\tr[\widetilde{S}^{(1)}] = 1$, yielding the trace constraint of the Lemma.

\noindent Causal order of $S^{(1)}$ implies that it satisfies $S^{(1)} = \mathcal{P}^{(1)}_{\textup{SC}}[S^{(1)}]$ [see Eq.~\eqref{eqn::caus_order_app}]. Therefore,
\begin{gather}
    {}_{3_\alpha}\widetilde{S}^{(1)} = {}_{2_\alpha 3_\alpha}\widetilde{S}^{(1)} \quad \text{and} \quad {}_{1_\alpha 2_\alpha 3_\alpha}\widetilde{S}^{(1)} = {}_{0_\alpha 1_\alpha 2_\alpha 3_\alpha}\widetilde{S}^{(1)},
\end{gather} 
as can be verified by direct insertion. Combining the above equations, we obtain 
\begin{align}
\notag
    \widetilde{S}^{(1)} &= \widetilde{S}^{(1)} - {}_{3_\alpha}\widetilde{S}^{(1)} + {}_{2_\alpha 3_\alpha}\widetilde{S}^{(1)} - {}_{1_\alpha 2_\alpha 3_\alpha}\widetilde{S}^{(1)} + {}_{0_\alpha 1_\alpha 2_\alpha 3_\alpha}\widetilde{S}^{(1)}\\
    &= \Pcal_{\textup{CO}}^{(1)}[\widetilde{S}^{(1)}]
\end{align}
In a similar vein, from the thermality constraint $\tau_0 \star S^{(1)} = \tau_1 \otimes S^{(0)}$, we obtain ${}_{0_\alpha}\widetilde{S}^{(1)} = {}_{0_\alpha 1_\alpha} \widetilde{S}^{(1)}$. Combined with $\tau_2 \star S^{(0)} = \tau_3$, this yields ${}_{0_\alpha1_\alpha 2_\alpha}\widetilde{S}^{(1)} = {}_{0_\alpha 1_\alpha 2_\alpha 3_\alpha} \widetilde{S}^{(1)}$. In combination we thus obtain
\begin{align}
\notag
    \widetilde{S}^{(1)}&= \widetilde{S}^{(1)} - {}_{0_\alpha}\widetilde{S}^{(1)} + {}_{0_\alpha 1_\alpha}\widetilde{S}^{(1)} - {}_{0_\alpha 1_\alpha 2_\alpha}\widetilde{S}^{(1)} + {}_{0_\alpha 1_\alpha 2_\alpha 3_\alpha}\widetilde{S}^{(1)}\\
    &=\Pcal_{\text{TH}}^{(1)}[\widetilde{S}^{(1)}].
\end{align}
Consequently, causally ordered and thermal supermaps $S^{(1)}$ satisfy $\tr[\widetilde{S}^{(1)}] = 1$ and $\widetilde{S}^{(1)} = (\Pcal^{(1)}_{\textup{CO}} \circ \Pcal^{(1)}_{\textup{TH}})[\widetilde{S}^{(1)}] = (\Pcal^{(1)}_{\textup{TH}} \circ \Pcal^{(1)}_{\textup{CO}})[\widetilde{S}^{(1)}] $. The property $\widetilde{S}^{(1)} \geq 0$ follows from $S^{(1)} \geq 0$, which holds for causally ordered $S^{(1)}$. In exactly the same manner we can show that, for the case of $n\neq 1$, causal ordering plus thermality implies $\widetilde{S}^{(n)} \geq 0$, $\tr[\widetilde{S}^{(n)}] = 1$, as well as $\widetilde{S}^{(n)} = (\Pcal^{(n)}_{\textup{CO}} \circ \Pcal^{(n)}_{\textup{TH}})[\widetilde{S}^{(n)}]$.  

\medskip
\noindent \emph{If.} ("$\Leftarrow$") Conversely, for the "if" direction (with $n=1$), assume that $\tr[\widetilde{S}^{(1)}] = 1$ and 
$\widetilde{S}^{(1)} = (\Pcal^{(1)}_{\textup{CO}} \circ \Pcal^{(1)}_{\textup{TH}})[\widetilde{S}^{(1)}] = (\Pcal^{(1)}_{\textup{TH}}\circ \Pcal^{(1)}_{\textup{CO}})[\widetilde{S}^{(1)}]$. Since the respective projectors commute, the latter implies that $\widetilde{S}^{(1)} = \Pcal^{(1)}_{\textup{TH}}[\widetilde{S}^{(1)}]$ and $\widetilde{S}^{(1)} = \Pcal^{(1)}_{\textup{CO}}[\widetilde{S}^{(1)}]$ hold individually. From $\widetilde{S}^{(1)} = \Pcal^{(1)}_{\textup{CO}}[\widetilde{S}^{(1)}]$, we obtain 
\begin{gather} 
\label{eqn::proj_CO_aux}
{}_{3_\alpha} \widetilde{S}^{(1)} = {}_{2_\alpha 3_\alpha} \widetilde{S}^{(1)} \quad \text{and} \quad {}_{1_\alpha 2_\alpha 3_\alpha} \widetilde{S}^{(1)} = {}_{0_\alpha 1_\alpha 2_\alpha 3_\alpha} \widetilde{S}^{(1)}.
\end{gather}
Using the definition of $\widetilde{S}^{(1)}$, the first of these equations implies $\tr_{3}[{S}^{(1)}] = \tr_{23}[\tau_2 S^{(1)}] \otimes \ident_2 =: S_{01}^{\prime(0)} \otimes \ident_2$, such that 
\begin{gather}
\label{eqn::caus_eq_app1}
    {}_3S^{(1)} = {}_{23}S^{(1)} 
\end{gather}
holds. The second equation in~\eqref{eqn::proj_CO_aux} implies 
\begin{gather}
    \tr_{123}[\tau_2 S^{(1)}] = \tr[(\tau_0 \otimes \tau_2)S^{(1)}]\ident_0,
\end{gather}
which yields 
\begin{align}
\notag
     \tr_1[S_{01}^{\prime(0)}] &= \tr_{123}[\tau_2 S^{(1)}] = \tr[(\tau_0 \otimes \tau_2)S^{(1)}]\ident_0 \\
     &= \tr[\widetilde{S}^{(1)}]\ident_0 = \ident_0.
\end{align}
Consequently, we have $\tr_{123}[S^{(1)}] = d_2\ident_0$ and hence 
\begin{gather}
\label{eqn::caus_eq_app2}
    {}_{123}S^{(1)} = {}_{0123}S^{(1)}
\end{gather}
and $\tr[S^{(1)}] = d_0d_2$. Combining Eqs.~\eqref{eqn::caus_eq_app1} and~\eqref{eqn::caus_eq_app2} then shows that $S^{(1)} = \mathcal{P}^{(1)}_{\rm{SC}}[S^{(1)}]$ holds. In addition, by assumption $\widetilde{S}^{(1)} \geq 0$, such that $S^{(1)} \geq 0$. Consequently, $S^{(1)}$ is a causally ordered supermap.\\

\noindent With respect to the thermal projector, we see that $\widetilde{S}^{(1)} = \Pcal_{\textup{TH}}^{(1)}[\widetilde{S}^{(1)}]$ implies 
\begin{gather}
    {}_{0_\alpha 1_\alpha 2_\alpha}\widetilde{S}^{(1)} = {}_{0_\alpha 1_\alpha 2_\alpha 3_\alpha}\widetilde{S}^{(1)} \quad \text{and} \quad  {}_{0_\alpha}\widetilde{S}^{(1)} = {}_{0_\alpha 1_\alpha}\widetilde{S}^{(1)}.
\end{gather}
The latter yields 
\begin{gather}
\label{eqn::thermal1}
   \tau_0 \star S^{(1)} = \tau_1 \otimes \tr_{01}[\tau_0 S^{(1)}] =: \tau_1 \otimes S^{(0)}_{23}. 
\end{gather}
In addition, it is easy to see that $\tau_2 \star S_{23}^{(0)} = \tr_{012}[\widetilde{S}^{(1)}]$. From ${}_{0_\alpha 1_\alpha 2_\alpha}\widetilde{S}^{(1)} = {}_{0_\alpha 1_\alpha 2_\alpha 3_\alpha}\widetilde{S}^{(1)}$ together with $\tr[\widetilde{S}^{(1)}] = 1$ we deduce that $\tr_{012}[\widetilde{S}^{(1)}] = \tau_3$, and thus
\begin{gather}
\label{eqn::thermal2}
    \tau_2 \star S^{(0)}_{23} = \tau_3. 
\end{gather}
From Eqs.~\eqref{eqn::thermal1} and~\eqref{eqn::thermal2}, we can conclude that $S^{(1)}$ is a thermal $1$-slot supermap according to Def.~\ref{def::thermal_n_slot}. The argument for general $n$ follows along the same lines: each causality constraint translates into one of the equalities appearing in $\Pcal_{\rm CO}^{(n)}$, while each thermality constraint translates into one equality appearing in $\Pcal_{\rm TH}^{(n)}$, yielding the stated result for $n\neq 1$.
\end{proof}

\subsection{Completely GPTP-preserving \texorpdfstring{$n$}{n}-slot supermaps}
\label{app::Comp_n_slot}
Recall that, by definition, the set $\mathsf{GPTP}$ of Gibbs- and trace-preserving maps exactly coincides with $\mathbf{CGPTPP}^{(0)}$. Similarly, as a consequence of Thm.~\ref{theo:ClassesCollapse} and Prop.~\ref{prop:app-GPP-link}, we know that $\mathbf{CGPTPP}^{(1)}$ exactly coincides with the set causally ordered and thermal $1$-slot supermaps. Consequently, for $n=0,1$, Lem.~\ref{lem::supermap_k} provides an immediate characterization of $\mathbf{CGPTPP}^{(0)}$ and $\mathbf{CGPTPP}^{(1)}$ in terms of a trace constraint and a projector constraint. Importantly, this characterization is amenable to generalization to the multi-slot case (via Lem.~\ref{lem::FullyGenProj}) and can now be used to prove that \textit{all} completely GPTP-preserving multi-slot maps $S^{(n)}$ are causally ordered and thermal. In particular, we have the rigorous version of Prop.~\ref{prop::multi-slot} from the main text:
\renewcommand{\theprop}{\mainpropnum$'$}
\begin{prop}
\label{prop::CGPTPP_n_slot}
    A supermap $S^{(n)} \in \Lcal(\Hcal_0 \otimes \Hcal_1 \otimes \cdots \otimes \Hcal_{2n+1})$ satisfies $S^{(n)}  \in \mathbf{CGPTPP}^{(n)}(0,1,\dots, 2n,2n+1)$ iff it is causally ordered and thermal (with respect to the ordering $0\prec 1 \prec 2 \prec \cdots \prec 2n\prec 2n+1$). 
\end{prop}
\makeatletter
\renewcommand{\theprop}{\arabic{prop}}
\makeatother

\begin{proof}
The proof follows by induction. To this end, let us assume that for some $k \in \mathbbm{N}$ all $S^{(k-1)} \in \mathbf{CGPTPP}^{(k-1)}(1\otimes a, 2, \dots, 2k-2, 2k \otimes b)$ are causally ordered and thermal. Note that, here, the $(k-1)$-slot supermap $S^{(k-1)}$ starts on the space $1\otimes a$ and ends on the space $2k \otimes b$. We now show that this implies that $R^{(k)} \in \mathbf{CGPTPP}^{(k)}(0,1,\dots,2k, 2k+1)$ iff it is causally ordered and thermal. For notational convenience and better distinction of the involved objects, we denote the $k$-slot supermaps by $R^{(k)}$ and the resulting maps by $T^{(0)}$ in what follows. 

\medskip
\noindent \emph{If.} ("$\Leftarrow$") For the "if" part, assume $R^{(k)} \in \Lcal(\Hcal_0 \otimes \Hcal_1 \otimes \cdots \otimes \Hcal_{2k} \otimes \Hcal_{2k+1})$ is causally ordered and thermal. Setting 
\begin{gather}
   T^{(0)}:= R^{(k)}\star S^{(k-1)}, 
\end{gather} we first note that $T^{(0)} \geq 0$, since $R^{(k)} \geq 0$ and $S^{(k-1)} \geq 0$. In addition, $T^{(0)}$ is TP, since it results from contracting two proper quantum combs with compatible causal ordering, yielding a quantum channel, i.e., $T^{(0)}\in \mathsf{CPTP}(0\otimes a, 2k+1 \otimes b)$ (see Fig.~\ref{fig::multi_app} for a depiction for $k=n$). Furthermore, $T^{(0)}$ is also Gibbs-preserving, since 
\begin{align}
    (\tau_0 \otimes \tau_a) \star T^{(0)} &= (\tau_0 \star R^{(k)}) \star (\tau_a \star S^{(k-1)}) \\
    &= (\tau_1 \otimes R^{(k-1)}) \star (\tau_a \star S^{(k-1)}) \\
    &= R^{(k-1)} \star [(\tau_1 \otimes \tau_a) \star S^{(k-1)}] \\
    &=  R^{(k-1)} \star (\tau_2 \otimes S^{(k-2)}) \\\
    &=  (\tau_2 \star R^{(k-1)}) \star S^{(k-2)} \\
    &= (\tau_3 \otimes R^{(k-2)}) \star  S^{(k-2)} \\
    &= \dots = R^{(0)} \star (\tau_{2k-1} \star S^{(0)}) \\
    &= (\tau_{2k} \star R^{(0)}) \otimes \tau_b = \tau_{2k+1} \otimes \tau_b,
\end{align}
where we have alternately used the fact that both $R^{(k)}$ and $S^{(k-1)}$ are thermal. Consequently, if $R^{(k)}$ is causal and thermal, then $R^{(k)} \in \mathbf{CGPTPP}^{(k)}(0,1,\dots, 2k,2k+1)$.

\medskip
\noindent \emph{Only if.} ("$\Rightarrow$") For the "only if" part, assume that $R^{(k)}\in \mathbf{CGPTPP}^{(k)}(0,1,\dots, 2k,2k+1)$, i.e., it satisfies $R^{(k)} \geq 0$ and 
\begin{gather}
\label{eqn::proof_only_if_multi1}
    R^{(k)} \star S^{(k-1)} \in \mathsf{GPTP}(0\otimes a, 2k+1 \otimes b) 
\end{gather}
for all $S^{(k-1)} \in \mathbf{CGPTPP}^{(k-1)}(1\otimes a,\dots, 2k-1,2k\otimes b)$. 
We now first use Lem.~\ref{lem::FullyGenProj} to derive the constraints satisfied by the Gibbs-transformed operator $\breve R^{\#(k)}$. These constraints decompose into three independent parts, which we show imply, respectively, (i) thermality, (ii) causal order, and (iii) the remaining normalization conditions. \\
\noindent Following the analogous proof for the properties of CGPTPP maps (see Sec.~\ref{app:completenessCausal}), we set 
\begin{gather}
    R^{\#(k)} =: R^{(k)} \otimes \kketbra{\ident}{\ident}_{a'a} \otimes \kketbra{\ident}{\ident}_{bb'},
\end{gather}
with $\mathcal{H}_{x'} \cong \mathcal{H}_{x}$ for $x\in \{a,b\}$ and $\kketbra{\ident}{\ident}_{a'a}$ and $\kketbra{\ident}{\ident}_{bb'}$ are Choi operators of identity channels $\mathcal{I}_{a' \rightarrow a}$ and $\mathcal{I}_{b\rightarrow b'}$, respectively. $R^{\#(k)}$ acts on \textit{all} degrees of freedom of $S^{(k-1)}$ and Eq.~\eqref{eqn::proof_only_if_multi1} equivalently reads 
\begin{gather}
T^{(0)} := R^{\#(k)} \star S^{(k-1)} \in \mathsf{GPTP}(0\otimes a', 2k+1 \otimes b') 
\end{gather}
for all $S^{(k-1)} \in \mathbf{CGPTPP}^{(k-1)}(1\otimes a,\dots, 2k-1,2k\otimes b)$. In what follows, we will use Lem.~\ref{lem::FullyGenProj} to first find the properties of $R^{\#(k)}$ and subsequently those of $R^{(k)}$. 

\noindent Specifically, we first express the input and output spaces in the Gibbs-transformed representations, allowing us to apply Lem.~\ref{lem::FullyGenProj}. To this end, following Lem.~\ref{lem::supermap_k} we set 
\begin{gather}
    \widetilde{S}^{(k-1)}:= (\sqrt{\tau_a} \otimes \sqrt{\tau_{1::2k-1}}) S^{(k-1)} (\sqrt{\tau_a} \otimes \sqrt{\tau_{1::2k-1}}), 
\end{gather}
where $\sqrt{\tau_{1::2k-1}} = \sqrt{\tau_1} \otimes \sqrt{\tau_3} \otimes \cdots \otimes \sqrt{\tau_{2k-1}}$, and 
\begin{gather}
\widetilde{T}^{(0)} := (\sqrt{\tau_{a'}} \otimes \sqrt{\tau_0}) T^{(0)} (\sqrt{\tau_{a'}} \otimes \sqrt{\tau_0}).
\end{gather}
 We recall that, by induction assumption, all $S^{(k-1)} \in \mathbf{CGPTPP}^{(k-1)}$ are causally ordered and thermal. From Lem.~\ref{lem::supermap_k}, we know that $S^{(k-1)}$ is causal and thermal, if and only if $\widetilde{S}^{(k-1)}\geq 0$ and 
\begin{gather}
\label{eqn::input_space_proof_n_slot}
\begin{split}
    \tr[\widetilde{S}^{(k-1)}] &= 1 =:\gamma_{\rm i}, \\
    \text{and} \quad  \widetilde{S}^{(k-1)} &= (\Pcal_{\textup{TH}}^{(k-1)\#} \circ \Pcal_{\textup{CO}}^{(k-1)\#})[\widetilde{S}^{(k-1)}] \\
    &=:\Pcal^{\texttt{i}}[\widetilde{S}^{(k-1)}],
    \end{split}
\end{gather}
where $\Pcal_{\textup{TH}}^{(k-1)\#}$ and $\Pcal_{\textup{CO}}^{(k-1)\#}$ are the projectors defined in Lem.~\ref{lem::supermap_k}, adjusted for the spaces that $\widetilde{S}^{(k-1)}$ is defined on, and we have added a "$\#$ superscript do emphasize that these projectors also act on the auxiliary spaces $a$ and $b$. Concretely, we have, 
\begin{gather}
\label{eqn::Proj_CO_hash}
\begin{split}
 &\Pcal^{(k-1)\#}_{\textup{CO}}[\widetilde{S}^{(k-1)}] \\
    &= \widetilde{S}^{(k-1)} - {}_{(2k)_\alpha b_\alpha}\widetilde{S}^{(k-1)} + {}_{(2k-1)_\alpha(2k)_\alpha b_\alpha}\widetilde{S}^{(k-1)} \\
   &\phantom{=,}- \cdots -{}_{2_\alpha 3_\alpha \cdots (2k)_\alpha b_\alpha}\widetilde{S}^{(k-1)} + {}_{a_\alpha 1_\alpha 2_\alpha \cdots (2k)_\alpha b_\alpha}\widetilde{S}^{(k-1)},
\end{split}
\end{gather}
as well as 
\begin{gather}
\label{eqn::Proj_TH_hash}
\begin{split}
&\Pcal^{(k-1)\#}_{\textup{TH}}[\widetilde{S}^{(k-1)}] \\
&= \widetilde{S}^{(k-1)} - {}_{a_\alpha 1_\alpha}\widetilde{S}^{(k-1)} + {}_{a_\alpha 1_\alpha 2_\alpha}\widetilde{S}^{(k-1)} \\
&\phantom{=,}- \cdots - {}_{a_\alpha 1_\alpha 2_\alpha \cdots (2k-1)_\alpha} \widetilde{S}^{(k-1)} + {}_{a_\alpha 1_\alpha 2_\alpha \cdots (2k)_\alpha b_\alpha}\widetilde{S}^{(k-1)}.
\end{split}
\end{gather}
In what follows, we denote the set of matrices $S^{(k-1)}$ satisfying the conditions of~\eqref{eqn::input_space_proof_n_slot} by $\widetilde{\mathsf{A}}_{\rm i}$. In addition, for consise notation, we will write the corresponding projectors simply as $\Pcal^{\#}_{\textup{CO}}$ and $\Pcal^{\#}_{\textup{TH}}$, respectively.

\noindent Analogously, from Prop.~\ref{prop::GPTP} we know that $T^{(0)} \in \mathsf{GPTP}(0\otimes a', 2k+1 \otimes b'$ iff $\widetilde{T}^{(0)} \geq 0$ and 
\begin{gather}
\label{eqn::GPTP_proof}
\begin{split}
 \tr[\widetilde{T}^{(0)}] &=1 =: \gamma_\texttt{o},\\
   \widetilde{T}^{(0)} &= \widetilde{T}^{(0)} - {}_{a_\alpha' 0_\alpha}\widetilde{T}^{(0)}  - {}_{(2k+1)_\alpha b_\alpha'}\widetilde{T}^{(0)} \\
   &\phantom{=,}+ 2 {}_{a_\alpha' 0_\alpha (2k+1)_\alpha b_\alpha'}\widetilde{T}^{(0)} =: \Pcal^{\rm o}[\widetilde{T}^{(0)}].
\end{split}
\end{gather}
we denote the set of matrices $\widetilde{T}^{(0)}$ satisfying the conditions of~\eqref{eqn::GPTP_proof} by $\widetilde{\mathsf{A}}_{\rm o}$.\\ 
\noindent From Lem.~\ref{lem::Stilde_S_gen}, we know that $R^{\#(k)} \star S^{(k-1)} = T^{(0)}$ iff $\breve{R}^{\#(k)} \star \widetilde{S}^{(k-1)} = \widetilde{T}^{(0)}$, with 
\begin{gather}
\breve{R}^{\#(k)} := (\sqrt{\xi_{a'0a}} \otimes \sqrt{\tau_{1::2k-1}}^{-1}) R^{\#(k)} (\sqrt{\xi_{a'0a}} \otimes \sqrt{\tau_{1::2k-1}}^{-1}),
\end{gather}
where $\sqrt{\xi_{a'0a}} = \sqrt{\tau_a'} \otimes \sqrt{\tau_0}\otimes \sqrt{\tau_a}^{-1}$. Note that, since $(\sqrt{\tau_a} \otimes \sqrt{\tau_{a'}}^{-1})\kket{\ident_{aa'}} = \kket{\ident}_{aa'}$, we have 
\begin{gather}
\begin{split}
    \breve{R}^{\#(k)} &= (\sqrt{\tau_0} \otimes \sqrt{\tau_{1::2k-1}}^{-1}) R^{(k)} (\sqrt{\tau_0} \otimes \sqrt{\tau_{1::2k-1}}^{-1}) \\
    &\phantom{=(}\otimes \kketbra{\ident}{\ident}_{aa'} \otimes \kketbra{\ident}{\ident}_{bb'} \\
    &=:  \breve{R}^{(k)}\otimes \kketbra{\ident}{\ident}_{aa'} \otimes \kketbra{\ident}{\ident}_{bb'}
\end{split}
\end{gather}Consequently, we can now first characterize the matrices $\breve{R}^{\#(k)}$ that map $\widetilde{\mathsf{A}}_{\rm i}$ to $ \widetilde{\mathsf{A}}_{\rm o}$, and subsequently find the properties of $\breve{R}^{(k)}$ and $R^{(k)}$, respectively. \\
From Lem.~\ref{lem::FullyGenProj}, we know that $\breve{R}^{\#(k)}$ maps $\widetilde{\mathsf{A}}_{\rm i}$ into $ \widetilde{\mathsf{A}}_{\rm o}$ iff
\begin{align}
\label{eqn::Prop_kcombs1}
    &\breve{R}^{\#(k)} = \breve{R}^{\#(k)} - (\Pcal^{\rm{i}})^{\mathrm{T}}[\breve{R}^{\#(k)}] + [(\Pcal^{\rm{i}})^{\mathrm{T}} \otimes \Pcal^{\rm{o}}][\breve{R}^{\#(k)}] \\
\label{eqn::Prop_kcombs2}
&\text{and} \quad (\Pcal^{\rm{i}})^{\mathrm{T}}[\ident_{\rm{i}}] = (\Pcal_{\rm{i}})^{\mathrm{T}}[\tr_{0a'(2k+1)b'}[\breve{R}^{\#(k)}]],
\end{align}
where $\ident_{\rm{i}} = \ident_{a12\dots (2k) b}$. Using ${}_{X_\alpha^{\mathrm{T}}} \bullet = {}_{X_\beta}\bullet$ (see Lem.~\ref{lem::adjoint}), we obtain 
\begin{gather}
    (\Pcal^{\rm{i}})^{\mathrm{T}} = (\Pcal_{\textup{TH}}^{\#})^\mathrm{T} \circ (\Pcal_{\textup{CO}}^{\#})^\mathrm{T}
\end{gather} 
 via the replacement $\alpha \mapsto \beta$ in the two projectors $\Pcal_{\textup{TH}}^{(k-1)}$ and $\Pcal_{\textup{CO}}^{(k-1)}$ given in Eqs.~\eqref{eqn::Proj_CO_hash} and~\eqref{eqn::Proj_TH_hash}. We denote the resulting projectors by $\Pcal_{\textup{TH}}^{\beta\#}$ and $\Pcal_{\textup{CO}}^{\beta\#}$, respectively. Insertion of the explicit forms of the projectors and $\breve{R}^{\#(k)}$ into Eq.~\eqref{eqn::Prop_kcombs1} yields
\begin{align}
\notag
    0 &= (\Pcal_{\textup{TH}}^{\beta\#} \circ \Pcal_{\textup{CO}}^{\beta\#}) [-{}_{a'_\alpha 0_\alpha}\breve{R}^{\#(k)} - {}_{(2k+1)_\alpha b'_\alpha}\breve{R}^{\#(k)} \\
    &\phantom{=  \Pcal_{\textup{TH}}^{\beta\#} \circ \Pcal_{\textup{CO}}^{\beta\#} a.} + 2{}_{a'_\alpha 0_\alpha (2k+1)_\alpha b_\alpha'}\breve{R}^{\#(k)}]\\
    \notag
    &=(\Pcal_{\textup{TH}}^{\beta\#} \circ \Pcal_{\textup{CO}}^{\beta\#}) [-{}_{0_\alpha}\breve{R}^{(k)} \otimes \ident_a 
    \otimes \tau_{a'} \otimes \kketbra{\ident}{\ident}_{bb'} \\
    \notag
    &\phantom{=aaaaaaaaa}- {}_{(2k+1)_\alpha}\breve{R}^{(k)} \otimes \kketbra{\ident}{\ident}_{aa'} \otimes \ident_b \otimes \tau_{b'} \\
    \label{eqn::fund_eq1}
    &\phantom{=aaaaaaaaa}+ 2{}_{0_\alpha (2k+1)_\alpha}\breve{R}^{(k)} \otimes \ident_a \otimes \tau_{a'} \otimes \ident_b \otimes \tau_{b'}]
\end{align}
In order to represent the action of the above projectors on the degrees of freedom of $\breve{R}^{(k)}$ alone (i.e., without the action on the auxiliary degrees of freedom $\{a,a',b,b'\}$), we note that 
\begin{gather}
\begin{split}
    \Pcal_{\textup{TH}}^{\beta\#}[M] &= {}_{a_\beta} \Pcal_{\textup{TH}}^{\beta(k-1)}[M] + M + {}_{a_\beta 1_\beta \cdots(2k)_\beta b_\beta}M \\
    &\phantom{=}- {}_{a_\beta}M - {}_{a_\beta 1_\beta \cdots (2k)_\beta}M 
\end{split}
\end{gather}
holds for all $M$, where $ \Pcal_{\textup{TH}}^{\beta(k-1)}$ is the special case of  $\Pcal_{\textup{TH}}^{\beta(k-1)\#}$ in Eq.~\eqref{eqn::Proj_TH_hash} for trivial auxiliary degrees of freedom (i.e., without $a_\beta$ and $b_\beta$). Analogously, we have 
\begin{gather}
\begin{split}
    \Pcal_{\textup{CO}}^{\beta\#}[M] &= {}_{b_\beta} \Pcal_{\textup{CO}}^{\beta(k-1)}[M] + M + {}_{a_\beta 1_\beta \cdots(2k)_\beta b_\beta}M \\
    &\phantom{=}- {}_{b_\beta}M - {}_{1_\beta \cdots (2k)_\beta b_\beta}M.
\end{split}
\end{gather}
Using ${}_{1_\beta \dots (2k)_\beta}P_{\textup{X}}^{\beta(k-1)}[M] = {}_{1_\beta \dots (2k)_\beta}M$ for $\textup{X} \in \{\textup{TH}, \textup{CO}\}$, their concatenation yields after short calculation
\begin{align}
\notag
&(\Pcal_{\textup{TH}}^{\beta\#}\circ \Pcal_{\textup{CO}}^{\beta\#})[M] \\
\notag
&={}_{a_\beta b_\beta}(\Pcal_{\textup{TH}}^{\beta(k-1)} \circ \Pcal_{\textup{CO}}^{\beta(k-1)})[M] + {}_{a_\beta}\Pcal_{\textup{TH}}^{\beta(k-1)}[M] \\
\notag 
&\phantom{=} - {}_{a_\beta b_\beta}\Pcal_{\textup{TH}}^{\beta(k-1)}[M]  + {}_{b_\beta}\Pcal_{\textup{CO}}^{\beta(k-1)}[M] - {}_{a_\beta b_\beta}\Pcal_{\textup{CO}}^{\beta(k-1)}[M] \\
\notag 
&\phantom{=}+ M + 2{}_{a_\beta 1_\beta \cdots (2k)_\beta b_\beta}M - {}_{b_\beta}M - {}_{a_\beta}M + {}_{a_\beta b_\beta}M \\
\label{eqn::ConcTH_CO} &\phantom{=}- {}_{1_\beta\cdots(2k)_\beta b_\beta}M 
- {}_{a_\beta 1_\beta\cdots (2k)_\beta}M.
\end{align}
Insertion of this projector into Eq.~\eqref{eqn::fund_eq1} and grouping it according to the respective terms on the auxiliary space yields an equation of the form 
\begin{equation}
\label{eqn::f_terms}
\begin{split}
0 &= f_1(\breve{R}^{(k)}) \otimes \ident_a 
    \otimes \tau_{a'} \otimes \kketbra{\ident}{\ident}_{bb'} \\
    &\phantom{=}+ f_2(\breve{R}^{(k)}) \otimes \kketbra{\ident}{\ident}_{aa'} \otimes \ident_b \otimes \tau_{b'} \\
    &\phantom{=}+ f_3(\breve{R}^{(k)}) \otimes \ident_a \otimes \tau_{a'} \otimes \ident_b \otimes \tau_{b'}
\end{split}
\end{equation}
Since the three terms on the auxiliary spaces are linearly independent, we obtain the three independent conditions \[ f_1(\breve R^{(k)})=0,\qquad f_2(\breve R^{(k)})=0,\qquad f_3(\breve R^{(k)})=0. \] We now proceed to show that these imply thermality, causal order, and the remaining normalization conditions, respectively.
For $f_1$, explicit evaluation yields 
\begin{gather}
\label{eqn::f1_expl}
0 = f_1(\breve{R}^{(k)}) = -{}_{0_\alpha}\Pcal_{\textup{TH}}^{\beta(k-1)}[\breve{R}^{(k)}] + {}_{0_\alpha 1_\beta \cdots (2k)_\beta}\breve{R}^{(k)}
\end{gather}
Inserting the definition of $\Pcal_\textup{TH}^{\beta(k-1)}$ into the above equation, we obtain
\begin{gather}
\begin{split}
    0 &= {}_{0_\alpha}\breve{R}^{(k)} - {}_{0_\alpha 1_\beta}\breve{R}^{(k)} + {}_{0_\alpha 1_\beta 2_\beta}\breve{R}^{(k)} - {}_{0_\alpha 1_\beta 2_\beta 3_\beta}\breve{R}^{(k)} \\
    &\phantom{=,}+ \dots - {}_{0_\alpha 1\beta \dots (2k-1)_\beta}\breve{R}^{(k)},   
\end{split}
\end{gather}
which is equivalent to
\begin{align}
    {}_{0_\alpha}\breve{R}^{(k)} &= {}_{0_\alpha 1_\beta}\breve{R}^{(k)}, \\
    {}_{0_\alpha 1_\beta 2_\beta} \breve{R}^{(k)} &= {}_{0_\alpha 1_\beta 2_\beta 3_\beta}\breve{R}^{(k)}\\\notag
    &\vdots\\
    {}_{0_\alpha 1_\beta \dots (2k-2)_\beta}\breve{R}^{(k)} &= {}_{0_\alpha 1_\beta \dots (2k-2)_\beta (2k-1)_\beta}\breve{R}^{(k)},   
\end{align}
holding simultaneously. Writing the above explicitly in terms of $R^{(k)}$ yields
\begin{align}
    \tr_0[\tau_0 R^{(k)}] &= \tau_1 \otimes \tr_{01}[\tau_0 R^{(k)}]  , \\
    \tr_{012}[(\tau_0 \otimes \tau_2)R^{(k)}] &= \tau_3 \otimes \tr_{0123}[(\tau_0 \otimes \tau_2)R^{(k)}]\\ \notag
    &\vdots\\
    \tr_{0\dots(2k-2)}[\tau_{0::2k-2}R^{(k)}] &= \tau_{2k-1} \otimes \tr_{0\dots(2k-1)}[\tau_{0::2k-2}R^{(k)}] ,   
\end{align}
where $\tau_{0::2k-2} = \tau_0 \otimes \tau_2 \otimes \tau_4 \otimes \cdots \otimes \tau_{2k-2}$. Setting $R^{(k-j)}:= \tr_{0\dots (2j-1)}[\tau_{0::2j-2} R^{(k)}]$, we thus have shown that $R^{(k)}$ satisfies 
\begin{align}
    \tau_0 \star R^{(k)} &= \tau_1 \otimes R^{(k-1)}, \\
    \tau_2 \star R^{(k-1)} &= \tau_3 \otimes R^{(k-2)} \\ \notag
    &\vdots\\
    \tau_{2k-2} \star R^{(1)} &= \tau_{2k-1} \otimes R^{(0)},   
\end{align}
i.e., we have shown that $R^{(k)}$ satisfies all thermality constraints (see Def.~\ref{def::thermal_n_slot}) except for the final one ($\tau_{2k} \star R^{(0)} = \tau_{2k+1}$).

To show satisfaction of (all but the final) causality constraints, we now evaluate term $f_2(\bar{S}^{(k)})$ in Eq.~\eqref{eqn::f_terms}, which reads
\begin{gather}
\label{eqn::f2_expl}
    0 = f_2(\breve{R}^{(k)}) = -{}_{(2k+1)_\alpha}\Pcal_{\textup{CO}}^{\beta(k-1)}[\breve{R}^{(k)}] + {}_{1_\beta \cdots (2k)_\beta (2k+1)_\alpha}\breve{R}^{(k)}. 
\end{gather}
Using the explicit form of $\Pcal_{\textup{CO}}^{\beta(k-1)}$, we thus obtain
\begin{align}
\notag 
0 &= {}_{(2k+1)_\alpha}\breve{R}^{(k)} - {}_{(2k)_\beta (2k+1)_\alpha} \breve{R}^{(k)} +  {}_{(2k-1)_\beta (2k)_\beta (2k+1)_\alpha} \breve{R}^{(k)}  \\
&\phantom{=.}-\cdots +{}_{3_\beta 4_\beta \dots (2k)_\beta (2k+1)_\alpha}R^{(k)} - {}_{2_\beta 3_\beta \dots (2k)_\beta (2k+1)_\alpha}R^{(k)}. 
\end{align}
It easy to see that this is equivalent to 
\begin{align}
\label{eqn::causality_multi_thermal1}
    {}_{(2k+1)_\alpha}\breve{R}^{(k)} &= {}_{(2k)_\beta (2k+1)_\alpha}\breve{R}^{(k)}, \\
    \label{eqn::causality_multi_thermal2}
    {}_{(2k-1)_\beta (2k)_\beta (2k+1)_\alpha} \breve{R}^{(k)} &= {}_{(2k-2)_\beta (2k-1)_\beta (2k)_\beta (2k+1)_\alpha} \breve{R}^{(k)}\\
    \notag &\vdots\\
    \label{eqn::causality_multi_thermal3}
    {}_{3_\beta 4_\beta \dots (2k)_\beta (2k+1)_\alpha} \breve{R}^{(k)} &= {}_{2_\beta 3_\beta \dots (2k)_\beta (2k+1)_\alpha} \breve{R}^{(k)} 
\end{align}
Expressing the condition of Eq.~\eqref{eqn::causality_multi_thermal1} explicitly in terms of $R^{(k)}$, we obtain 
\begin{gather}
\begin{split}
    \tr_{(2k+1)}[R^{(k)}] &= \tr_{(2k)(2k+1)}[\tau_{2k}R^{(k)}]\otimes \ident_{2k}\\
    &=: R^{\prime (k-1)}\otimes \ident_{2k},
\end{split}
\end{gather}
which corresponds to the first causality constraint in Prop.~\ref{prop::alt_charac_caus}. From Eq.~\eqref{eqn::causality_multi_thermal2}, we see that
\begin{gather}
\label{eqn::prime_op_def}
\begin{split}
    &\tr_{(2k-1)(2k)(2k+1)}[\tau_{2k}R^{(k)}]\\
    &= \tr_{(2k-1)}[R^{\prime(k-1)}]\\
    &= \tr_{(2k-2)(2k-1)(2k)(2k+1)}[(\tau_{2k-2} \otimes \tau_{2k})R^{(k)}] \otimes \ident_{2k-2}\\
    &=: R^{\prime(k-2)} \otimes \ident_{2k-2}
\end{split}
\end{gather}
holds. That is, $R^{(k)}$ also satisfies the second causality constraint of Prop.~\ref{prop::alt_charac_caus}. We can show the remaining causality constraints, i.e., 
\begin{gather}
    \tr_{(2j+1)}[R^{\prime(j)}] = R^{\prime(j-1)} \otimes \ident_{2j} \quad \forall j\in\{1,2,\dots,k-1\},
\end{gather}
where
\begin{gather}
\label{eqn::def_red_comb}
\begin{split}
    R^{\prime(j)} &= \tr_{(2j+2)(2j+3)\dots (2k+1)}[(\tau_{2j+2} \otimes \tau_{2j+4} \otimes \cdots \otimes \tau_{2k})R^{(k)}]\\
    &= \tr_{(2j+2)(2j+3)}[\tau_{2j+2}R^{\prime(j+1)}]
\end{split}
\end{gather}
in the same manner, going down the conditions~\eqref{eqn::causality_multi_thermal1} to~\eqref{eqn::causality_multi_thermal3}, ending on $\tr_{3}[R^{\prime(1)}] = R^{\prime(0)}$. That is, we have shown all but the final causality constraint. 

\noindent To obtain the remaining causality and thermality constraint, we evaluate $f_3(\breve R^{(k)})$ in Eq.~\eqref{eqn::f_terms}. In particular, identifying the terms in Eq.~\eqref{eqn::fund_eq1} with $\ident_{a}\otimes \tau_{a'} \otimes \ident_b \otimes \tau_{b'}$ on the auxiliary space yields, after short calculation:
\begin{gather}
\begin{split}
    &0 = f_3(\breve{R}^{(k)}) \\
    &= (\Pcal_{\textup{TH}}^{\beta(k-1)} \circ \Pcal_{\textup{CO}}^{\beta(k-1)} - \Pcal_{\textup{TH}}^{\beta(k-1)} + {}_{1\beta \dots (2k)_\beta}\bullet) [-{}_{0_\alpha}\breve{R}^{(k)}]\\
    &\phantom{=}+ (\Pcal_{\textup{TH}}^{\beta(k-1)} \circ \Pcal_{\textup{CO}}^{\beta(k-1)} - \Pcal_{\textup{CO}}^{\beta(k-1)} + {}_{1\beta \dots (2k)_\beta}\bullet) [-{}_{(2k+1)_\alpha}\breve{R}^{(k)}]\\
    &\phantom{=} +2 (\Pcal_{\textup{TH}}^{\beta(k-1)} \circ \Pcal_{\textup{CO}}^{\beta(k-1)})[{}_{0_\alpha (2k+1)_\alpha}\breve{R}^{(k)}].
\end{split}
\end{gather}
Using Eqs.~\eqref{eqn::f1_expl} and ~\eqref{eqn::f2_expl}, we see that some of the terms in the above expression vanish, leaving us with 
\begin{gather}
\begin{split}
    0 = (\Pcal_{\textup{TH}}^{\beta(k-1)} \circ \Pcal_{\textup{CO}}^{\beta(k-1)})&[-{}_{0_\alpha}\breve{R}^{(k)} - {}_{(2k+1)_\alpha}\breve{R}^{(k)} \\
    &\phantom{[}+ 2{}_{0_\alpha (2k+1)_\alpha}\breve{R}^{(k)}]. 
\end{split}
\end{gather}
Since $\Pcal_{\textup{TH}}^{\beta(k-1)}$ and $\Pcal_{\textup{CO}}^{\beta(k-1)}$ commute, the above is equivalent to 
\begin{align}
\label{eqn::thermal_final1}
0 &= \Pcal_{\textup{TH}}^{\beta(k-1)}[-{}_{0_\alpha}\breve{R}^{(k)} - {}_{(2k+1)_\alpha}\breve{R}^{(k)} + 2{}_{0_\alpha (2k+1)_\alpha}\breve{R}^{(k)}],\\
\label{eqn::thermal_final12}
0 &= \Pcal_{\textup{CO}}^{\beta(k-1)}[-{}_{0_\alpha}\breve{R}^{(k)} - {}_{(2k+1)_\alpha}\breve{R}^{(k)} + 2{}_{0_\alpha (2k+1)_\alpha}\breve{R}^{(k)}].
\end{align}
Applying ${}_{0_\alpha}\bullet$ on both sides of Eq.~\eqref{eqn::thermal_final1} and using ${}_{0_\alpha}\Pcal_{\textup{TH}}^{\beta(k-1)}[\breve{R}^{(k)}] = {}_{0_\alpha 1_\beta \cdots (2k)_\beta}\breve{R}^{(k)}$ [see Eq.~\eqref{eqn::f1_expl}], we obtain
\begin{gather}
\label{eqn::thermal_final_final1}
\begin{split}
    {}_{0_\alpha 1_\beta\dots (2k)_\beta}\breve{R}^{(k)} &= {}_{(2k+1)}[{}_{0_\alpha} \Pcal_{\textup{TH}}^{\beta(k-1)}[\breve{R}^{(k)}]] \\
    &= {}_{0_\alpha 1_\beta\dots (2k)_\beta (2k+1)_\alpha}\breve{R}^{(k)}.
\end{split}
\end{gather}
Using the definition $R^{(0)} = \tr_{01\dots (2k-1)}[\tau_0 \otimes \tau_2 \otimes \cdots \otimes \tau_{2k-2} R^{(k)}]$ from above, Eq.~\eqref{eqn::thermal_final_final1} implies 
\begin{gather}
    \tr_{2k}[\tau_{2k} R^{(0)}] = \tau_{2k} \star R^{(0)} = \tr[\tau_{2k} R^{(0)}] \tau_{2k+1},
\end{gather}
i.e., the remaining condition for thermality, modulus the additional factor $\tr[\tau_{2k} R^{(0)}]$ (we show $\tr[\tau_{2k} R^{(0)}] = 1$ below). \\
\noindent In a similar vein, applying ${}_{(2k+1)_\alpha}\bullet$ to both sides of Eq.~\eqref{eqn::thermal_final12} and using ${}_{(2k+1)_\alpha}\Pcal_{\textup{CO}}^{\beta(k-1)}[\breve{R}^{(k)}] = {}_{1_\beta \cdots (2k)_\beta (2k+1)_\alpha}\breve{R}^{(k)}$ [see Eq.~\eqref{eqn::f2_expl}] yields
\begin{gather}
\label{eqn::causality_constr_final}
\begin{split}
    {}_{1_\beta \cdots (2k)_\beta (2k+1)_\alpha}\breve{R}^{(k)} &= {}_{0_\alpha}[{}_{(2k+1)_\alpha}\Pcal_{\textup{CO}}^{\beta(k-1)}[\breve{R}^{(k)}]] \\
    &={}_{0_\alpha 1_\beta \cdots (2k)_\beta (2k+1)_\alpha}\breve{R}^{(k)}.
\end{split}
\end{gather}
From Eq.~\eqref{eqn::def_red_comb}, we have 
\begin{gather}
    R^{\prime (0)} = \tr_{23\dots (2k+1)}[(\tau_2 \otimes \tau_4 \otimes \cdots \otimes \tau_{2k})R^{(k)}].
\end{gather}
Consequently, expressing $\breve{R}^{(k)}$ in explicitly in terms of $R^{(k)}$, Eq.~\eqref{eqn::causality_constr_final} tells us that 
\begin{gather}
\tr_1[R^{\prime (0)}] = \tr[\tau_0 R^{\prime (0)}]\ident_0,
\end{gather}
holds, i.e., we obtain the final causality constraint, up to the factor $ \tr[\tau_0 R^{\prime (0)}]$ (we show below that $ \tr[\tau_0 R^{\prime (0)}]= 1$).

\noindent In order to show the final normalisation constraints, 
\begin{gather}
\label{eqn::causality_thermal_final_constr}
\begin{split}
\tr[\tau_{2k}R^{(0)}] &= 1 \ \text{(thermality)} \\
\text{and} \quad \tr[\tau_0 R^{\prime (0)}] &= 1 \ \text{(causality)},
\end{split}
\end{gather}
we use the final condition~\eqref{eqn::Prop_kcombs2} that $\breve{R}^{\#(k)}$ must satisfy, namely
\begin{gather}
\label{eqn::trace_and_thermality}
     (\Pcal^{\rm{i}})^{\mathrm{T}}[\ident_{a1\dots (2k)b}] = (\Pcal_{\rm{i}})^{\mathrm{T}}[\tr_{0a'(2k+1)b'}[\breve{R}^{\#(k)}]],
\end{gather}
with $(\Pcal^{\rm{i}})^{\mathrm{T}} = \Pcal_{\textup{TH}}^{\beta\#}\circ \Pcal_{\textup{CO}}^{\beta\#}$. To this end, first, it is easy to see that we have 
\begin{gather}
     (\Pcal_{\textup{TH}}^{\beta\#}\circ \Pcal_{\textup{CO}}^{\beta\#})[\ident_{a1\dots (2k)b}] = \ident_{a1\dots (2k)b}. 
\end{gather}
For the rhs of Eq.~\eqref{eqn::trace_and_thermality}, we first note that 
\begin{gather}
    \tr_{0a'(2k+1)b'}[\breve{R}^{\#(k)}] = \tr_{0(2k+1)}[\breve{R}^{(k)}] \otimes \ident_a \otimes \ident_b 
\end{gather}
holds. As a consequence, we obtain
\begin{gather}
\begin{split}
     &(\Pcal_{\textup{TH}}^{\beta\#}\circ \Pcal_{\textup{CO}}^{\beta\#})[\tr_{0a'(2k+1)b'}[\breve{R}^{\#(k)}]] \\
     &= (\Pcal_{\textup{TH}}^{\beta(k-1)}\circ \Pcal_{\textup{CO}}^{\beta(k-1)})[\tr_{0(2k+1)}[\breve{R}^{(k-1)}]]\otimes \ident_a \otimes \ident_b.
\end{split}
\end{gather}
With this, Eq.~\eqref{eqn::trace_and_thermality} simplifies to
\begin{gather}
     \ident_{1\dots (2k)} = (\Pcal_{\textup{TH}}^{\beta(k-1)}\circ \Pcal_{\textup{CO}}^{\beta(k-1)})[\tr_{0(2k+1)}[\breve{R}^{(k-1)}]],
\end{gather}
which implies 
\begin{gather}
    \tr_{(2k+1)}[ \Pcal_{\textup{CO}}^{\beta(k-1)}[{}_{0_\alpha}\Pcal_{\textup{TH}}^{\beta(k-1)}[\breve{R}^{(k)}]] = \tau_0 \otimes \ident_{1\dots (2k)}
\end{gather}
Using ${}_{0_\alpha}\Pcal_{\textup{TH}}^{\beta(k-1)}[\breve{R}^{(k)}] = {}_{0_\alpha 1_\beta \cdots (2k)_\beta}\breve{R}^{(k)}$ [see Eq.~\eqref{eqn::f1_expl}] and $\Pcal_{\textup{CO}}^{\beta(k-1)}[{}_{0_\alpha 1_\beta \cdots (2k)_\beta}\breve{R}^{(k)}] = {}_{0_\alpha 1_\beta \cdots (2k)_\beta}\breve{R}^{(k)}$, the above yields
\begin{gather}
    \tr_{(2k+1)}[{}_{0_\alpha 1_\beta \cdots (2k)_\beta}\breve{R}^{(k)}]] = \tau_0 \otimes \ident_{1\dots (2k)},
\end{gather}
and thus
\begin{gather}
    \tr[(\tau_0 \otimes \tau_2 \otimes \cdots \otimes \tau_{2k})\breve{R}^{(k)}]\tau_0 = \tau_0,
\end{gather}
which implies 
\begin{gather}
\label{eqn::normalisation_final}
    \tr[(\tau_0 \otimes \tau_2 \otimes \cdots \otimes \tau_{2k})\breve{R}^{(k)}] = 1.
\end{gather}
Since we have 
\begin{align}
    R^{(0)} &= \tr_{01\dots (2k-1)}[(\tau_{0} \otimes \tau_2 \otimes \cdots \otimes \tau_{2k-2}) R^{(k)}],\\
    \& \quad R^{\prime (0)} &= \tr_{23\dots (2k+1)}[(\tau_{2} \otimes \tau_{4} \otimes \cdots \otimes \tau_{2k})R^{(k)}], 
\end{align}
Eq.~\eqref{eqn::normalisation_final} yields
\begin{gather}
    \tr[\tau_{2k}R^{(0)}] = 1\quad \text{and} \quad \tr[\tau_{0}R^{\prime(0)}] = 1, 
\end{gather}
which are exactly the two final thermality and causality constraints (see Eq.~\eqref{eqn::causality_thermal_final_constr}). \\
\noindent Since, throughout, we assume $R^{(k)}\geq 0$, we have thus shown that $R^{(k)} \star S^{(k-1)} \in \mathsf{GPTP}(0\otimes a, 2k+1 \otimes b)$ for all $S^{(k-1)}$ that are causally ordered and thermal (with ordering $1\otimes a \prec 2 \prec \cdots \prec 2k-1 \prec 2k \otimes b$) if and only if $R^{(k)}$ is causally ordered and thermal (with ordering $0\prec 1 \prec \cdots \prec 2k \prec 2k+1$) itself. We have already shown that $\mathbf{CGPTPP}^{(1)}(1,2,3,4)$ coincides with the set of causally ordered and thermal supermaps with ordering $1\prec 2 \prec 3 \prec 4$ (see Sec.~\ref{app:completenessCausal} with the relabeling $P \mapsto 1, I \mapsto 2, O \mapsto 3,$ and $F\mapsto 4$), providing a starting point for the induction. Setting $S^{(n)} \equiv R^{(n)}$ then yields the statement of the Theorem. 
\end{proof}

\section{Free energy of quantum channels}

\subsection{\texorpdfstring{$\mathfrak{F}$}{F} contains identity supermap and is closed under composition}
\label{app:compositionId}

    \begin{lemma}\label{lem:composition}
        Let $\mathfrak{F}\in\{\mathbf{pGPP},\mathbf{pGPTPP}\}$. Then $\mathfrak{F}$ contains the identity supermap and is closed under composition.

        \begin{proof}
            First, the identity supermap belongs to both classes. Indeed, if $\mathcal{G}\in\mathsf{GP}(I,O)$, then
            \begin{equation}
                \mathfrak{id}[\mathcal{G}]=\mathcal{G}\in\mathsf{GP}(I,O),
            \end{equation}
            so $\mathfrak{id}\in\mathbf{pGPP}$. Similarly, if $\mathcal{G}\in\mathsf{GPTP}(I,O)$, then
            \begin{equation}
                \mathfrak{id}[\mathcal{G}]=\mathcal{G}\in\mathsf{GPTP}(I,O),
            \end{equation}
            so $\mathfrak{id}\in\mathbf{pGPTPP}$.

            Now let $\mathcal{S}_2: \mathsf{GP}(I', O') \rightarrow \mathsf{GP}(P,F)$ and $\mathcal{S}_1: \mathsf{GP}(I, O) \rightarrow \mathsf{GP}(I',O')$ be $\mathbf{pGPP}$ supermaps. For every $\mathcal{G}\in\mathsf{GP}(I,O)$, we have
            \begin{equation}
                \mathcal{S}_1[\mathcal{G}]\in\mathsf{GP}(I',O'),
            \end{equation}
            and then, since $\mathcal{S}_2$ is also GP-preserving,
            \begin{equation}
                (\mathcal{S}_2\circ\mathcal{S}_1)[\mathcal{G}] = \mathcal{S}_2[\mathcal{S}_1[\mathcal{G}]] \in\mathsf{GP}(P,F).
            \end{equation}
            Moreover, the composition of proper supermaps is proper. Hence
            \begin{equation}
                \mathcal{S}_2\circ\mathcal{S}_1\in\mathbf{pGPP}.
            \end{equation}

            The same argument applies to $\mathbf{pGPTPP}$. If $\mathcal{S}_1,\mathcal{S}_2\in\mathbf{pGPTPP}$ and $\mathcal{G}\in\mathsf{GPTP}(I,O)$, then
            \begin{equation}
                \mathcal{S}_2[\mathcal{G}]\in\mathsf{GPTP}(I',O'),
            \end{equation}
            and therefore
            \begin{equation}
                (\mathcal{S}_1\circ\mathcal{S}_2)[\mathcal{G}] = \mathcal{S}_1[\mathcal{S}_2[\mathcal{G}]] \in\mathsf{GPTP}(P,F).
            \end{equation}
            Again, properness is preserved under composition. Hence
            \begin{equation}
                \mathcal{S}_1\circ\mathcal{S}_2\in\mathbf{pGPTPP}.
            \end{equation}
        \end{proof}
\end{lemma}

\subsection{Proof of Lem.\ref{lem:entirepGPP}: \texorpdfstring{$\mathtt{F}_{\rm st}^{\mathfrak{F}}$}{FstF} as Gibbs-response free energy}
\label{app:sub:lemmaResp1}

Let $\mathcal{S} \in \mathfrak{F}$ and define a trace-and-replace supermap $\mathcal{S}'$ by
\begin{equation}
    \mathcal{S}'[\mathcal{T}][\rho_P] := \operatorname{Tr}[\rho_P] \Bigl(\mathcal{S}[\mathcal{T}][\tau_{P}]\Bigr),
\end{equation}
for any $\rho_P \in \mathcal{L}(\mathcal{H}_P)$. Since $\mathcal{S}\in\mathfrak{F}$ is a proper superchannel, $\mathcal{S}[\mathcal{T}] \in \mathsf{CPTP}(P,F)$ whenever $\mathcal{T} \in \mathsf{CPTP}(I,O)$. Hence, for every $\mathcal{T}\in\mathsf{CPTP}(I,O)$, the resulting $\mathcal{S}[\mathcal{T}][\tau_P]$ is a normalized state, and therefore $\mathcal{S}'[\mathcal{T}]$ is a trace-and-replace channel. Thus $\mathcal{S}'$ is atrace-and-replace superchannel.

It remains to show that $\mathcal{S}' \in \mathfrak{F}$. If $\mathfrak{F} = \mathbf{pGPP}$, then for every $\mathcal{G}\in\mathsf{GP}(I,O)$ we have $\mathcal{S}[\mathcal{G}]\in\mathsf{GP}(P,F)$, and hence
\begin{align}
    \nonumber \mathcal{S}'[\mathcal{G}][\rho_P] &=
    \operatorname{Tr}[\rho_P] \Bigl(\mathcal{S}[\mathcal{G}][\tau_P]\Bigr) \\
    &= \operatorname{Tr}[\rho_P]\tau_F.
\end{align}
Therefore $\mathcal{S}'[\mathcal{G}] \in \mathsf{GP}(P,F)$. Similarly, if $\mathfrak{F} = \mathbf{pGPTPP}$, then for every $\mathcal{G}\in\mathsf{GPTP}(I,O)$ we have $\mathcal{S}[\mathcal{G}]\in\mathsf{GPTP}(P,F)$, and the same computation as above gives
\begin{equation}
    \mathcal{S}'[\mathcal{G}][\rho_P] = \operatorname{Tr}[\rho_P]\tau_F.
\end{equation}
Thus $\mathcal{S}'[\mathcal{G}] \in \mathsf{GPTP}(P,F)$. Hence, in both cases, $\mathcal{S}'\in\mathfrak{F}_{\rm TR}$.
Now, let $\mathcal{T}\in\mathsf{CPTP}(I,O)$. For arbitrary $\rho_P$, we have
\begin{align}
    \nonumber \mathcal{F}\Bigl[\mathcal{S}[\mathcal{T}][\tau_P]\Bigr] &= \mathcal{F}\Bigl[\mathcal{S}'[\mathcal{T}][\rho_P]\Bigr] \\
    \nonumber &\leq \sup_{\mathcal{S}' \in \mathfrak{F}_{\rm TR}}\mathcal{F}\Bigl[(\mathcal{S}'[\mathcal{T}])[\rho_P]\Bigr] \\
    &= \mathtt{F}_{\rm st}^{\mathfrak{F}}[\mathcal{T}].
\end{align}
Since $\mathcal{S}\in\mathfrak{F}$ was arbitrary, this implies
\begin{equation}\label{eq:opt_pGPP_proof1}
    \sup_{\mathcal{S}\in\mathfrak{F}} \mathcal{F}\Bigl[\mathcal{S}[\mathcal{T}][\tau_P]\Bigr] \leq \mathtt{F}_{\rm st}^{\mathfrak{F}}[\mathcal{T}].
\end{equation}

Conversely, since $\mathfrak{F}_{\rm TR}\subset\mathfrak{F}$, taking $\rho_P = \tau_P$ in \eqref{eqn::free_energy} gives
\begin{align}\label{eq:opt_pGPP_proof2}
    \nonumber
    \mathtt{F}_{\rm st}^{\mathfrak{F}}[\mathcal{T}] &= \sup_{\mathcal{S}\in\mathfrak{F}_{\rm TR}}
    \mathcal{F}\Bigl[\mathcal{S}[\mathcal{T}][\tau_P]\Bigr] \\
    &\leq \sup_{\mathcal{S}\in\mathfrak{F}}
    \mathcal{F}\Bigl[\tilde{\mathcal{S}}[\mathcal{T}][\tau_P]\Bigr].
\end{align}
Combining \eqref{eq:opt_pGPP_proof1} and \eqref{eq:opt_pGPP_proof2} proves Eq.~\eqref{eq:opt_pGPP}.

\subsection{Proof of Thm.~\ref{theo:GrespClosed}: Monotonicity and state free energy correspondence of \texorpdfstring{$\mathtt{F}_{\rm st}^{\mathfrak{F}}$}{FstF} and its closed form for
\texorpdfstring{$\mathbf{pGPP}$}{pGPP}}
\label{app:sub:lemmaResp2}
We first prove monotonicity of $\mathtt{F}_{\rm st}^{\mathfrak{F}}$ for $\mathfrak{F}\in\{\mathbf{pGPP},\mathbf{pGPTPP}\}$. Let $\mathcal{S}_0\in\mathfrak{F}$. Using Lem.~\ref{lem:entirepGPP}, we obtain
    \begin{align}
        \nonumber \mathtt{F}_{\rm st}^{\mathfrak{F}}[\mathcal{S}_0[\mathcal{T}]] &= \sup_{\mathcal{S}\in\mathfrak{F}} \mathcal{F}\Bigl[
        (\mathcal{S}[\mathcal{S}_0[\mathcal{T}]])[\tau_P] \Bigr] \\
        \nonumber &= \sup_{\mathcal{S}\in\mathfrak{F}} \mathcal{F}\Bigl[ ((\mathcal{S}\circ\mathcal{S}_0)[\mathcal{T}])[\tau_P] \Bigr] \\
        \nonumber &\leq \sup_{\mathcal{S}'\in\mathfrak{F}} \mathcal{F}\Bigl[ (\mathcal{S}'[\mathcal{T}])[\tau_P] \Bigr] \\
        &= \mathtt{F}_{\rm st}^{\mathfrak{F}}[\mathcal{T}],
    \end{align}
    where we used Lem.~\ref{lem:composition}: $\{\mathcal{S}\circ\mathcal{S}_0:\mathcal{S}\in\mathfrak{F}\} \subset \mathfrak{F}$, and taking the supremum over this subset cannot exceed taking the supremum over all of $\mathfrak{F}$. Therefore, we conclude that $\mathtt{F}_{\rm st}^{\mathfrak{F}}$ is monotone under transformations from $\mathfrak{F}$.

    We now prove correspondence to the standard free energy of states. Let $\mathcal{T}_{\eta}$ be the trace-and-replace channel
    \begin{equation}
        \mathcal{T}_{\eta}[\rho_I] = \operatorname{Tr}[\rho_I]\eta_O.
    \end{equation}
    Since the identity supermap belongs to $\mathfrak{F}$ by Lem.~\ref{lem:composition}, we immediately obtain
    \begin{equation}
        \mathtt{F}_{\rm st}^{\mathfrak{F}}[\mathcal{T}_{\eta}] \geq \mathcal{F}[\eta_O].
    \end{equation}
    
    To prove the converse inequality, fix $\mathcal{S}\in\mathfrak{F}$ and define the induced map on states
    \begin{equation}
        \Phi_{\mathcal{S}}[\eta_O] := (\mathcal{S}[\mathcal{T}_{\eta}])[\tau_P].
    \end{equation}
    Since $\mathcal{S}$ is a superchannel, it maps CPTP maps to CPTP maps. Therefore $\Phi_{\mathcal{S}} \in \mathsf{CPTP}(O,F)$.

    Now observe that, for $\eta_O=\tau_O$, the channel $\mathcal{T}_{\tau_O}$ is the completely thermalizing channel from $I$ to $O$. Since this channel is both GP and GPTP, and since $\mathcal{S}\in\mathfrak{F}$, the output channel $\mathcal{S}[\mathcal{T}_{\tau_O}]$ belongs to $\mathsf{GP}(P,F)$ or $\mathsf{GPTP}(P,F)$, respectively. Therefore
    \begin{gather}
        \Phi_{\mathcal{S}}[\tau_O] = (\mathcal{S}[\mathcal{T}_{\tau_O}])[\tau_P] = \tau_F.
    \end{gather}
    Hence $\Phi_{\mathcal{S}} \in \mathsf{GP}(O,F)$, and, being a CPTP map, 
    \begin{equation}
        \Phi_{\mathcal{S}}\in\mathsf{GPTP}(O,F).
    \end{equation}
    In turn,
    \begin{align}
        \nonumber \mathcal{F}\Bigl[ (\mathcal{S}[\mathcal{T}_{\eta}])[\tau_P] \Bigr] &= \beta^{-1}D\Bigl(\Phi_{\mathcal{S}}[\eta_O] \,\Big\|\, \tau_F \Bigr) \\
        \nonumber &=
        \beta^{-1}D\Bigl( \Phi_{\mathcal{S}}[\eta_O] \,\Big\|\, \Phi_{\mathcal{S}}[\tau_O] \Bigr) \\
        &\leq \beta^{-1}D(\eta_O\|\tau_O) \\
        &= \mathcal{F}[\eta_O],
    \end{align}
    where we used the data processing inequality 
    \begin{equation}\label{eq:proof:StateDPI}
        D\bigl(\mathcal{T}[\rho]\|\mathcal{T}[\sigma]\bigr) \leq D(\rho\|\sigma),
    \end{equation}
    valid for any states $\rho, \sigma$ and CPTP map $\mathcal{T}$. Taking the supremum over $\mathcal{S}\in\mathfrak{F}$ gives
    \begin{equation}
        \mathtt{F}_{\rm st}^{\mathfrak{F}}[\mathcal{T}_{\eta}] \leq \mathcal{F}[\eta_O].
    \end{equation}
    Therefore,
    \begin{equation}
        \mathtt{F}_{\rm st}^{\mathfrak{F}}[\mathcal{T}_{\eta}] = \mathcal{F}[\eta_O],
    \end{equation}
    for both $\mathfrak{F} = \mathbf{pGPP}$ and $\mathfrak{F} = \mathbf{pGPTPP}$.

    Finally, we prove the closed-form expression for $\mathfrak{F} = \mathbf{pGPP}$. Since $\mathbf{pGPP} \subset \mathbf{GPP}$, for every $\mathcal{S}\in\mathbf{pGPP}$, Eq.~\eqref{eq:PropGPP} implies the existence of $\mathcal{G}_{\mathcal{S}}\in\mathsf{GPTP}(O,F)$ such that
    \begin{equation}
        (\mathcal{S}[\mathcal{T}])[\tau_P] = \mathcal{G}_{\mathcal{S}}\bigl[\mathcal{T}[\tau_I]\bigr].
    \end{equation}
    Using Lem.~\ref{lem:entirepGPP}, we obtain
    \begin{align}
        \nonumber \mathtt{F}_{\rm st}^{\mathbf{pGPP}}[\mathcal{T}] &= \sup_{\mathcal{S}\in\mathbf{pGPP}} \mathcal{F}\Bigl[ (\mathcal{S}[\mathcal{T}])[\tau_P] \Bigr] \\
        \nonumber &= \sup_{\mathcal{S}\in\mathbf{pGPP}} \beta^{-1}D\Bigl( \mathcal{G}_{\mathcal{S}}\bigl[\mathcal{T}[\tau_I]\bigr] \,\Big\|\, \tau_F \Bigr) \\
        \nonumber &= \sup_{\mathcal{S}\in\mathbf{pGPP}} \beta^{-1}D\Bigl( \mathcal{G}_{\mathcal{S}}\bigl[\mathcal{T}[\tau_I]\bigr] \,\Big\|\, \mathcal{G}_{\mathcal{S}}[\tau_O] \Bigr) \\
        &\leq \beta^{-1}D\Bigl( \mathcal{T}[\tau_I] \,\Big\|\, \tau_O \Bigr) \\
        &= \mathcal{F}\bigl[\mathcal{T}[\tau_I]\bigr],
    \end{align}
    where we used the data processing inequality \eqref{eq:proof:StateDPI}.

    The converse inequality follows because the identity supermap belongs to $\mathbf{pGPP}$ by Lem.~\ref{lem:composition}:
    \begin{equation}
        \mathtt{F}_{\rm st}^{\mathbf{pGPP}}[\mathcal{T}] \geq \mathcal{F}\bigl[\mathcal{T}[\tau_I]\bigr].
    \end{equation}
    Hence
    \begin{equation}
        \mathtt{F}_{\rm st}^{\mathbf{pGPP}}[\mathcal{T}] = \mathcal{F}\bigl[\mathcal{T}[\tau_I]\bigr].
    \end{equation}

\subsection{Monotonicity and hierarchy of \texorpdfstring{$\mathtt{F}_{\rm ch}$}{Fch} and
\texorpdfstring{$\mathtt{F}_{\rm am}$}{Fam}}
\label{app:freeEnHier}

We start by proving the monotonicity of $\mathtt{F}_{\rm ch}$ and $\mathtt{F}_{\rm am}$ under equilibrium superchannels $\mathcal{S} \in \mathfrak{F}$. In particular, we have the following Proposition, encapsulating the first part of Thm.~\ref{thm::hier_monotones} in the main text:
\begin{prop}[Monotonicity of $\mathtt{F}_{\rm ch}^{\mathfrak{F}}$ and $\mathtt{F}_{\rm am}^{\mathfrak{F}}$]
    Let $\mathfrak{F}\in\{\mathbf{pGPP},\mathbf{pGPTPP}\}$. Then, for every $\mathcal{S}_0\in\mathfrak{F}$ and every $\mathcal{T}\in\mathsf{CPTP}(I,O)$,
    \begin{align}
        \mathtt{F}_{\rm ch}^{\mathfrak{F}}[\mathcal{S}_0[\mathcal{T}]] &\leq \mathtt{F}_{\rm ch}^{\mathfrak{F}}[\mathcal{T}], \\
        \mathtt{F}_{\rm am}^{\mathfrak{F}}[\mathcal{S}_0[\mathcal{T}]] &\leq \mathtt{F}_{\rm am}^{\mathfrak{F}}[\mathcal{T}].
    \end{align}

    \begin{proof}
        We first prove monotonicity of $\mathtt{F}_{\rm ch}$. By definition,
        \begin{align}
            \nonumber \mathtt{F}_{\rm ch}^{\mathfrak{F}}[\mathcal{S}_0[\mathcal{T}]] &= \sup_{\mathcal{S}\in\mathfrak{F}} \beta^{-1}D_{\rm ch}\!\left( \mathcal{S}[\mathcal{S}_0[\mathcal{T}]] \middle\| \overline{\mathcal{G}} \right)\\
            &= \sup_{\mathcal{S}\in\mathfrak{F}} \beta^{-1}D_{\rm ch}\!\left( (\mathcal{S}\circ\mathcal{S}_0)[\mathcal{T}] \middle\| \overline{\mathcal{G}} \right).
        \end{align}
        Since $\mathfrak{F}$ is closed under composition by Lem.~\ref{lem:composition}, for every $\mathcal{S}\in\mathfrak{F}$ we have
        \begin{equation}
            \mathcal{S}\circ\mathcal{S}_0\in\mathfrak{F}.
        \end{equation}
        Therefore,
        \begin{align}
            \nonumber \mathtt{F}_{\rm ch}^{\mathfrak{F}}[\mathcal{S}_0[\mathcal{T}]] &\leq \sup_{\mathcal{S}'\in\mathfrak{F}} \beta^{-1}D_{\rm ch}\!\left( \mathcal{S}'[\mathcal{T}] \middle\| \overline{\mathcal{G}} \right)\\
            &= \mathtt{F}_{\rm ch}^{\mathfrak{F}}[\mathcal{T}].
        \end{align}

        The proof for $\mathtt{F}_{\rm am}^{\mathfrak{F}}$ is identical. Namely,
        \begin{align}
            \nonumber \mathtt{F}_{\rm am}^{\mathfrak{F}}[\mathcal{S}_0[\mathcal{T}]] &= \sup_{\mathcal{S}\in\mathfrak{F}} \beta^{-1}D_{\rm am}\!\left( \mathcal{S}[\mathcal{S}_0[\mathcal{T}]] \middle\| \overline{\mathcal{G}} \right)\\
            \nonumber &= \sup_{\mathcal{S}\in\mathfrak{F}} \beta^{-1}D_{\rm am}\!\left( (\mathcal{S}\circ\mathcal{S}_0)[\mathcal{T}] \middle\| \overline{\mathcal{G}} \right)\\
            \nonumber &\leq \sup_{\mathcal{S}'\in\mathfrak{F}} \beta^{-1}D_{\rm am}\!\left( \mathcal{S}'[\mathcal{T}] \middle\| \overline{\mathcal{G}} \right)\\
            &= \mathtt{F}_{\rm am}^{\mathfrak{F}}[\mathcal{T}].
        \end{align}        
    \end{proof}
\end{prop}
Now, we proceed with proving the hierarchy $\mathtt{F}_{\rm st}^{\mathfrak{F}}[\mathcal{T}] \leq \mathtt{F}_{\rm ch}^{\mathfrak{F}}[\mathcal{T}] \leq \mathtt{F}_{\rm am}^{\mathfrak{F}}[\mathcal{T}]$ of Eq.~\eqref{eq:freeEnHier}. For every $\mathcal{S}\in\mathfrak{F}$, choosing a trivial reference system and the input state $\rho_P = \tau_P$ in the definition \eqref{eq:chanDiv} of $D_{\rm ch}$ gives
\begin{align}
    \nonumber D_{\rm ch}(\mathcal{S}[\mathcal{T}]\|\overline{\mathcal{G}}) &\geq D\!\left( (\mathcal{S}[\mathcal{T}])[\tau_P] \middle\| \overline{\mathcal{G}}[\tau_P] \right) \\
    &= D\!\left( (\mathcal{S}[\mathcal{T}])[\tau_P] \middle\| \tau_F \right).
\end{align}
Taking the supremum over $\mathcal{S}\in\mathfrak{F}$ yields
\begin{equation}
    \mathtt{F}_{\rm st}^{\mathfrak{F}}[\mathcal{T}] \leq \mathtt{F}_{\rm ch}^{\mathfrak{F}}[\mathcal{T}].
\end{equation}

On the other hand, consider $\sigma_{RI}=\rho_{RI}$ in \eqref{eq:chanAmorDiv}. Then
\begin{align}
    \nonumber D_{\rm am}(\mathcal{M}\|\mathcal{N}) &= \sup_{\rho_{RI},\sigma_{RI}} \Bigl[ D\!\left( (\mathcal{I}_R\otimes\mathcal{M})[\rho_{RI}] \middle\| (\mathcal{I}_R\otimes\mathcal{N})[\sigma_{RI}] \right) \\
    \nonumber &\qquad - D(\rho_{RI}\|\sigma_{RI}) \Bigr] \\
    \nonumber &\geq \sup_{\rho_{RI}} D\!\left( (\mathcal{I}_R\otimes\mathcal{M})[\rho_{RI}] \middle\| (\mathcal{I}_R\otimes\mathcal{N})[\rho_{RI}] \right) \\
    &= D_{\rm ch}(\mathcal{M}\|\mathcal{N}), \label{eq:domChAm}
\end{align}
where we used $D(\rho_{RI}\|\rho_{RI}) = 0$. Applying this to $\mathcal{M} = \mathcal{S}[\mathcal{T}]$ and $\mathcal{N}=\overline{\mathcal{G}}$, we get, for every $\mathcal{S}\in\mathfrak{F}$,
\begin{equation}
    D_{\rm ch}(\mathcal{S}[\mathcal{T}]\|\overline{\mathcal{G}}) \leq D_{\rm am}(\mathcal{S}[\mathcal{T}]\|\overline{\mathcal{G}}).
\end{equation}
Taking the supremum over $\mathcal{S}\in\mathfrak{F}$ gives
\begin{equation}
    \mathtt{F}_{\rm ch}^{\mathfrak{F}}[\mathcal{T}] \leq \mathtt{F}_{\rm am}^{\mathfrak{F}}[\mathcal{T}].
\end{equation}
Combining both inequalities, we obtain Eq.~\eqref{eq:freeEnHier}.

\subsection{Convertibility and complete family of monotones}
\label{app::Convertibility}
The free-energy-like quantities introduced in Secs.~\ref{sec::Gibbs_res} and~\ref{sec::free_en_div} are monotone under equilibrium superchannels. However, monotonicity alone does not generally characterize convertibility of channels. This mirrors the situation in standard thermodynamics, where a single quantity such as the relative-entropy-based free energy is insufficient to characterize state transitions in general. Instead, complete characterizations are typically expressed in terms of families of monotones, such as generalized free energies or R\'enyi divergences.

This suggests that, in the higher-order setting, one should likewise not expect a single free energy of channels to capture the full convertibility structure. Rather, one is naturally led to consider \textit{families} of monotones derived from generalized channel divergences.

To formalize convertibility and find a complete set of monotones, we first introduce the reachable set generated by free higher-order transformations.

\begin{definition}[Reachable set]
    Let $\mathcal{T}\in\mathsf{CPTP}(I,O)$. The reachable set of $\mathcal{T}$ is
    \begin{equation}
        \mathsf{R}_{\mathfrak{F}}(\mathcal{T}) := \{\mathcal{S}[\mathcal{T}] \mid \mathcal{S}\in\mathfrak{F}\}.
    \end{equation}
\end{definition}

Thus, $\mathsf{R}_{\mathfrak{F}}(\mathcal{T})$ consists of all channels obtainable exactly from $\mathcal{T}$ via equilibrium superchannels $\mathcal{S} \in \mathfrak{F}$, while its closure contains all channels that can be approximated arbitrarily well. This naturally induces a convertibility preorder.

\begin{lemma}[Preorder on the set of channels]\label{lem:preorder}
    The relation $\succeq_{\mathfrak{F}}$ defined by
    \begin{equation}
        \mathcal{T}\succeq_{\mathfrak{F}}\mathcal{T}' \quad\Longleftrightarrow\quad \mathcal{T}'\in\overline{\mathsf{R}_{\mathfrak{F}}(\mathcal{T})}
    \end{equation}
    is a preorder on $\mathsf{CPTP}(I,O)$, where the closure is taken in any equivalent finite-dimensional topology on channels, e.g. the Choi norm or the diamond norm.
    \begin{proof}
        See App.~\ref{app:preorder}.
    \end{proof}
\end{lemma}

The preorder interpretation immediately implies an inclusion relation between reachable sets.

\begin{corollary}[Inclusion of reachable sets]\label{lem:inclusionR}
    Let $\mathcal{T}, \mathcal{T}' \in \mathsf{CPTP}(I,O)$. If $\mathcal{T} \succeq_{\mathfrak{F}} \mathcal{T}'$ then
    \begin{equation}
        \overline{\mathsf{R}_{\mathfrak{F}}(\mathcal{T}')} \subseteq \overline{\mathsf{R}_{\mathfrak{F}}(\mathcal{T})}.
    \end{equation}
    \begin{proof}
        Let $\mathcal{X}\in \overline{\mathsf{R}_{\mathfrak{F}}(\mathcal{T}')}$. Then $\mathcal{T}'\succeq_{\mathfrak{F}}\mathcal{X}$. Since $\mathcal{T}\succeq_{\mathfrak{F}}\mathcal{T}'$, transitivity of $\succeq_{\mathfrak{F}}$ implies $\mathcal{T}\succeq_{\mathfrak{F}}\mathcal{X}$. Hence $\mathcal{X}\in \overline{\mathsf{R}_{\mathfrak{F}}(\mathcal{T})}$.        
    \end{proof}
\end{corollary}

We are therefore interested in families of monotones that not only provide necessary conditions for convertibility, but fully characterize the preorder induced by free higher-order transformations.

\begin{definition}[Complete family of monotones]\label{def:completeFamily}
    Let $\mathfrak{F}$ be a set of superchannels inducing the preorder $\succeq_{\mathfrak{F}}$ on $\mathsf{CPTP}(I,O)$. A family of monotones $\{M_i\}_{i\in I}$ is called \emph{complete} if the preorder $\succeq_{\mathfrak{F}}$ is fully characterized by the ordering induced by the family, i.e.,
    \begin{equation}
        \mathcal{T}\succeq_{\mathfrak{F}}\mathcal{T}' \quad\Longleftrightarrow\quad M_i(\mathcal{T})\preceq M_i(\mathcal{T}') \qquad \forall i\in I,
    \end{equation}
    where the same relation $\preceq$ is used for all monotones in the family, being either $\leq$ or $\geq$, depending on the chosen monotonicity convention.
\end{definition}

We now show that natural free-energy-like families of monotones do not generally characterize exact channel convertibility.

We begin with monotones derived from Gibbs responses.

\begin{prop}[No-go theorem for $\mathtt{F}_{\rm st}^{\mathfrak{F}}$-like monotone families]\label{theo:nogoFst}
    
    Let $\tau$ be a full-rank Gibbs state that is not proportional to the identity (i.e., $\beta \neq 0$).

    Let $\{f_i\}_{i\in I}$ be any family of state functionals, and define a family of Gibbs-response monotones
    \begin{equation}\label{eq:familySt}
        \{\mathtt{F}_{i,\mathrm{st}}^{\mathfrak{F}}[\mathcal{T}]\}_{i \in I} := \Bigl\{ \sup_{\mathcal{S}\in\mathfrak{F}} f_i\bigl((\mathcal{S}[\mathcal{T}])[\tau_P]) \bigr)\Bigr\}_{i\in I}.
    \end{equation}
    Then the family \eqref{eq:familySt} is not complete for the convertibility preorder $\succeq_\mathfrak{F}$.

    \begin{proof}
        See App.~\ref{app:nogoFst}.
    \end{proof}
\end{prop}

We now turn to families of monotones constructed from channel divergences. Motivated by the channel free energies introduced in Def.~\ref{def:chanFreeEn}, let $\{\mathfrak{D}_i\}_{i\in I}$ be a family of \textit{state} divergences satisfying the data-processing inequality
\begin{equation}\label{eq:stateDPI}
    \mathfrak{D}_i\bigl(\mathcal{T}[\rho]\|\mathcal{T}[\sigma]\bigr) \leq \mathfrak{D}_i(\rho\|\sigma),
\end{equation}
for all states $\rho,\sigma\in\mathcal{L}(\mathcal{H}_I)$ and all channels $\mathcal{T}\in\mathsf{CPTP}(I,O)$. For every $i\in I$, we define the associated channel divergence
\begin{equation}
    D_{i,\mathrm{ch}}(\mathcal{T}\|\mathcal{T}') := \sup_{\rho_{RI}} \mathfrak{D}_i\!\left((\mathcal{I}_R\otimes\mathcal{T})[\rho_{RI}] \middle\| (\mathcal{I}_R\otimes\mathcal{T}')[\rho_{RI}] \right),
\end{equation}
and the associated amortized channel divergence
\begin{align}
    \nonumber D_{i,\mathrm{am}}(\mathcal{T}\|\mathcal{T}') &:= \sup_{\rho_{RI},\sigma_{RI}} \Bigl[ \mathfrak{D}_i\!\left((\mathcal{I}_R\otimes\mathcal{T})[\rho_{RI}] \middle\| (\mathcal{I}_R\otimes\mathcal{T}')[\sigma_{RI}] \right) \\
    &\qquad -\mathfrak{D}_i(\rho_{RI}\|\sigma_{RI}) \Bigr].
\end{align}

Using the completely thermalizing channel $\overline{\mathcal{G}}$ defined in \eqref{eq:channelTherm}, we define the corresponding families of free-energy-like monotones
\begin{align}\label{eq:familyCh}
    \mathtt{F}_{i,\mathrm{ch}}^{\mathfrak{F}}[\mathcal{T}] := \sup_{\mathcal{S}\in\mathfrak{F}} \beta^{-1}D_{i,\mathrm{ch}}( \mathcal{S}[\mathcal{T}] \| \overline{\mathcal{G}} ),\\
 \label{eq:familyAm}
\text{and}\quad \mathtt{F}_{i, \mathrm{am}}^{\mathfrak{F}}[\mathcal{T}] := \sup_{\mathcal{S}\in\mathfrak{F}} \beta^{-1}D_{i,\mathrm{am}}( \mathcal{S}[\mathcal{T}] \| \overline{\mathcal{G}} ).
\end{align}

The next theorem shows that these divergence-based families are likewise insufficient to characterize convertibility.

\begin{prop}[No-go theorem for $\mathtt{F}_{\rm ch}^{\mathfrak{F}}$- and $\mathtt{F}_{\rm am}^{\mathfrak{F}}$-like monotone families]
\label{theo:nogoFchFam}
    Let $\tau$ be a full-rank Gibbs state that is not proportional to the identity (i.e., $\beta \neq 0$). Then neither family $\{\mathtt{F}_{i,\mathrm{ch}}^{\mathfrak{F}}\}_{i\in I}$ and $\{\mathtt{F}_{i,\mathrm{am}}^{\mathfrak{F}}\}_{i\in I}$ are complete for the convertibility preorder $\succeq_\mathfrak{F}$.
    
    \begin{proof}
        See App.~\ref{app:nogoFchFam}.
    \end{proof}
\end{prop}

The no-go theorems above show that monotones defined relative to a fixed reference channel cannot characterize channel convertibility. The key idea to obtain a complete family of monotones is therefore to allow the reference channel itself to vary.

To formulate this construction, we first define general quantum channel divergences:
\begin{definition}[Channel divergence {\cite{Gour2021}}]
A real-valued function $\mathbb{D}$ on pairs $\mathcal{T}, \mathcal{T}' \in \mathsf{CPTP}(I,O)$ of quantum channels is called a \textit{channel divergence} if it satisfies the generalized data-processing inequality
\begin{equation}\label{eq:DPI}
    \mathbb{D}\bigl(\mathcal{S}[\mathcal{T}] \| \mathcal{S}[\mathcal{T}']\bigr) \leq \mathbb{D}(\mathcal{T} \| \mathcal{T}'),
\end{equation}
for every proper supermap $\mathcal{S}: \mathsf{CPTP}(I,O) \rightarrow \mathsf{CPTP}(P,F)$.
\end{definition}

Starting from this general Definition, we restrict attention to the following class of \textit{admissible} divergences.

\begin{definition}[Admissible divergence]\label{def:admDiv}
    Let $\mathbb{D}: \mathsf{CPTP}(I,O) \times \mathsf{CPTP}(I,O) \rightarrow \mathbb{R}$ be a channel divergence. We say that $\mathbb{D}$ is admissible if, for any $\mathcal{T}, \mathcal{T}' \in \mathsf{CPTP}(I,O)$, $\mathbb{D}$ satisfies
    \begin{enumerate}
        \item \textbf{Nonnegativity:} $\mathbb{D}(\mathcal{T}\|\mathcal{T}') \geq 0$,
        \item \textbf{Faithfulness:} $\mathbb{D}(\mathcal{T}\|\mathcal{T}') = 0$ iff $\mathcal{T} = \mathcal{T}'$,
        \item \textbf{Topological detectability:} whenever $\mathbb{D}(\mathcal{T}_n\|\mathcal{T}')\to 0$, one has $\mathcal{T}_n\to \mathcal{T}'$ in the topology on channels.
    \end{enumerate}
\end{definition}

The first two properties ensure that $\mathbb{D}$ behaves as a genuine distinguishability measure, while the third guarantees compatibility with the topology underlying approximate channel conversion.

The channel relative entropy $D_{\rm ch}$ and the amortized channel relative entropy $D_{\rm am}$ are admissible in this sense (see App.~\ref{app:admissibilityDchDam}).

Using an admissible divergence, we now define a monotone that measures the minimal divergence between a target channel and the set of channels reachable from a given resource channel under free processing.

\begin{definition}[Residual divergence to a target channel]
\label{def:deltaDiv}
    Let $\mathcal{T}, \mathcal{R}\in\mathsf{CPTP}(I,O)$ and let $\mathbb{D}$ be an admissible divergence. Define
    \begin{equation}
        \Delta_{\mathbb{D},\mathcal{R}}(\mathcal{T}) := \inf_{\mathcal{X} \in \overline{\mathsf{R}_{\mathfrak{F}}(\mathcal{T})}} \mathbb{D}(\mathcal{X}\|\mathcal{R}).
    \end{equation}
\end{definition}

Operationally, $\Delta_{\mathbb{D},\mathcal{R}}(\mathcal{T})$ quantifies how well the target channel $\mathcal{R}$ can be approximated using the resource channel $\mathcal{T}$ together with free higher-order processing, as measured by the divergence $\mathbb{D}$.

Unlike the free-energy-like quantities introduced previously, the reference channel $\mathcal{R}$ is now allowed to vary. The resulting family of monotones fully characterizes channel convertibility.

\begin{prop}[Completeness of the $\Delta_{\mathbb{D}, \mathcal{R}}$-family of monotones]\label{theo:completeFamily}
    Let $\mathcal{T}, \mathcal{T}' \in \mathsf{CPTP}(I,O)$ and $\mathbb{D}$ be an admissible divergence. Then $\mathcal{T} \succeq_{\mathfrak{F}} \mathcal{T}'$ if and only if, for every $\mathcal{R} \in \mathsf{CPTP}(I,O)$,
    \begin{equation}\label{eq:complMonotone}
        \Delta_{\mathbb{D}, \mathcal{R}}(\mathcal{T}) \leq \Delta_{\mathbb{D}, \mathcal{R}}(\mathcal{T}').
    \end{equation}
    \begin{proof}
        See App.~\ref{app:sub:theoCompl}.
    \end{proof}

\end{prop}

In particular, choosing $\mathbb{D} = D_{\rm ch}$ or $\mathbb{D} = D_{\rm am}$ yields two explicit complete families $\{ \Delta_{D_{\rm ch}, \mathcal{R}}\}_{\mathcal{R} \in \mathsf{CPTP}(I,O)}$ and $\{ \Delta_{D_{\rm am}, \mathcal{R}}\}_{\mathcal{R} \in \mathsf{CPTP}(I,O)}$. Moreover, for the channel relative entropy, the family can always be reduced to a countable complete subfamily. More precisely, there exists a countable set of reference channels $\{ \mathcal{R}_n \}_{n\in\mathbb{N}}$ such that $\{ \Delta_{D_{\rm ch}, \mathcal{R}_n}\}_{n \in \mathbb{N}}$ remains complete (see App.~\ref{app:countableComplete}).

\subsection{Proof of Lem.~\ref{lem:preorder}: Preorder on the set of channels}
\label{app:preorder}

As a first step to show that $\succeq_{\mathfrak{F}}$ is a preorder, we  prove reflexivity. By Lem.~\ref{lem:composition}, the identity supermap belongs to $\mathfrak{F}$. Hence
\begin{equation}
    \mathcal{T} \in \mathsf{R}_{\mathfrak{F}}(\mathcal{T}) \subseteq \overline{\mathsf{R}_{\mathfrak{F}}(\mathcal{T})}.
\end{equation}
Therefore,
\begin{equation}
    \mathcal{T} \succeq_{\mathfrak{F}} \mathcal{T}.
\end{equation}

We now prove transitivity. Suppose that
\begin{gather}
    \mathcal{T} \succeq_{\mathfrak{F}} \mathcal{T}' \quad \text{and} \quad \mathcal{T}' \succeq_{\mathfrak{F}} \mathcal{T}''.
\end{gather}
Equivalently,
\begin{gather}
    \mathcal{T}' \in \overline{\mathsf{R}_{\mathfrak{F}}(\mathcal{T})} \quad \text{and} \quad 
    \mathcal{T}'' \in \overline{\mathsf{R}_{\mathfrak{F}}(\mathcal{T}')}.
\end{gather}

Let $\varepsilon > 0$. Since $\mathcal{T}'' \in \overline{\mathsf{R}_{\mathfrak{F}}(\mathcal{T}')}$, there exists $\mathcal{S}'\in\mathfrak{F}$ such that
\begin{equation}
    \| \mathcal{S}'[\mathcal{T}'] - \mathcal{T}'' \| < \frac{\varepsilon}{2}.
\end{equation}

Since $\mathcal{S}'$ is a proper supermap, it is a linear map between finite-dimensional vector spaces of linear maps. Therefore $\mathcal{S}'$ is continuous with respect to any norm on channels, e.g., the diamond norm or the Choi norm. Hence there exists $\delta > 0$ such that, for every channel $\mathcal{M}$,
\begin{equation}
    \| \mathcal{M}-\mathcal{T}' \| < \delta \quad\Longrightarrow\quad \| \mathcal{S}'[\mathcal{M}] - \mathcal{S}'[\mathcal{T}'] \| < \frac{\varepsilon}{2}.
\end{equation}

Since $\mathcal{T}'\in\overline{\mathsf{R}_{\mathfrak{F}}(\mathcal{T})}$, there exists $\mathcal{S}\in\mathfrak{F}$ such that
\begin{equation}
    \| \mathcal{S}[\mathcal{T}] - \mathcal{T}' \| < \delta.
\end{equation}
Therefore,
\begin{equation}
    \| \mathcal{S}'[\mathcal{S}[\mathcal{T}]] - \mathcal{S}'[\mathcal{T}'] \| < \frac{\varepsilon}{2}.
\end{equation}

By Lem.~\ref{lem:composition}, the class $\mathfrak{F}$ is closed under composition. Hence
\begin{equation}
    \mathcal{S}'\circ\mathcal{S} \in \mathfrak{F},
\end{equation}
and therefore
\begin{equation}
    (\mathcal{S}'\circ\mathcal{S})[\mathcal{T}] \in \mathsf{R}_{\mathfrak{F}}(\mathcal{T}).
\end{equation}

Using the triangle inequality, we obtain
\begin{align}
    \nonumber \| (\mathcal{S}'\circ\mathcal{S})[\mathcal{T}] - \mathcal{T}'' \| &\leq \| \mathcal{S}'[\mathcal{S}[\mathcal{T}]] - \mathcal{S}'[\mathcal{T}'] \| \\
    \nonumber &\quad + \| \mathcal{S}'[\mathcal{T}'] - \mathcal{T}'' \| \\ 
    &< \varepsilon.
\end{align}

Since $\varepsilon > 0$ was arbitrary, it follows that
\begin{equation}
    \mathcal{T}'' \in \overline{\mathsf{R}_{\mathfrak{F}}(\mathcal{T})}.
\end{equation}
Therefore,
\begin{equation}
    \mathcal{T} \succeq_{\mathfrak{F}} \mathcal{T}''.
\end{equation}
Hence $\succeq_{\mathfrak{F}}$ is a preorder.

\subsection{Proof of Prop.~\ref{theo:nogoFst}: No-go theorem for \texorpdfstring{$\mathtt{F}_{\rm st}^{\mathfrak{F}}$}{FstF}-like monotone families}
\label{app:nogoFst}

    Here, we show that no family of Gibbs response free energies is complete. Before proceeding with the proof, we prove the following useful Lemmas:
    
    \begin{lemma}[Optimized $\mathtt{F}_{\rm st}^{\mathfrak{F}}$-like monotones depend only on the Gibbs response]
    \label{lem:FstDependsOnlyOnGibbsResponse}
        Let $\mathfrak{F} \in \{\mathbf{pGPP}, \mathbf{pGPTPP}\}$, and let
        \begin{equation}
            \mathtt{F}_{i,\rm st}^{\mathfrak{F}}[\mathcal{T}] := \sup_{\mathcal{S}\in\mathfrak{F}} f_i\Bigl((\mathcal{S}[\mathcal{T}])[\tau_P]\Bigr)
        \end{equation}
        be an optimized $\mathtt{F}_{\rm st}^{\mathfrak{F}}$-like monotone family generated by state functionals $\{f_i\}_{i\in I}$.

        Then, for every pair of channels $\mathcal{T},\mathcal{T}'\in\mathsf{CPTP}(I,O)$,
        \begin{equation}
            \mathcal{T}[\tau_I] = \mathcal{T}'[\tau_I]
        \end{equation}
        implies
        \begin{equation}
            \mathtt{F}_{i,\rm st}^{\mathfrak{F}}[\mathcal{T}] = \mathtt{F}_{i,\rm st}^{\mathfrak{F}}[\mathcal{T}'],
        \end{equation}
        for all $i\in I$.

        \begin{proof}
            Let
            \begin{equation}
                \mathcal{X} := \mathcal{T} - \mathcal{T}'.
            \end{equation}
            Since
            \begin{equation}
                \mathcal{T}[\tau_I] = \mathcal{T}'[\tau_I],
            \end{equation}
            we have
            \begin{equation}
                \mathcal{X}[\tau_I] = 0.
            \end{equation}
            Moreover, because both channels are trace-preserving,
            \begin{equation}
                \operatorname{Tr}[\mathcal{X}[\rho]] = 0
            \end{equation}
            for every state $\rho$.

            Let $\overline{\mathcal{G}}$ denote the completely thermalizing channel defined in \eqref{eq:channelTherm}. Since its Choi operator is full rank, there exists $\varepsilon > 0$ such that
            \begin{equation}
                \overline{\mathcal{G}}\pm\varepsilon\mathcal{X}
            \end{equation}
            remain completely positive. Furthermore, they are trace-preserving and Gibbs-preserving because $\mathcal{X}$ is trace-annihilating and satisfies $\mathcal{X}[\tau_I] = 0$. Hence $\overline{\mathcal{G}}\pm\varepsilon\mathcal{X} \in \mathsf{GPTP}(I,O)$, and therefore also $\overline{\mathcal{G}}\pm\varepsilon\mathcal{X} \in \mathsf{GP}(I,O)$.

            Now let $\mathcal{S}\in\mathfrak{F}$. Since $\mathfrak{F}$ is a subclass of $\mathbf{pGPP}$ or $\mathbf{pGPTPP}$, we have
            \begin{equation}
                \Bigl( \mathcal{S}[\overline{\mathcal{G}}\pm \varepsilon\mathcal{X}] \Bigr)[\tau_P] = \tau_F.
            \end{equation}
            Using linearity of $\mathcal{S}$, we obtain
            \begin{equation}
                ( \mathcal{S}[\overline{\mathcal{G}}] )[\tau_P] \pm \varepsilon ( \mathcal{S}[\mathcal{X}] )[\tau_P] = \tau_F.
            \end{equation}
            Subtracting the two equations gives
            \begin{equation}
                ( \mathcal{S}[\mathcal{X}] )[\tau_P] = 0.
            \end{equation}
            Hence
            \begin{equation}
                (\mathcal{S}[\mathcal{T}])[\tau_P] = (\mathcal{S}[\mathcal{T}'])[\tau_P].
            \end{equation}
            Since $\mathcal{S}$ was arbitrary,
            \begin{equation}
                f_i\Bigl((\mathcal{S}[\mathcal{T}])[\tau_P]\Bigr) = f_i\Bigl((\mathcal{S}[\mathcal{T}'])[\tau_P] \Bigr).
            \end{equation}
            Taking the supremum over $\mathcal{S}\in\mathfrak{F}$ proves
            \begin{equation}
                \mathtt{F}_{i,\rm st}^{\mathfrak{F}}[\mathcal{T}] = \mathtt{F}_{i,\rm st}^{\mathfrak{F}}[\mathcal{T}'],
            \end{equation}
            for all $i\in I$.
        \end{proof}
    \end{lemma}
    
        \begin{lemma}\label{lem:rankoneChoi}
            Let $\mathcal{S}'$ be a linear map acting on Choi operators $X$ that is positive, i.e., $\mathcal{S}'[X]\geq 0$ if $X\geq 0$. Suppose $A>0$ and $\mathcal{S}'[A]$ has rank one. Then, for every $B\geq 0$, there exists $c_B\geq 0$ such that
            \begin{equation}
                \mathcal{S}'[B] = c_B \mathcal{S}'[A].
            \end{equation}
            \begin{proof}
                Since $A > 0$, it is invertible. For every $B \geq 0$, define
                \begin{equation}
                    C := A^{-1/2} B A^{-1/2} \geq 0.
                \end{equation}
                Let $\lambda := \|C\|_\infty$. Then $C \leq \lambda \ident$, and therefore
                \begin{equation}
                    B \leq \lambda A.
                \end{equation}
                Applying the positive map $\mathcal{S}'$, we obtain
                \begin{equation}
                    0 \leq \mathcal{S}'[B] \leq \lambda \mathcal{S}'[A].
                \end{equation}
                Since $\mathcal{S}'[A]$ has rank one, any positive operator dominated by $\lambda \mathcal{S}'[A]$ has support contained in the one-dimensional support of $\mathcal{S}'[A]$. Hence there exists $c_B\geq 0$ such that
                \begin{equation}
                    \mathcal{S}'[B] = c_B \mathcal{S}'[A].
                \end{equation}
            \end{proof}
        \end{lemma}
            
       Now, to show that families of Gibbs-response free energies are incomplete, we focus on the case $\mathcal{H}_I \cong \mathcal{H}_O$. For the proof, we first note that there exists a unitary $\hat{U}$ such that
        \begin{equation}\label{eq:proofUnitTau1}
            \mathcal{U}[\tau_I] := \hat{U}\tau_I \hat{U}^\dagger \neq \tau_O.
        \end{equation}
        On the other hand, denoting $\eta := \hat{U}\tau_I \hat{U}^\dagger$, consider a trace-and-replace channel
        \begin{equation}
            \mathcal{T}_\eta[\rho] := \operatorname{Tr}[\rho]\eta.
        \end{equation}
        These channels have the same Gibbs response:
        \begin{equation}
            \mathcal{U}[\tau_I] = \eta = \mathcal{T}_\eta[\tau_I].
        \end{equation}
        Therefore, by Lem.~\ref{lem:FstDependsOnlyOnGibbsResponse}, for every $i\in I$,
        \begin{equation}
            \mathtt{F}_i^{\mathrm{st}}[\mathcal{U}] = \mathtt{F}_i^{\mathrm{st}}[\mathcal{T}_\eta].
        \end{equation}
        Now, suppose that there exists a free transformation $\mathcal{S}\in\mathfrak{F}$ such that
        \begin{equation}
            \mathcal{S}[\mathcal{T}_\eta] = \mathcal{U}.
        \end{equation}
        Let $\mathcal{S}'$ be the linear map induced by $\mathcal{S}$ on Choi operators. Since $\mathcal{S}$ is CP-preserving, $\mathcal{S}'$ is completely positive.
        
        Since $\mathcal{T}_{\eta} \in \mathsf{CP}(I,O)$ is completely positive, its Choi operator $T_{\eta, IO}$ satisfies
        \begin{equation}
            T_{\eta,IO} \geq 0
        \end{equation}
        by \eqref{eq:condCP}. Moreover, the Choi operator of the trace-and-replace channel $\mathcal{T}_\eta$ is
        \begin{equation}
            T_{\eta,IO} = \ident_I \otimes \eta_O.
        \end{equation}
        Since $\eta_O = \hat{U}\tau_I\hat{U}^\dagger$ and $\tau_I$ is full rank, the state $\eta_O$ is full rank. Hence
        \begin{equation}
            T_{\eta,IO}>0.
        \end{equation}
        On the other hand, $\mathcal{S}'[T_{\eta,IO}]$ is the Choi operator
        \begin{equation}
            U_{IO} = (\ident \otimes \hat{U})|\ident\rangle\!\rangle\langle\!\langle\ident| (\ident \otimes \hat{U}^\dagger)
        \end{equation}
        of $\mathcal{U}$, which has rank one.
        
        By Lem.~\ref{lem:rankoneChoi}, for every positive semidefinite Choi operator $B$ there exists $c_B\geq 0$ such that
        \begin{equation}
            \mathcal{S}'[B] = c_B \mathcal{S}'[T_{\eta,IO}].
        \end{equation}
        Since the Choi representation is one-to-one, this implies that for every $\mathcal{M} \in \mathsf{CPTP}(I,O)$ there exists $c_{\mathcal{M}}\geq 0$ such that
        \begin{equation}
            \mathcal{S}[\mathcal{M}] = c_{\mathcal{M}}\mathcal{U}.
        \end{equation}
        Now consider the completely thermalizing channel \eqref{eq:channelTherm}, which is both GP and GPTP. By the previous argument,
        \begin{equation}
            \mathcal{S}[\overline{\mathcal{G}}] = c_{\overline{\mathcal{G}}}\mathcal{U}.
        \end{equation}

        If $\mathfrak{F} = \mathbf{pGPP}$, then $\mathcal{S}[\overline{\mathcal{G}}]$ must be GP. Hence
        \begin{equation}
            c_{\overline{\mathcal{G}}}\mathcal{U}[\tau_I] = c_{\overline{\mathcal{G}}}\eta = \tau_O.
        \end{equation}
        Taking traces gives $c_{\overline{\mathcal{G}}} = 1$, and therefore
        \begin{equation}
            \eta = \tau_O,
        \end{equation}
        contradicting \eqref{eq:proofUnitTau1}.

        If $\mathfrak{F} = \mathbf{pGPTPP}$, then $\mathcal{S}[\overline{\mathcal{G}}]$ must be GPTP. In particular, it must be trace-preserving. Since a scalar multiple of a unitary channel is trace-preserving only if
        \begin{equation}
            c_{\overline{\mathcal{G}}} = 1,
        \end{equation}
        we obtain
        \begin{equation}
            \mathcal{S}[\overline{\mathcal{G}}] = \mathcal{U}.
        \end{equation}
        But $\mathcal{U} \notin \mathsf{GP}(I,O)$, since
        \begin{equation}
            \mathcal{U}[\tau_I] = \eta \neq \tau_O.
        \end{equation}
        Hence $\mathcal{U}$ is not GPTP either, again contradicting the defining property of $\mathfrak{F}$.

        Thus, in both cases, no $\mathcal{S} \in \mathfrak{F}$ can satisfy
        \begin{equation}
            \mathcal{S}[\mathcal{T}_\eta] = \mathcal{U}.
        \end{equation}
        Therefore,
        \begin{equation}
            \mathcal{T}_\eta \not\succeq_{\mathfrak{F}} \mathcal{U}.
        \end{equation}
        Since the family assigns the same values to $\mathcal{T}_\eta$ and $\mathcal{U}$ for all $i\in I$, it cannot be complete.

\subsection{Proof of Prop.~\ref{theo:nogoFchFam}: No-go theorem for \texorpdfstring{$\mathtt{F}_{\rm ch}^{\mathfrak{F}}$}{FchF}- and \texorpdfstring{$\mathtt{F}_{\rm am}^{\mathfrak{F}}$}{FamF}-like monotone families}
\label{app:nogoFchFam}

Here, we show that families of channel free energies based on $\mathtt{F}_{\rm ch}^{\mathfrak{F}}$ and $\mathtt{F}_{\rm am}^{\mathfrak{F}}$ are similarly incomplete as those based on Gibbs-response free energies. Before proceeding, we state the following simple consequence of data processing.

    \begin{lemma}[Unitary invariance from data processing]
    \label{lem:unitaryInvFromDPI}
        Let $\mathfrak D$ be a state divergence satisfying the data-processing inequality. Then, for every unitary $\hat{V}$,
        \begin{equation}
            \mathfrak{D}(\hat{V} \rho \hat{V}^\dagger \| \hat{V} \sigma \hat{V}^\dagger) = \mathfrak{D}(\rho\|\sigma).
        \end{equation}
        \begin{proof}
            Applying data processing inequality \eqref{eq:stateDPI} to the unitary channel $\mathcal{V}[\cdot] = \hat{V}(\cdot)\hat{V}^\dagger$, we obtain
            \begin{equation}
                \mathfrak{D}(\hat{V} \rho \hat{V}^\dagger \| \hat{V} \sigma \hat{V}^\dagger) \leq \mathfrak{D}(\rho\|\sigma).
            \end{equation}
            Applying \eqref{eq:stateDPI} again to the inverse unitary channel $\mathcal{V}^{-1}[\cdot] = \hat{V}^\dagger(\cdot)\hat{V}$ gives
            \begin{equation}
                \mathfrak{D}(\rho\|\sigma) \leq \mathfrak{D} (\hat{V}\rho \hat{V}^\dagger\|\hat{V}\sigma \hat{V}^\dagger).
            \end{equation}
            Combining the two inequalities proves the claim.
        \end{proof}
\end{lemma}

Throughout, we consider the case $\mathcal H_I\simeq\mathcal H_O$ and follow a similar line of argument as in the previous Section. Since $\tau$ is not proportional to the identity, there exists a unitary channel $\mathcal{U}$ such that
\begin{equation}\label{eq:proofUnitTau}
    \mathcal{U}[\tau_I] = \hat{U}\tau_I\hat{U}^\dagger \neq \tau_O.
\end{equation}
Thus $\mathcal{U}\notin\mathsf{GP}(I,O)$ and therefore $\mathcal{U}\notin\mathsf{GPTP}(I,O)$.

By contrast, the identity channel $\mathcal{I}$ is both GP and GPTP. If $\mathfrak{F} = \mathbf{pGPP}$, then every element of
$\mathsf{R}_{\mathfrak{F}}(\mathcal{I})$ is GP. Since the set of GP channels is closed, being defined by finite-dimensional linear constraints,
\begin{equation}
    \overline{\mathsf{R}_{\mathfrak{F}}(\mathcal{I})} \subseteq \mathsf{GP}(I,O).
\end{equation}
Similarly, if $\mathfrak{F}=\mathbf{pGPTPP}$, then
\begin{equation}
    \overline{\mathsf{R}_{\mathfrak{F}}(\mathcal{I})} \subseteq \mathsf{GPTP}(I,O).
\end{equation}
Hence, in both cases,
\begin{equation}
    \mathcal{I}\not\succeq_{\mathfrak{F}}\mathcal{U}.
\end{equation}

We now show that the monotone families assign the same values to $\mathcal{I}$ and $\mathcal{U}$. First consider the $\mathtt{F}_{\rm ch}$-family. For every state $\rho_{RI}$, \begin{equation}
    (\mathcal{I}_R\otimes\overline{\mathcal{G}})[\rho_{RI}] = \rho_R\otimes\tau_O. 
\end{equation} 
Since $\mathcal{I}_R\otimes\mathcal{U}$ is a unitary channel from $\mathcal{L}(\mathcal{H}_R\otimes \mathcal{H}_I)$ to $\mathcal{L}(\mathcal{H}_R\otimes \mathcal{H}_O)$, the change of variables
\begin{equation}
    \omega_{RO} := (\mathcal{I}_R\otimes\mathcal{U})[\rho_{RI}] 
\end{equation}
is bijective on states and preserves the marginal on $\mathcal{L}(\mathcal{H}_R)$, i.e., $\tr_I(\rho_{RI}) = \tr_O(\omega_{RO})$. Hence
\begin{align} 
    \nonumber D_{i,\mathrm{ch}}(\mathcal{U}\|\overline{\mathcal{G}}) &= \sup_{\rho_{RI}} \mathfrak D_i\!\left( (\mathcal{I}_R\otimes\mathcal{U})[\rho_{RI}] \middle\| \rho_R\otimes\tau_O \right) \\ 
    &= \sup_{\omega_{RO}} \mathfrak D_i\!\left( \omega_{RO} \middle\| \omega_R\otimes\tau_O \right) \\ 
    &= D_{i,\mathrm{ch}}(\mathcal{I}\|\overline{\mathcal{G}}). 
\end{align}

Next, let $\mathcal{T}\in\mathsf{CPTP}(I,O)$ be arbitrary and define 
\begin{equation} 
    \omega_{RO} := (\mathcal{I}_R\otimes\mathcal{T})[\rho_{RI}]. 
\end{equation} 
Since $\mathcal{T}$ is trace-preserving, we have $\omega_R=\rho_R$. Therefore 
\begin{align} 
    \nonumber D_{i,\mathrm{ch}}(\mathcal{T}\|\overline{\mathcal{G}}) &= \sup_{\rho_{RI}} \mathfrak D_i\!\left( \omega_{RO} \middle\| \omega_R\otimes\tau_O \right) \\ 
    &\leq \sup_{\omega_{RO}} \mathfrak D_i\!\left( \omega_{RO} \middle\| \omega_R\otimes\tau_O \right) \\ 
    &= D_{i,\mathrm{ch}}(\mathcal{I}\|\overline{\mathcal{G}}).
\end{align}
Applying this to $\mathcal{T} = \mathcal{S}[\mathcal{I}]$ and $\mathcal{T} = \mathcal{S}[\mathcal{U}]$ for arbitrary $\mathcal{S}\in\mathfrak{F}$ gives 
\begin{align} 
    \mathtt{F}_{i,\mathrm{ch}}^{\mathfrak{F}}[\mathcal{I}] &\leq \beta^{-1} D_{i,\mathrm{ch}}(\mathcal{I}\|\overline{\mathcal{G}}), \\ 
    \mathtt{F}_{i,\mathrm{ch}}^{\mathfrak{F}}[\mathcal{U}] &\leq \beta^{-1} D_{i,\mathrm{ch}}(\mathcal{I}\|\overline{\mathcal{G}}).
\end{align} 
Since the identity supermap belongs to $\mathfrak{F}$ by Lem.~\ref{lem:composition}, we have
\begin{align}
    \mathtt{F}_{i,\mathrm{ch}}^{\mathfrak{F}}[\mathcal{I}] &\geq \beta^{-1} D_{i,\mathrm{ch}}(\mathcal{I}\|\overline{\mathcal{G}}),\\
    \mathtt{F}_{i,\mathrm{ch}}^{\mathfrak{F}}[\mathcal{U}] &\geq \beta^{-1} D_{i,\mathrm{ch}}(\mathcal{U}\|\overline{\mathcal{G}}).
\end{align}
Together with the upper bounds above and the equality
\begin{equation}
    D_{i,\mathrm{ch}}(\mathcal{U}\|\overline{\mathcal{G}}) = D_{i,\mathrm{ch}}(\mathcal{I}\|\overline{\mathcal{G}}),
\end{equation}
this gives
\begin{equation}
    \mathtt{F}_{i,\mathrm{ch}}^{\mathfrak{F}}[\mathcal{I}] = \mathtt{F}_{i,\mathrm{ch}}^{\mathfrak{F}}[\mathcal{U}].
\end{equation}
for all $i\in I$.

We now consider the $\mathtt{F}_{\rm am}^\mathfrak{F}$-family. For arbitrary $\mathcal{T}\in\mathsf{CPTP}(I,O)$ and arbitrary states $\rho_{RI},\sigma_{RI}$, define 
\begin{align} 
    \omega_{RO} &:= (\mathcal{I}_R\otimes\mathcal{T})[\rho_{RI}],\\ 
    \zeta_{RO} &:= (\mathcal{I}_R\otimes\mathcal{T})[\sigma_{RI}]. 
\end{align} 
Then $\omega_R=\rho_R$ and $\zeta_R=\sigma_R$. By data processing inequality \eqref{eq:stateDPI}, 
\begin{equation}
    \mathfrak D_i(\omega_{RO}\|\zeta_{RO}) \leq \mathfrak D_i(\rho_{RI}\|\sigma_{RI}). 
\end{equation} 
Therefore,
\begin{align} 
    \nonumber &\mathfrak D_i\!\left( (\mathcal{I}_R\otimes\mathcal{T})[\rho_{RI}] \middle\| (\mathcal{I}_R\otimes\overline{\mathcal{G}})[\sigma_{RI}] \right) - \mathfrak D_i(\rho_{RI}\|\sigma_{RI})\\ 
    \nonumber &\leq \mathfrak D_i\!\left( \omega_{RO} \middle\| \zeta_R\otimes\tau_O \right) - \mathfrak D_i(\omega_{RO}\|\zeta_{RO})\\ 
    &\leq D_{i,\mathrm{am}}(\mathcal{I}\|\overline{\mathcal{G}}). 
\end{align} 
Taking the supremum over $\rho_{RI},\sigma_{RI}$ gives
\begin{equation} 
    D_{i,\mathrm{am}}(\mathcal{T}\|\overline{\mathcal{G}}) \leq D_{i,\mathrm{am}}(\mathcal{I}\|\overline{\mathcal{G}}). 
\end{equation}

It remains to compare $\mathcal{U}$ and $\mathcal{I}$. For the unitary channel, define 
\begin{align} 
    \omega_{RO} &:= (\mathcal{I}_R\otimes\mathcal{U})[\rho_{RI}],\\ 
    \zeta_{RO} &:= (\mathcal{I}_R\otimes\mathcal{U})[\sigma_{RI}]. 
\end{align}
This change of variables is bijective on pairs of states and preserves the $R$-marginals. Moreover, by Lem.~\ref{lem:unitaryInvFromDPI}, 
\begin{equation} 
    \mathfrak D_i(\rho_{RI}\|\sigma_{RI}) = \mathfrak D_i(\omega_{RO}\|\zeta_{RO}). 
\end{equation} 
Thus
\begin{align} 
    \nonumber D_{i,\mathrm{am}}(\mathcal{U}\|\overline{\mathcal{G}}) &= \sup_{\omega_{RO},\zeta_{RO}} \Big[ \mathfrak D_i( \omega_{RO} \| \zeta_R\otimes\tau_O ) - \mathfrak D_i(\omega_{RO}\|\zeta_{RO}) \Big]\\ 
    &= D_{i,\mathrm{am}}(\mathcal{I}\|\overline{\mathcal{G}}).
\end{align} 
Since the identity supermap belongs to $\mathfrak{F}$, by Lem.~\ref{lem:composition} the lower bounds
\begin{align}
    \mathtt{F}_{i,\mathrm{am}}^{\mathfrak{F}}[\mathcal{I}] &\geq \beta^{-1} D_{i,\mathrm{am}}(\mathcal{I}\|\overline{\mathcal{G}}),\\
    \mathtt{F}_{i,\mathrm{am}}^{\mathfrak{F}}[\mathcal{U}] &\geq \beta^{-1} D_{i,\mathrm{am}}(\mathcal{U}\|\overline{\mathcal{G}})
\end{align}
hold. Combining them with the universal upper bound and with
\begin{equation}
    D_{i,\mathrm{am}}(\mathcal{U}\|\overline{\mathcal{G}}) = D_{i,\mathrm{am}}(\mathcal{I}\|\overline{\mathcal{G}})
\end{equation}
yields
\begin{equation}
    \mathtt{F}_{i,\mathrm{am}}^{\mathfrak{F}}[\mathcal{I}] = \mathtt{F}_{i,\mathrm{am}}^{\mathfrak{F}}[\mathcal{U}].
\end{equation}
for all $i\in I$.

Hence both families assign the same values to $\mathcal{I}$ and $\mathcal{U}$, while
\begin{equation}
    \mathcal{I}\not\succeq_{\mathfrak{F}}\mathcal{U}.
\end{equation}
Therefore neither family is complete.

\subsection{Admissibility of \texorpdfstring{$D_{\rm ch}$}{Dch} and \texorpdfstring{$D_{\rm am}$}{Dam}}
\label{app:admissibilityDchDam}

We prove that the channel relative entropy $D_{\rm ch}$ and the amortized channel relative entropy $D_{\rm am}$ are admissible in the sense of Def.~\ref{def:admDiv}. Both quantities are channel divergences, i.e., they satisfy generalized data processing inequality 
\begin{align}
\notag D_{\rm ch}(\mathcal{S}[\mathcal{T}]\| \mathcal{S}[\mathcal{T}']) &\leq D_{\rm ch}(\mathcal{T}\| \mathcal{T}')\\
\text{and} \quad D_{\rm am}(\mathcal{S}[\mathcal{T}]\| \mathcal{S}[\mathcal{T}']) &\leq D_{\rm am}(\mathcal{T}\| \mathcal{T}')
\end{align}
under general superchannels $\mathcal{S}$. Thus it remains to verify nonnegativity, faithfulness, and topological detectability.

We use the following standard facts about the quantum relative entropy. For states $\rho,\sigma$,
\begin{equation}
    D(\rho\|\sigma) \geq 0,
\end{equation}
with equality if and only if $\rho = \sigma$. Moreover, Pinsker's inequality implies that
\begin{equation}
    D(\rho_n\|\sigma)\to 0
    \quad\Longrightarrow\quad
    \rho_n\to\sigma
\end{equation}
in trace norm.

\begin{prop}[Admissibility of $D_{\rm ch}$]
\label{prop:DchAdmissible}
    The channel relative entropy
    \begin{equation}
        D_{\rm ch}(\mathcal{T}\|\mathcal{T}') = \sup_{\rho_{RI}} D\!\left((\mathcal{I}_R\otimes\mathcal{T})[\rho_{RI}] \middle\| (\mathcal{I}_R\otimes\mathcal{T}')[\rho_{RI}] \right)
    \end{equation}
    is admissible.

    \begin{proof}
        Nonnegativity follows immediately from nonnegativity of the state relative entropy:
        \begin{equation}
            D_{\rm ch}(\mathcal{T}\|\mathcal{T}')\geq 0.
        \end{equation}

        We next prove faithfulness. If $\mathcal{T} = \mathcal{T}'$, then clearly $D_{\rm ch}(\mathcal{T}\|\mathcal{T}') = 0$. Conversely, suppose
        \begin{equation}
            D_{\rm ch}(\mathcal{T}\|\mathcal{T}') = 0.
        \end{equation}
        Since every term in the supremum is nonnegative, for every input state $\rho_{RI}$ we must have
        \begin{equation}
            D\!\left((\mathcal{I}_R\otimes\mathcal{T})[\rho_{RI}] \middle\| (\mathcal{I}_R\otimes\mathcal{T}')[\rho_{RI}]\right) = 0.
        \end{equation}
        By faithfulness of the state relative entropy,
        \begin{equation}
            (\mathcal{I}_R\otimes\mathcal{T})[\rho_{RI}] = (\mathcal{I}_R\otimes\mathcal{T}')[\rho_{RI}]
        \end{equation}
        for every $\rho_{RI}$. Hence the two extended channels coincide on all states, and therefore
        \begin{equation}
            \mathcal{T} = \mathcal{T}'.
        \end{equation}

        It remains to prove topological detectability. Suppose
        \begin{equation}
            D_{\rm ch}(\mathcal{T}_n\|\mathcal{T}')\to 0.
        \end{equation}
        Let $\mathcal{H}_R \simeq \mathcal{H}_I$ and let $\Phi_{RI} = \frac{1}{d_I}|\ident\rangle\!\rangle\langle\!\langle\ident|_{RI}$ be a normalized maximally entangled state. Since the supremum defining $D_{\rm ch}$ is bounded below by its value on this particular input,
        \begin{equation}
            D\!\left((\mathcal{I}_R\otimes\mathcal{T}_n)[\Phi_{RI}] \middle\| (\mathcal{I}_R\otimes\mathcal{T}')[\Phi_{RI}]\right) \leq D_{\rm ch}(\mathcal{T}_n\|\mathcal{T}').
        \end{equation}
        Hence
        \begin{equation}
            D\!\left((\mathcal{I}_R\otimes\mathcal{T}_n)[\Phi_{RI}] \middle\| (\mathcal{I}_R\otimes\mathcal{T}')[\Phi_{RI}] \right) \to 0.
        \end{equation}
        By Pinsker's inequality,
        \begin{equation}
            \left\| (\mathcal{I}_R\otimes\mathcal{T}_n)[\Phi_{RI}] - (\mathcal{I}_R\otimes\mathcal{T}')[\Phi_{RI}] \right\|_1 \to 0.
        \end{equation}
        The states $(\mathcal{I}_R\otimes\mathcal{T}_n)[\Phi_{RI}]$ and $(\mathcal{I}_R\otimes\mathcal{T}')[\Phi_{RI}]$ are the normalized Choi operators of $\mathcal{T}_n$ and $\mathcal{T}'$, namely, $\frac{1}{d_I}T_{RO}$ and $\frac{1}{d_I}T'_{RO}$. Therefore the corresponding Choi operators converge in trace norm. Since the space of linear maps is finite-dimensional, all norms on it are equivalent. Consequently,
        \begin{equation}
            \mathcal{T}_n \to \mathcal{T}'
        \end{equation}
        in any equivalent channel topology, e.g., the Choi norm or the diamond norm. This proves topological detectability, and hence $D_{\rm ch}$ is admissible.
    \end{proof}
\end{prop}

\begin{prop}[Admissibility of $D_{\rm am}$]
\label{prop:DamAdmissible}
    The amortized channel relative entropy
    \begin{align}
        \nonumber D_{\rm am}(\mathcal{T}\|\mathcal{T}') &:= \sup_{\rho_{RI},\sigma_{RI}} \Bigl[ D\!\left((\mathcal{I}_R\otimes\mathcal{T})[\rho_{RI}] \middle\| (\mathcal{I}_R\otimes\mathcal{T}')[\sigma_{RI}] \right) \\
        &\qquad - D(\rho_{RI}\|\sigma_{RI}) \Bigr]
    \end{align}
    is admissible.

    \begin{proof}
        Nonnegativity follows immediately from nonnegativity of $D_{\rm ch}$, which is dominated by $D_{\rm am}$ by \eqref{eq:domChAm},
        \begin{equation}
            D_{\rm am}(\mathcal{T}\|\mathcal{T}') \geq D_{\rm ch}(\mathcal{T}\|\mathcal{T}') \geq 0.
        \end{equation}

        We now prove faithfulness. If $\mathcal{T} = \mathcal{T}'$, then by data processing inequality \eqref{eq:DPI},
        \begin{equation}
            D\!\left((\mathcal{I}_R \otimes \mathcal{T})[\rho_{RI}] \middle\| (\mathcal{I}_R \otimes \mathcal{T})[\sigma_{RI}] \right) - D(\rho_{RI}\|\sigma_{RI}) \leq 0,
        \end{equation}
        for all $\rho_{RI},\sigma_{RI}$. Taking $\rho_{RI} = \sigma_{RI}$ gives value $0$, and hence $D_{\rm am}(\mathcal{T}\|\mathcal{T}) = 0$. Conversely, if
        \begin{equation}
            D_{\rm am}(\mathcal{T}\|\mathcal{T}') = 0,
        \end{equation}
        then
        \begin{equation}
            0 \leq D_{\rm ch}(\mathcal{T}\|\mathcal{T}') \leq D_{\rm am}(\mathcal{T}\|\mathcal{T}') = 0.
        \end{equation}
        Thus
        \begin{equation}
            D_{\rm ch}(\mathcal{T}\|\mathcal{T}') = 0.
        \end{equation}
        By Proposition~\ref{prop:DchAdmissible},
        \begin{equation}
            \mathcal{T}=\mathcal{T}'.
        \end{equation}
        Finally, suppose
        \begin{equation}
            D_{\rm am}(\mathcal{T}_n\|\mathcal{T}')\to 0.
        \end{equation}
        Using
        \begin{equation}
            0 \leq D_{\rm ch}(\mathcal{T}_n\|\mathcal{T}') \leq D_{\rm am}(\mathcal{T}_n\|\mathcal{T}'),
        \end{equation}
        we obtain
        \begin{equation}
            D_{\rm ch}(\mathcal{T}_n\|\mathcal{T}')\to 0.
        \end{equation}
        By Proposition~\ref{prop:DchAdmissible}, this implies
        \begin{equation}
            \mathcal{T}_n\to\mathcal{T}'
        \end{equation}
        in any equivalent finite-dimensional topology on channels. Therefore $D_{\rm am}$ is admissible.
    \end{proof}
\end{prop}

\subsection{Proof of Prop.~\ref{theo:completeFamily}: Completeness of the \texorpdfstring{$\Delta_{\mathbb{D}, \mathcal{R}}$}{DeltaD,R}-family of monotones}
\label{app:sub:theoCompl}
        Assume first that $\mathcal{T} \succeq_{\mathfrak{F}} \mathcal{T}'$. By Corollary \ref{lem:inclusionR},
        \begin{equation}
            \overline{\mathsf{R}_{\mathfrak{F}}(\mathcal{T})} \supseteq \overline{\mathsf{R}_{\mathfrak{F}}(\mathcal{T}')}
        \end{equation}
        Therefore, for every $\mathcal{R}\in\mathsf{CPTP}(I,O)$, taking the infimum of $\mathbb{D}(\bullet\|\mathcal{R})$ over the larger set can only decrease the value:
        \begin{align}
            \nonumber \Delta_{\mathbb{D},\mathcal{R}}(\mathcal{T}) &= \inf_{\mathcal{X}\in\overline{\mathsf{R}_{\mathfrak{F}}(\mathcal{T})}} \mathbb{D}(\mathcal{X}\|\mathcal{R}) \\
            &\leq \inf_{\mathcal{X}\in\overline{\mathsf{R}_{\mathfrak{F}}(\mathcal{T}')}} \mathbb{D}(\mathcal{X}\|\mathcal{R}) \\
            &= \Delta_{\mathbb{D},\mathcal{R}}(\mathcal{T}').
        \end{align}
        This proves \eqref{eq:complMonotone}. Conversely, suppose that 
        \begin{equation}
            \Delta_{\mathbb{D},\mathcal{R}}(\mathcal{T}) \leq \Delta_{\mathbb{D},\mathcal{R}}(\mathcal{T}')
        \end{equation}
        for every $\mathcal{R}\in\mathsf{CPTP}(I,O)$. Choosing $\mathcal{R} = \mathcal{T}'$, we obtain
        \begin{equation}
            \Delta_{\mathbb{D},\mathcal{T}'}(\mathcal{T}) \leq \Delta_{\mathbb{D},\mathcal{T}'}(\mathcal{T}').
        \end{equation}
        
        Since $\mathcal{T}' \in \overline{\mathsf{R}_{\mathfrak{F}}(\mathcal{T}')}$, it follows from faithfulness and nonnegativity that
        \begin{equation}
            \Delta_{\mathbb{D},\mathcal{T}'}(\mathcal{T}') = 0.
        \end{equation}
        Hence
        \begin{equation}\label{eq:deltaZero}
            \Delta_{\mathbb{D},\mathcal{T}'}(\mathcal{T}) = 0.
        \end{equation}
By definition of the infimum, there exists a sequence $\{ \mathcal{X}_k \}_{k\in\mathbb{N}} \subseteq \overline{\mathsf{R}_{\mathfrak{F}}(\mathcal{T})}$ such that
        \begin{equation}
            \mathbb{D}(\mathcal{X}_k\|\mathcal{T}') \to 0.
        \end{equation}
        Since $\mathbb{D}$ is admissible, topological detectability implies
        \begin{equation}
            \mathcal{X}_k \to \mathcal{T}'.
        \end{equation}
        Because $\overline{\mathsf{R}_{\mathfrak{F}}(\mathcal{T})}$ is closed, we conclude that
        \begin{equation}
            \mathcal{T}' \in \overline{\mathsf{R}_{\mathfrak{F}}(\mathcal{T})}.
        \end{equation}
        Therefore
        \begin{equation}
            \mathcal{T} \succeq_{\mathfrak{F}} \mathcal{T}'.
        \end{equation}

\subsection{Countable complete subfamilies for \texorpdfstring{$D_{\rm ch}$}{Dch}}
\label{app:countableComplete}
We prove that the complete family
\begin{equation}
    \{\Delta_{D_{\rm ch},\mathcal{R}}\}_{\mathcal{R}\in\mathsf{CPTP}(I,O)}
\end{equation}
admits a countable complete subfamily. We begin with a continuity property of the channel relative entropy under full-rank regularization.

\begin{lemma}[Continuity after full-rank regularization]
\label{lem:fullRankContinuity}
    Let $\eta_O > 0$ be a full-rank state and let $\mathcal{T}_\eta$ be the trace-and-replace channel
    \begin{equation}
        \mathcal{T}_\eta[\rho_I] = \operatorname{Tr}[\rho_I]\eta_O .
    \end{equation}
    Fix $\varepsilon > 0$ and define
    \begin{equation}
        \mathcal{R}_\varepsilon := (1-\varepsilon)\mathcal{R} + \varepsilon\mathcal{T}_\eta .
    \end{equation}
    Then, for fixed $\mathcal{T}'$, the map
    \begin{equation}
        \mathcal{R} \mapsto D_{\rm ch}(\mathcal{T}'\|\mathcal{R}_\varepsilon)
    \end{equation}
    is continuous on $\mathsf{CPTP}(I,O)$.

    \begin{proof}
        For every input state $\rho_{RI}$,
        \begin{align}
            \nonumber (\mathcal{I}_R\otimes\mathcal{R}_\varepsilon)[\rho_{RI}] &= (1-\varepsilon) (\mathcal{I}_R\otimes\mathcal{R})[\rho_{RI}] + \varepsilon\,\rho_R\otimes\eta_O \\
            &\geq \varepsilon\,\rho_R\otimes\eta_O .
        \end{align}
        Hence the second argument of the relative entropy has a uniformly full-rank component on the support compatible with the $R$-marginal. In finite dimensions this removes possible support singularities.

        Therefore the function
        \begin{equation}
            (\rho_{RI},\mathcal{R}) \mapsto D\!\left((\mathcal{I}_R\otimes\mathcal{T}')[\rho_{RI}] \middle\| (\mathcal{I}_R\otimes\mathcal{R}_\varepsilon)[\rho_{RI}] \right)
        \end{equation}
        is continuous. By the maximum theorem,
        \begin{equation}
            \mathcal{R} \mapsto \sup_{\rho_{RI}} D\!\left((\mathcal{I}_R\otimes\mathcal{T}')[\rho_{RI}] \middle\| (\mathcal{I}_R\otimes\mathcal{R}_\varepsilon)[\rho_{RI}]\right)
        \end{equation}
        is continuous. This is precisely $D_{\rm ch}(\mathcal{T}'\|\mathcal{R}_\varepsilon)$.
    \end{proof}
\end{lemma}

Since $\mathsf{CPTP}(I,O)$ is compact and separable in finite dimensions, let
\begin{equation}
    \{\mathcal Q_n\}_{n\in\mathbb{N}} \subset \mathsf{CPTP}(I,O)
\end{equation}
be a countable dense subset.

Fix a full-rank trace-and-replace channel $\mathcal{T}_\eta$ as in Lem.~\ref{lem:fullRankContinuity}. For $n\in\mathbb{N}$ and $k\geq2$, define
\begin{equation}
    \mathcal{R}_{n,k} := \left(1-\frac{1}{k}\right)\mathcal Q_n + \frac{1}{k}\mathcal{T}_\eta .
\end{equation}
The set $\{\mathcal{R}_{n,k}\}_{n\in\mathbb{N},\ k\geq 2}$ is countable. We claim that
\begin{equation}
    \{\Delta_{D_{\rm ch},\mathcal{R}_{n,k}}\}_{n\in\mathbb{N},\ k\geq 2}
\end{equation}
is complete. The necessary direction follows immediately from Thm.~\ref{theo:completeFamily}. Indeed, if
\begin{equation}
    \mathcal{T}\succeq_{\mathfrak{F}}\mathcal{T}',
\end{equation}
then
\begin{equation}
    \Delta_{D_{\rm ch},\mathcal{R}_{n,k}}(\mathcal{T}) \leq \Delta_{D_{\rm ch},\mathcal{R}_{n,k}}(\mathcal{T}')
\end{equation}
for all $n,k$. Conversely, assume that
\begin{equation}
    \Delta_{D_{\rm ch},\mathcal{R}_{n,k}}(\mathcal{T}) \leq \Delta_{D_{\rm ch},\mathcal{R}_{n,k}}(\mathcal{T}')
\end{equation}
for all $n,k$. We prove that then
\begin{equation}
    \mathcal{T}\succeq_{\mathfrak{F}}\mathcal{T}'
\end{equation}
holds, or, equivalently,
\begin{equation}
    \mathcal{T}' \in \overline{\mathsf{R}_{\mathfrak{F}}(\mathcal{T})}.
\end{equation}
For each $k\geq2$, define
\begin{equation}
    \mathcal{R}_k^\star := \left(1-\frac{1}{k}\right)\mathcal{T}' + \frac{1}{k}\mathcal{T}_\eta .
\end{equation}
Then
\begin{equation}
    \mathcal{R}_k^\star\to\mathcal{T}'.
\end{equation}
Moreover, for every input state $\rho_{RI}$,
\begin{equation}
    (\mathcal{I}_R\otimes\mathcal{R}_k^\star)[\rho_{RI}] \geq \left(1-\frac{1}{k}\right) (\mathcal{I}_R\otimes\mathcal{T}')[\rho_{RI}].
\end{equation}
Hence, by operator monotonicity of the logarithm,
\begin{equation}
    D_{\rm ch}(\mathcal{T}'\|\mathcal{R}_k^\star) \leq -\log\left(1-\frac{1}{k}\right) \longrightarrow 0.
\end{equation}
Fix $k\geq2$. By density of $\{\mathcal Q_n\}_{n\in\mathbb{N}}$, choose $n_k$ sufficiently large so that
\begin{equation}
    \|\mathcal Q_{n_k} - \mathcal{T}'\| < \frac{1}{k}.
\end{equation}
Moreover, by Lem.~\ref{lem:fullRankContinuity}, applied with
$\varepsilon=1/k$, we may choose $n_k$ such that also
\begin{equation}
    D_{\rm ch}(\mathcal{T}'\|\mathcal{R}_{n_k,k}) \leq D_{\rm ch}(\mathcal{T}'\|\mathcal{R}_k^\star) + \frac{1}{k} .
\end{equation}
Set
\begin{equation}
    \mathcal{R}_k := \mathcal{R}_{n_k,k}.
\end{equation}
Then
\begin{equation}
    \mathcal{R}_k\to\mathcal{T}'
\end{equation}
and
\begin{equation}
    D_{\rm ch}(\mathcal{T}'\|\mathcal{R}_k)\to0.
\end{equation}
Since
\begin{equation}
    \mathcal{T}' \in \overline{\mathsf{R}_{\mathfrak{F}}(\mathcal{T}')},
\end{equation}
we obtain
\begin{equation}
    \Delta_{D_{\rm ch},\mathcal{R}_k}(\mathcal{T}') \leq D_{\rm ch}(\mathcal{T}'\|\mathcal{R}_k) \longrightarrow 0.
\end{equation}
By assumption,
\begin{equation}
    \Delta_{D_{\rm ch},\mathcal{R}_k}(\mathcal{T}) \leq \Delta_{D_{\rm ch},\mathcal{R}_k}(\mathcal{T}'),
\end{equation}
and therefore
\begin{equation}
    \Delta_{D_{\rm ch},\mathcal{R}_k}(\mathcal{T}) \longrightarrow 0.
\end{equation}
By definition of the infimum, for each $k$ there exists
\begin{equation}
    \mathcal{X}_k \in \overline{\mathsf{R}_{\mathfrak{F}}(\mathcal{T})}
\end{equation}
such that
\begin{equation}
    D_{\rm ch}(\mathcal{X}_k\|\mathcal{R}_k) \leq \Delta_{D_{\rm ch},\mathcal{R}_k}(\mathcal{T}) + \frac{1}{k} .
\end{equation}
Hence
\begin{equation}
    D_{\rm ch}(\mathcal{X}_k\|\mathcal{R}_k) \to 0.
\end{equation}
Let $\Phi_{RI}$ be a normalized maximally entangled state. Evaluating the defining supremum of $D_{\rm ch}$ on $\Phi_{RI}$ and applying Pinsker's inequality yields
\begin{equation}
    \left\| (\mathcal{I}_R\otimes\mathcal{X}_k)[\Phi_{RI}] - (\mathcal{I}_R\otimes\mathcal{R}_k)[\Phi_{RI}] \right\|_1 \to 0.
\end{equation}
Since $\Phi_{RI}$ is normalized,
\begin{align}
    (\mathcal{I}_R\otimes\mathcal{X}_k)[\Phi_{RI}] &=
    \frac{1}{d_I} X_{k,RO}, \\
    (\mathcal{I}_R\otimes\mathcal{R}_k)[\Phi_{RI}] &=
    \frac{1}{d_I} R_{k,RO},
\end{align}
where $X_{k,RO}$ and $R_{k,RO}$ are Choi operators of $\mathcal{X}_k$ and $\mathcal{R}_k$, respectively. Hence,
\begin{equation}
    \|X_{k,RO} - R_{k,RO} \|_1 \to 0.
\end{equation}
Because
\begin{equation}
    \mathcal{R}_k\to\mathcal{T}',
\end{equation}
we also have
\begin{equation}
    \|R_{k,RO} - T'_{RO}\|_1 \to 0.
\end{equation}
Therefore,
\begin{equation}
    \|X_{k,RO} - T'_{RO} \|_1 \to 0,
\end{equation}
and hence
\begin{equation}
    \mathcal{X}_k\to\mathcal{T}'.
\end{equation}

Since
\begin{equation}
    \overline{\mathsf{R}_{\mathfrak{F}}(\mathcal{T})}
\end{equation}
is closed and each $\mathcal{X}_k$ belongs to it, the limit $\mathcal{T}'$ also belongs to it. Thus
\begin{equation}
    \mathcal{T}' \in \overline{\mathsf{R}_{\mathfrak{F}}(\mathcal{T})},
\end{equation}
i.e.
\begin{equation}
    \mathcal{T}\succeq_{\mathfrak{F}}\mathcal{T}'.
\end{equation}
Therefore the countable family $\{\Delta_{D_{\rm ch},\mathcal{R}_{n,k}}\}_{n\in\mathbb{N},\ k\geq 2}$ is complete.

\end{document}